\documentclass[11pt]{amsart}

\usepackage{graphicx}
\usepackage{framed}
\IfFileExists{ulem.sty}{\usepackage[normalem]{ulem}}{}
\usepackage{pgf,tikz,pgfplots}
\pgfplotsset{compat=1.18}
\usepackage{float}
\usepackage{mathrsfs}
\usepackage{mathtools}
\usepackage{mathpple}
\usepackage[all,cmtip]{xy}
\usepackage{amsmath}
\usepackage{appendix}
\usepackage{tikz-cd}
\IfFileExists{mathdots.sty}{\usepackage{mathdots}}{}
\IfFileExists{yhmath.sty}{\usepackage{yhmath}}{}
\usepackage{cancel}
\usepackage{color}
\definecolor{deepgreen}{rgb}{0.0,0.35,0.0}
\IfFileExists{siunitx.sty}{\usepackage{siunitx}}{}
\usepackage{array}
\usepackage{multirow}
\usepackage{amssymb}
\usepackage{gensymb}
\usepackage{tabularx}
\usepackage{booktabs}
\usepackage{makecell}
\usetikzlibrary{fadings}
\usetikzlibrary{patterns}
\usetikzlibrary{shapes}
\usetikzlibrary{arrows.meta,calc,decorations.markings}

\usetikzlibrary{matrix,arrows,decorations.pathmorphing}
\usepackage{enumerate}
\usepackage{enumitem}
\IfFileExists{extarrows.sty}{\usepackage{extarrows}}{}
\usepackage{appendix}

\usepackage[bookmarks=true, bookmarksopen=true,hypertexnames=false]{hyperref}

\DeclareMathOperator{\Id}{Id}

\theoremstyle{plain}
\newtheorem{thm}{Theorem}[section]
\newtheorem{thm-defn}[thm]{Theorem/Definition}
\newtheorem{lem}[thm]{Lemma}
\newtheorem{lem-defn}[thm]{Lemma/Definition}
\newtheorem{prop}[thm]{Proposition}
\newtheorem{cor}[thm]{Corollary}
\newtheorem{conj}[thm]{Conjecture}

\newtheorem{defn}[thm]{Definition}

\newtheorem{eg}[thm]{Example}
\newtheorem{rmk}[thm]{Remark}

\theoremstyle{definition}
\newtheorem{assumption}{Assumption}

\hypersetup{colorlinks=true,linkcolor=blue,citecolor=blue,urlcolor=blue}

\newcommand{\R}{\mathbb{R}}
\newcommand{\C}{\mathbb{C}}
\newcommand{\Z}{\mathbb Z}
\newcommand{\dd}{\mathrm d}

\newcommand{\J}{\mathcal J}

\newcommand{\ip}[2]{\left\langle #1,#2\right\rangle}
\newcommand{\norm}[1]{\left\lVert #1\right\rVert}

\newcommand{\PP}{\mathcal{P}}
\newcommand{\HH}{\mathbb{H}}

\newcommand{\var}{\text{var}}
\newcommand{\cont}{\text{cont}}
\newcommand{\rad}{\text{rad}}

\DeclareMathOperator{\re}{Re}
\DeclareMathOperator{\im}{Im}
\DeclareMathOperator{\diag}{diag}
\DeclareMathOperator{\Span}{Span}
\DeclareMathOperator{\coker}{Coker}
\DeclareMathOperator{\ind}{ind}

\title[Heat Kernel and Resurgence]{Heat Kernel and Resurgence}

\author[Si Li]{Si Li}
\author[Yong Li]{Yong Li}
\author[Xinxing Tang]{Xinxing Tang}

\thanks{S.~Li: Yau Mathematical Sciences Center, Tsinghua University, Beijing, China.
Email: \texttt{sili@mail.tsinghua.edu.cn}.}

\thanks{Y.~Li: Beijing Institute of Mathematical Sciences and Applications,
Beijing, China.
Email: \texttt{liyong@bimsa.cn}.}

\thanks{X.~Tang: Beijing Institute of Mathematical Sciences and Applications,
Beijing, China.
Email: \texttt{tangxinxing@bimsa.cn}.}

\date{}


\newcommand{\eps}{\varepsilon}

\newcommand{\cX}{\mathcal X}
\newcommand{\cY}{\mathcal Y}

\newcommand{\loc}{\mathrm{loc}}

\DeclareMathOperator{\dist}{dist}

\newcommand{\Krot}{K_{\mathrm{rot}}}
\begin{document}

\begin{abstract}
We study the resurgent structure of short--time heat kernel asymptotics from the viewpoint of Picard--Lefschetz theory. For a real analytic Riemannian manifold, we show the heat kernel admits a 1--
Gevrey small--time expansion whose Borel transform detects complex--geometric data beyond the real geodesic sector. We formulate an infinite--dimensional Picard--Lefschetz problem of Morse--Floer type for the holomorphic energy functional on the complexified path space, and propose a heat--kernel analogue of the
Picard--Lefschetz/Alien correspondence. In this framework, pointed alien operators acting on the asymptotic expansion associated with the real geodesic are predicted to produce the formal heat--kernel sectors
associated with other holomorphic geodesics, with coefficients given by signed counts of connecting trajectories of the Morse flow. We perform a confirming test of this proposal on the hyperbolic plane $\mathbb{H}_{\R}^2$.
\end{abstract}
\maketitle
\tableofcontents

\section{Introduction}

The purpose of this paper is to explain how the resurgent structure of a heat-kernel expansion is reflected by Picard--Lefschetz theory on an infinite-dimensional complexified path space. This is motivated by the Picard–Lefschetz/Alien correspondence in the finite-dimensional case. We start by reviewing this correspondence following \cite{LiLiTang2026}  and refer to it for further details and references therein. 

\subsection{Lefschetz thimble and Borel singularity}\label{sec:intro-Borel}
Consider a finite-dimensional exponential integral
\begin{equation*}
I_\Gamma(\hbar)=\int_\Gamma e^{-S(z)/\hbar}\,\mu
\end{equation*}
where $S(z)$ is a holomorphic function on a complex manifold, $\mu$ is a holomorphic form, and $\Gamma$ is a contour (or cycle) in a suitable relative homology class $[\Gamma]$. The asymptotic behavior of $I_\Gamma(\hbar)$ for small $\hbar$ is controlled by the critical points of the phase function
$$\dd S=0.$$

Assume $S$ has only non-degenerate isolated critical points and the corresponding critical values are pairwise distinct. For each critical point $p$, there corresponds an integration cycle $\mathcal J_p$, called the Lefschetz thimble, which can be constructed by the downward gradient flow of $\re(S/\hbar)$ (with a chosen Hermitian metric on $X$). One can decompose the cycle in homology class
$$
[\Gamma]=\sum_p n_p [\mathcal J_p]
$$
and obtain a representation of the integral by
$$
I_\Gamma(\hbar)=\sum_{p}n_p \int_{\mathcal J_p} e^{-S(z)/\hbar}\,\mu.
$$
Near each critical point $p$, one obtains a formal asymptotic expansion of the form 
$$
I_p(\hbar) :=\int_{\mathcal J_p} e^{-S(z)/\hbar}\,\mu \sim e^{-S(p)/\hbar}\tilde \phi_p(\hbar)= e^{-S(p)/\hbar}\hbar^{n/2}\left(c_{p,0}+c_{p,1}\hbar+c_{p,2}\hbar^2+\cdots\right)\quad \text{as}\ \hbar\to 0 
$$
uniformly on closed subsectors of a sufficiently small sector in the $\hbar$-plane. This is the standard local steepest-descent (or stationary-phase) expansion of a Laplace-type integral near a non-degenerate critical point (see e.g. \cite{BleisteinHandelsman,PhamSaddlepoint}); equivalently, it may be viewed as the local WKB expansion associated with the saddle $p$.
This expansion is typically divergent and only asymptotic in nature. They all together capture the asymptotic behavior of $I_\Gamma(\hbar)$ for small $\hbar$. 

The theory of resurgence developed by \'Ecalle \cite{Ecalle1981I,Ecalle1981II,Ecalle1985III} provides a general framework to study global analytic natures of divergent asymptotic expansions. Under appropriate growth condition, the Borel transform $\hat \phi_p(\xi)$ of the $\hbar$-asymptotic series $\tilde \phi_p(\hbar)$ from the critical point $p$ will be an integrable germ locally near $\xi=0$ and admits analytic continuation beyond it. For $\arg \hbar = \theta$,  one has a resummed function with the expected asymptotic expansion
\[
I_p(\hbar)
:=
\int_{\mathcal J_p} e^{-S(z)/\hbar}\,\mu
=
e^{-S(p)/\hbar}\,
\mathcal L_\theta\widehat\phi_p(\hbar) \sim e^{-S(p)/\hbar}
\widetilde\phi_p(\hbar),
\qquad \hbar\to0.
\]
Here $\mathcal L_\theta$ is the Laplace transform along the ray  $e^{i\theta}\R_{\geq 0}$ (see Definition \ref{defn:Laplace}). 
The Borel transform $\hat \phi_p(\xi)$ may exhibit singularities in the $\xi$-plane. Such Borel singularities lead to possible ambiguities in the resummation. One has to choose a direction $e^{i\theta}$ so that the ray in the Laplace transform avoids Borel singularities. Two different directions could give rise to different analytic objects when one ray has to pass certain singularities to reach the other. This is precisely the Stokes phenomenon. 

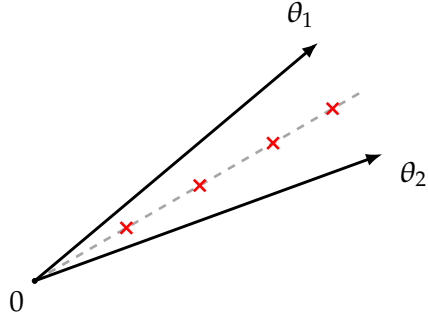
\begin{figure}[H]
\centering
\begin{tikzpicture}[scale=0.7,>=Latex]

\tikzset{
  stokesray/.style={
    gray!70,
    dashed,
    line width=1pt
  },
  lateralsum/.style={
    black,
    line width=1.1pt,
    -{Latex[length=2.2mm]}
  },
  stokesarcminus/.style={
    black,
    line width=1.05pt,
    -{Latex[length=2.2mm]}
  },
  stokesarcplus/.style={
    black,
    line width=1.05pt,
    -{Latex[length=2.2mm]}
  }
}
\coordinate (O) at (0,0);
\draw[stokesray] (O) -- (30:7.2);
\draw[lateralsum] (O) -- (20:7.0);
\draw[lateralsum] (O) -- (40:7.0);
\filldraw[black] (O) circle (1.3pt);
\node[below left] at (O) {$0$};
\foreach \r in {2.0,3.6,5.2,6.5}{
  \draw[red, line width=1pt]
    ($(30:\r)+(-0.11,-0.11)$) -- ($(30:\r)+(0.11,0.11)$);
  \draw[red, line width=1pt]
    ($(30:\r)+(-0.11,0.11)$) -- ($(30:\r)+(0.11,-0.11)$);
}
\node at (20:7.15) [below right] {$\theta_2$};
\node at (40:7.15) [above left] {$\theta_1$};
\end{tikzpicture}
\caption{$\xi$-plane}
\label{fig:stokes-automorphism-zeta-plane}
\end{figure}

For the exponential integral, the ambiguity of the ray for Laplace transform is related to the ambiguity of the choice of Lefschetz thimble. Take $\hbar$ to be valued in $\C^\times$ and write
$$
\hbar=|\hbar|e^{i\theta}.
$$
Then the flow of $\re(S/\hbar)$ is the same as that of $\re(e^{-i\theta}S)$. The magic of holomorphic Morse function $S$ tells that along the  the downward gradient flow of $\re(S/\hbar)$ :
\begin{align*}
&(1)~ \re(e^{-i\theta} S)\ \text{is {decreasing}}, &\qquad (2)~ \im (e^{-i\theta} S) \ \text{is {constant}}.
\end{align*}
The corresponding Lefschetz thimble is represented by (here we specify its $\theta$-dependence)
$$
\mathcal J_p^\theta=\{z\in X\mid \text{$z$ flows to $p$ at time $t=+\infty$}\}.
$$

With our assumptions on $S$, at generic $\theta$, the values of $\im(e^{-i\theta}S)$ at critical points are all distinct, so all such constructed thimbles are disjoint. If we sit at some $\theta_*$ such that
$$
\im(e^{-i\theta_*}S(p))=\im(e^{-i\theta_*}S(q))
$$
for two critical points $p,q$, then one could have flow trajectories that connect one critical point to another. Such $\theta_*$ is called a Stokes phase and the ray in the direction $e^{i\theta_*}$ is called a Stokes ray. 

Assume we have a trajectory that flows from $q$ to $p$ at a Stokes phase $\theta_*$ (this requires $\re(e^{-i\theta_*}S(p))<\re(e^{-i\theta_*}S(q))$ by the downward gradient flow). Then the Borel transform $\hat \phi_p(\xi)$ will see a singularity in the direction $e^{i\theta_*}$ at 
$$
\omega=S(q)-S(p).
$$
Thus Borel singularities correspond to relative critical values with respect to a reference critical point.

In  \cite{KS-2022,KS-2024}, Kontsevich and Soibelman investigated analyticity and resurgence phenomena in the context of Lefschetz thimbles and their wall-crossing behavior, and develop a Floer-theoretic framework for exponential integrals. Our study of the complexified path space in this work follows the perspective of holomorphic Morse–Floer theory and connects it to the alien calculus, as we explain next.

\subsection{Connecting trajectory and Alien operator}\label{sec:intro-Alien} Borel singularities lead to rich algebraic structures on the global analytic nature of divergence series. They are captured by the so-called Alien operators. 

Let $\widetilde\phi(\hbar)$ 
be a $1$--Gevrey formal series (see Section \ref{sec:Gevrey-Borel}). Let
\[
\widehat\phi(\xi)=\mathcal B\widetilde\phi(\xi)
\]
be its Borel transform. Suppose that $\widehat\phi$ admits analytic
continuation from the origin to a singular point $\omega\in\C^*$, and the germ is integrable at $\omega$. Then we can obtain a formal series in $\hbar$, denoted by  $$
\Delta_{\omega}^{+}\widetilde\phi
$$
by extracting the local
variation of the Borel transform at $\omega$, recentering it at the origin,
and applying inverse Borel transform. See Section \ref{sec:heat-alien} for the precise definition and references.

$\Delta_{\omega}^{+}$ is called the positive alien operator. The corresponding pointed alien operator is
\begin{equation*}
\dot\Delta_{\omega}^{+}:=
e^{-\omega/\hbar}\Delta_{\omega}^{+}.
\end{equation*}
The exponential factor restores the action shift carried by the Borel
singularity. For instance, if a formal solution with exponential weight
$e^{-\Phi_0/\hbar}$ has a Borel singularity at
\[
\omega=\Phi_1-\Phi_0,
\]
then $\dot\Delta_{\omega}^{+}$ produces a contribution with weight
\[
e^{-\Phi_1/\hbar}.
\]

The finite-dimensional analysis of \cite{LiLiTang2026} provides the guiding
principle for the comparison below. There, the pointed alien operator acts on
the formal saddle expansion $\widetilde I_p$ attached to a critical point
$p$, and produces the expansion $\widetilde I_q$ attached to another
critical point $q$ when
\[
\omega=S(q)-S(p)
\]
is a Borel singularity. More precisely, if
\[
\dot\Delta_{\omega}^{+}\,\widetilde I_p
=
N_{pq}\,\widetilde I_q,
\]
then $N_{pq}$ agrees with the corresponding Picard--Lefschetz
wall-crossing coefficient; in elementary situations, this is the signed
count of direct connecting trajectories at the relevant Stokes phase:
$$
N_{pq}=\text{signed}\ \sharp \{\text{flow}\ \gamma(t): \gamma(-\infty)=q, \gamma(+\infty)=p\}.
$$
This identification is established in \cite{LiLiTang2026} by comparing the
Stokes automorphism on the resurgent side with the thimble jump formula on
the Picard--Lefschetz side.

Here is an example for illustration. Consider $X=\C^*$ and the exponential integral 
\[
\int_{\Gamma}e^{-\frac{1}{2\hbar}\left(z-\frac{1}{z}\right)}\frac{\dd z}{z}.
\]
There are two critical points $p_{\pm}=\pm i$, which lead to two asymptotic series $ I_{\pm}(\hbar)$. The Stokes phases appear at $\theta_*=\frac{\pi}{2}\ \text{or}\ \frac{3\pi}{2}$. 
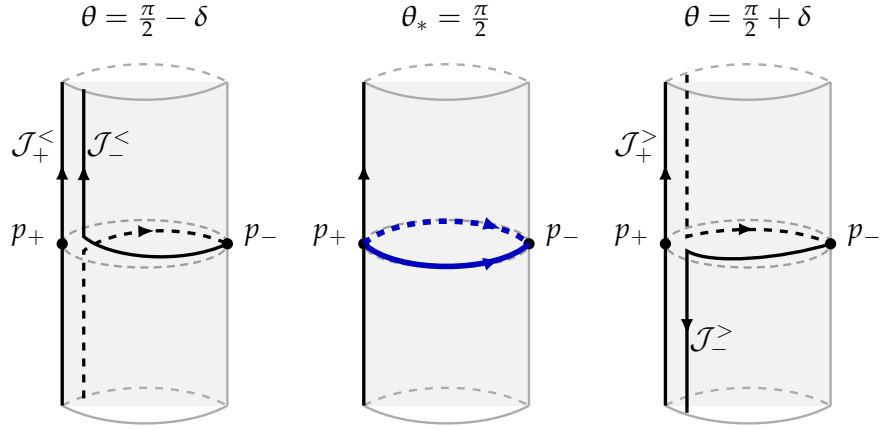
\begin{figure}[H]
\centering

\begin{tikzpicture}[>=Latex, scale=0.95]

\tikzset{
  cyl/.style={
    gray!65,
    line width=0.9pt
  },
  cylfill/.style={
    fill=gray!10,
    draw=none
  },
  waist/.style={
    gray!75,
    densely dashed,
    line width=0.9pt
  },
  crit/.style={
    circle,
    fill=black,
    inner sep=1.5pt
  },
  Jcurve/.style={
    black,
    line width=1.25pt,
    postaction={decorate},
    decoration={
      markings,
      mark=at position 0.75 with {\arrow{Latex[length=2.2mm]}}
    }
  },
  JcurveBack/.style={
    black,
    line width=1.25pt,
    dashed,
    postaction={decorate},
    decoration={
      markings,
      mark=at position 0.75 with {\arrow{Latex[length=2.2mm]}}
    }
  },
  traj/.style={
    blue!75!black,
    line width=2.1pt,
    postaction={decorate},
    decoration={
      markings,
      mark=at position 0.82 with {\arrow{Latex[length=2.4mm]}}
    }
  },
  trajback/.style={
    blue!75!black,
    line width=2.1pt,
    dashed,
    postaction={decorate},
    decoration={
      markings,
      mark=at position 0.82 with {\arrow{Latex[length=2.4mm]}}
    }
  }
}
\def\rx{1.15}
\def\ry{0.33}
\def\H{2.25}

\newcommand{\drawCylinderBase}{
\fill[cylfill] (-\rx,-\H) rectangle (\rx,\H);

\draw[cyl,dashed] (-\rx,\H)
  .. controls (-0.55,\H+\ry) and (0.55,\H+\ry)..
    (\rx,\H);
\draw[cyl,dashed] (-\rx,-\H)
  .. controls (-0.55,-\H+\ry) and (0.55,-\H+\ry)..
    (\rx,-\H);

\draw[cyl] (-\rx,-\H) -- (-\rx,\H);
\draw[cyl] (\rx,-\H) -- (\rx,\H);

\draw[cyl] (-\rx,\H)
  .. controls (-0.55,\H-\ry) and (0.55,\H-\ry)..
    (\rx,\H);
\draw[cyl] (-\rx,-\H)
  .. controls (-0.55,-\H-\ry) and (0.55,-\H-\ry)..
    (\rx,-\H);
\draw[waist] (0,0) ellipse[x radius=\rx, y radius=\ry];

\node[crit] (wp) at (-\rx,0) {};
\node[crit] (wm) at (\rx,0) {};
\node[left] at (-\rx-0.08,0.12) {$p_+$};
\node[right] at (\rx+0.08,0.12) {$p_-$};
}

\begin{scope}[shift={(-4.2,0)}]
  \node at (0,\H+0.85) {$\theta=\frac\pi2-\delta$};
  \drawCylinderBase

\draw[Jcurve]
(\rx,0).. controls  (0.5,-0.3) and  (-0.5,-0.2)..
    (-0.85,0.1)
    --
    (-0.85,\H-0.1);
\draw[JcurveBack]
    (-0.85,-\H+0.1)
    --
    (-0.85,-0.1)
  .. controls (-0.5,0.2) and (0.5,0.3)..
    (\rx,0);
\draw[Jcurve]
    (-\rx,-\H)--(-\rx,\H);
  \node[black] at (-1.55,1.35) {$\mathcal J_+^{<}$};
  \node[black] at (-0.5,1.35) {$\mathcal J_-^{<}$};
\end{scope}

\begin{scope}[shift={(0,0)}]
  \node at (0,\H+0.85) {$\theta_*=\frac{\pi}{2}$};
  \drawCylinderBase

\draw[traj]
    (-\rx,0)
  .. controls (-0.70,-0.42) and (0.70,-0.42)..
    (\rx,0);
\draw[Jcurve]
    (-\rx,-\H)--(-\rx,\H);
\draw[trajback]
    (-\rx,0)
  .. controls (-0.70,0.42) and (0.70,0.42)..
    (\rx,0);
\end{scope}

\begin{scope}[shift={(4.2,0)}]
  \node at (0,\H+0.85) {$\theta=\frac\pi2+\delta$};
  \drawCylinderBase
\draw[Jcurve]
(\rx,0).. controls  (0.5,-0.2) and  (-0.5,-0.3)..
    (-0.85,-0.1)
    --
    (-0.85,-\H-0.1);
\draw[JcurveBack]
    (-0.85,\H+0.1)
    --
    (-0.85,0.1)
  .. controls (-0.5,0.2) and (0.5,0.3)..
    (\rx,0);
\draw[Jcurve]
    (-\rx,-\H)--(-\rx,\H);
  \node[black] at (-1.55,1.35) {$\mathcal J_+^{>}$};
\node[black] at (-0.5,-1.35) {$\mathcal J_-^{>}$};
\end{scope}

\end{tikzpicture}

\end{figure}

At $\theta_*=\frac{\pi}{2}$, there are two connecting trajectories from $p_+$ to $p_-$. The two Lefschetz thimbles exhibit topological change when crossing the Stokes phase. On the other hand, explicit computation via analytic continuation shows
$$
 \dot\Delta_{\omega}^{+}\widetilde I_-=
2\,\widetilde I_+.
$$
The coefficient $2$ coincides precisely with the two connecting trajectories above. See \cite{LiLiTang2026} for details. 

\subsection{Holomorphic Morse--Floer theory}\label{sec:intro-MF} We formulate a Picard--Lefschetz problem on complexified path space, following Witten's idea of analytic continuation of quantum mechanics \cite{Witten2010}. In \cite{Witten2010}, the path integral is formulated on the phase space, and an appropriate complexification of the target is a hyperK\"{a}hler manifold. Our formulation is  instead modeled on the path integral over the configuration space, which could be understood as the adiabatic limit of the hyperK\"ahler construction.

Let $(M,g)$ be a Riemannian manifold. Fix $\tau>0$ and $a,b\in M$. Consider the Sobolev path space
\[
\PP_{a,b}(M):=
\bigl\{\gamma\in H^{1}([0,\tau],M)\,:\,\gamma(0)=a,\ \gamma(\tau)=b\bigr\}
\]
equipped with the energy functional (the $\frac{1}{4}$ convention is for convenience of resurgence analysis)
\[
\mathcal E(\gamma):=\frac{1}{4}\int_{0}^{\tau}g(\dot\gamma(t),\dot{\gamma}(t))\,\dd t.
\]
Critical points of $\mathcal E$ correspond to geodesics in $M$ from $a$ to $b$. 

We are interested in Picard–Lefschetz theory on path spaces. To set things up, we need a holomorphic structure. Assume $(M,g)$ is real analytic and extends to {an appropriate} complexification  $(M_\C,g_{\C})$ where $M_{\C}$ is a complex manifold, $g_{\C}$ is a holomorphic metric such that $M$ is a connected component of the real locus of $M_{\C}$ and $g$ is the restriction of $g_{\C}$. We discuss a few related constructions of complexification in Section \ref{sec: complexification}. Then the natural candidate of our complexified path space is 
\[
\mathcal P_{a,b}(M_{\C})
=\left\{\gamma\in H^1([0,\tau],M_{\C}) :
\gamma(0)=a,\ \gamma(\tau)=b\right\}
\]
which inherits a complex structure from that of $M_{\C}$. The complexified energy by
\[
\mathcal E_{\C}(\gamma)
:=
\frac{1}{4}\int_0^\tau
g_{\C}(\dot\gamma(t),\dot\gamma(t))\,\dd t
\]
can be viewed as a holomorphic function on the infinite-dimensional complex manifold $\mathcal P_{a,b}(M_{\C})$. 

Note that if we normalize the path by re-parametrization 
$$
\tilde \gamma(t):=\gamma(\tau t): [0,1]\to M_{\C},
$$
then 
$$
\mathcal E_{\C}(\gamma)=\frac{1}{4\tau}\int_0^1
g_{\C}(\dot{\tilde\gamma}(t),\dot{\tilde \gamma}(t))\,\dd t.
$$
Thus $\tau$ plays the analogous role of $\hbar$. Motivated by Feynman-Kac formula, Picard–Lefschetz theory of such type is related to the resurgence problem of heat kernel at time $\tau$. The following local Gevrey estimate provides the first analytic input. 

\begin{thm}[Theorem \ref{thm:local-borel-summability}]
\label{thm:local-borel-summability-introduction}
Let $(M,g)$ be real analytic and $\Delta$ be the  Laplace--Beltrami
operator acting on functions
\[
  \Delta f = - \mathrm{div}(\nabla f).
\]
Let $V$ be a sufficiently small compact
convex neighborhood with the heat kernel asymptotic expansion 
\[
K(\tau;x,y)\sim (4\pi \tau)^{-n/2}e^{-d(x,y)^2/(4\tau)}\sum_{j\ge0}u_j(x,y)\tau^j, \quad (x,y)\in V\times V. 
\]
Then the coefficient series in time $\tau$
\[
\widetilde u(\tau;x,y)=\sum_{j\ge0}u_j(x,y)\tau^j
\]
is $1$--Gevrey, uniformly for $(x,y)\in V'\times V'$ with
$V'\Subset V$. Thus its Borel transform
\[
\widehat u(\xi;x,y)=\sum_{j\ge0}\frac{u_j(x,y)}{j!}\xi^j\quad \text{converges near $\xi=0$. }
\]
\end{thm}

We prove Theorem \ref{thm:local-borel-summability-introduction} in Section \ref{subsec:proof-borel-summability}. It allows us to do resurgence analysis in the plane of Borel transform of time. In particular, the analytic continuation of the Borel transform $\widehat u(\xi;a,b)$ is expected to reflect the Picard–Lefschetz geometry of
$$
\tau\mathcal E_{\C}: \PP_{a,b}(M_{\C})\to \C. 
$$

The first surprise comes from the critical points. Standard calculus of variation shows that critical points of $\mathcal E_{\C}$ are holomorphic geodesics in $M_{\C}$ from $a$ to $b$ (see Section \ref{sec:PL-problem-path}). They contain usual real geodesics inside $M$, as well as possible new holomorphic geodesics that go out of $M$ into $M_{\C}$ before reaching $b$. Let $a,b\in M$ be close and such that there exists a unique minimal geodesic in $M$ from $a$ to $b$. The heat kernel $K(\tau;a,b)$ is constructed purely from the data on $M$, and we conjecture that analytic continuation of its Borel germ detects geometric information on $M_{\C}$ that is invisible from $M$. For example, consider the set of Borel singularities of $\widehat u(\xi;a,b)$. The finite-dimensional correspondence of Section~\ref{sec:intro-Borel} suggests relative complex-geodesic energies as candidate singularities. Below this is a conditional statement in general and an explicit calculation for $\HH^2_{\R}$. This can be viewed as the Picard–Lefschetz analogue of the quantum tunneling phenomenon.

As an example, consider the two dimensional hyperboloid space
\[
\HH_{\R}^2=\{-x_0^2+x_1^2+x_2^2=-1\mid x_i\in\R, x_0>0\}, \quad \dd s^2=-\dd x_0^2+\dd x_1^2+\dd x_2^2\big|_{\HH_{\R}^2}.
\] 
It has a natural complexification 
\[
\HH_{\C}^2:=\bigl\{z=(z_0,z_1,z_2)\in \C^3:\ -z_0^2+z_1^2+z_2^2=-1\bigr\}, \quad \dd s^2_{\C}=-\dd z_0^2+\dd z_1^2+\dd z_2^2\big|_{\HH_{\C}^2}.
\]
Consider two end-points  
$$
a=(1,0,0),~ b=(\cosh d,\sinh d,0) \in \HH_{\R}^2,
$$ 
where $d$ is the real geodesic distance between $a$ and $b$. Then the set of all holomorphic geodesics in $\HH_{\C}^2$ from $a$ to $b$ is indexed by $k\in \Z$ and explicitly given by (see Example \ref{eg:H^2-PL})
\begin{equation*}
\gamma_{{k},d}(t)=\bigl(\cosh(c_{k,d}t),\,\sinh(c_{k,d}t),\,0\bigr),
\qquad 0\le t\le \tau
\end{equation*}
where $c_{k,d}:=\frac{d+2\pi i k}{\tau}$. Only $\gamma_{0,d}(t)$ lies in $\HH_{\R}^2$ and all the other $\gamma_{{k},d}(t)$'s are invisible from $\HH_{\R}^2$. The critical values of the energy functional are 
$$
\tau\mathcal E_{\C}(\gamma_{{k},d})=\frac{(d+2\pi i k)^2}{4}.
$$ 

The exact heat kernel formula on $\HH^2_{\R}$ is known by
\begin{equation*}
K_{\HH_{\R}^{2}}(\tau;a,b)
=
\frac{2\sqrt{2}e^{-\tau/4}}{(4\pi \tau)^{3/2}}\int_{d^2/4}^{\infty}\frac{e^{-w/\tau}}{\sqrt{\cosh\sqrt{4w}-\cosh d}}~\dd w
\end{equation*}
which has already the form of Laplace-transform. In fact, the Borel coordinate is related to $w$ by 
$$
\xi=w-d^2/4. 
$$
The Borel singularities occur at
$$
{\cosh\sqrt{4w}}=\cosh d, \quad \text{i.e.}\quad \xi=\frac{(d+2\pi i k)^2}{4}-\frac{d^2}{4}, \quad k\in \Z.
$$
They correspond precisely to the critical values of $\tau\mathcal E_{\C}$ relative to the real geodesic $\gamma_{{0},d}$
$$
\tau\mathcal E_{\C}(\gamma_{{k},d})-\tau\mathcal E_{\C}(\gamma_{{0},d})
$$
as expected. This example was explained to us by Kontsevich \cite{Kont,Kont2020}. We make precise the following Conjecture \ref{conj:singularity} which is essentially due to Kontsevich. 

\begin{conj}[\cite{Kont, Kont2020}]\label{conj:singularity} Singularities of the Borel transform of the heat kernel $K(\tau;a,b)$ correspond to (relative) critical values of the complexified energy functional $\tau \mathcal E_{\C}$. 
\end{conj}

The Borel transform of the  short-time heat kernel asymptotics in Conjecture \ref{conj:singularity}  is established by Theorem \ref{thm:local-borel-summability-introduction}. We prove a version of Conjecture \ref{conj:singularity} in Theorem \ref{thm:Borel-singularity-complex-geodesic} 
under the assumption of its WKB asymptotic expansion along analytic continuation.

We next consider the infinite-dimensional analogue of Picard–Lefschetz/Alien correspondence as explained in Section \ref{sec:intro-Alien}. We start with the Picard–Lefschetz side. Let us choose a hermitian metric $h$ on $M_\C$ which induces a $L^2$-hermitian inner product $<\cdot,\cdot>_\hbar$ on $\mathcal P_{a,b}(M_{\C})$ by 
$$
<\delta \gamma_1, \delta \gamma_2>_h:=\int_0^\tau h(\delta \gamma_1, \delta \gamma_2) \dd t, \qquad \delta \gamma_1, \delta \gamma_2 \in T_\gamma \mathcal P_{a,b}(M_{\C}).
$$
Let $\theta$ be a chosen phase. Consider the downward gradient flow of 
$$
\re(e^{-i\theta}\mathcal E_\C). 
$$
It can be described by a map 
\[
\gamma(t, s): [0,\tau] \times [0,+\infty) \to M_{\C},\qquad \gamma(0,s)=a, \ \gamma(\tau,s)=b
\]
satisfying the Morse flow equation  (see Section \ref{sec:PL-problem-path} for details)
\begin{equation}\label{eqn:morse-flow}
\partial_s\gamma=-\nabla_{L^2}\re(e^{-i\theta}\mathcal E_\C)(\gamma) \tag{$\dagger$}.
\end{equation}

Let us next consider the resurgence side. Assume $a,b$ are chosen close enough such that the heat kernel $K(\tau;a,b)$ has an asymptotic expansion in time $\tau$
$$
K(\tau;a,b)\sim \widetilde K_0= (4\pi \tau)^{-n/2}e^{-d_0^2/(4\tau)}\sum_{j\ge0}u_{0,j}\tau^j
$$
whose coefficient series is $1$--Gevrey by Theorem \ref{thm:local-borel-summability-introduction}. Here $d_0$ is the length of the minimal real geodesic $\gamma_0$ in $M$ from $a$ to $b$, and $u_0(a,a)=1$ by the standard heat kernel normalization. 

Let $\gamma_1$ be another holomorphic geodesic in $M_{\C}$ from $a$ to $b$, and the corresponding (complexified) length is $d_1$. By the above discussion, we expect a singularity of the Borel transform of $\tilde K_0$ at 
$$
\omega= \tau \mathcal E(\gamma_1)- \tau \mathcal E(\gamma_0)=\frac{d_1^2}{4}-\frac{d_0^2}{4}. 
$$
Let 
$$
\widetilde K_1= (4\pi \tau)^{-n/2}e^{-d_1^2/(4\tau)}\sum_{j\ge0}v_{1,j}\tau^j
$$
be the formal solution of the heat equation with $v_{1,0}$ normalized by the Jacobian factor along the holomorphic geodesic $\gamma_1$. This is obtained similarly to the standard heat kernel $\tilde K_0$ by solving the transport equation recursively but with the new phase factor $e^{-d_1^2/(4\tau)}$ (see Theorem \ref{thm:heat-kernel-expansion-complex-geodesic}). Such $\widetilde K_1$ is expected to be the asymptotic expansion of certain path integral over the Lefschetz thimble associated to $\gamma_1$. 

Based on the finite-dimensional Picard–Lefschetz/Alien correspondence, we conjecture the following result on the pointed Alien derivative of the heat kernel asymptotic expansion (see also Conjecture \ref{conj:alien-PL-integrality-heat}). 

\begin{conj}[Picard–Lefschetz/Alien correspondence for heat kernel]\label{conj:alien} The pointed Alien derivative of $\widetilde K_0$ at the singularity $\omega$ is 
$$
\dot\Delta_{\omega}^{+} \widetilde K_0=N \widetilde K_1
$$
where $N$ is the appropriate signed counting of solutions of 
$$
\gamma: [0,\tau]\times (-\infty,\infty)\to M_{\C}
$$
satisfying 
\begin{enumerate}
\item the downward gradient Morse flow equation \eqref{eqn:morse-flow} at the Stokes phase $\theta=\theta_*$ where
$$
\im(e^{-i\theta_*}\mathcal E_\C(\gamma_1))=\im(e^{-i\theta_*}\mathcal E_\C(\gamma_0)), \qquad \re(e^{-i\theta_*}\mathcal E_\C(\gamma_1))>\re(e^{-i\theta_*}\mathcal E_\C(\gamma_0))
$$
\item the boundary condition 
$$
\gamma(0,s)=a, \qquad \gamma(\tau,s)=b, \qquad \forall\,s\in\R
$$
\item the initial and final condition 
\[\lim_{s\rightarrow-\infty}\|\gamma(\cdot,s)-\gamma_1\|_{H_t^1}=0,\qquad \lim_{s\rightarrow+\infty}\|\gamma(\cdot,s)-\gamma_0\|_{H_t^1}=0.\]
\end{enumerate} 

\end{conj}

In the second part of this paper, we perform a concrete test of Conjecture \ref{conj:alien} on the hyperbolic plane $\HH_{\R}^2$. For each $k\in \Z$ we have a critical path $\gamma_{{k},d}$ and the corresponding Borel singularity 
$$
\omega_k=\frac{(d+2\pi i k)^2}{4}-\frac{d^2}{4}. 
$$
In Section \ref{sec:H^2-resurgence}, we compute explicitly each asymptotic series $\widetilde K_k$ and show that
\[
\dot\Delta_{\omega_k}^{+}\,\widetilde{K}_0
=2\,\widetilde{K}_k.
\]
By Conjecture \ref{conj:alien}, this predicts two connecting trajectories from $\gamma_k$ to $\gamma_0$ at the corresponding Stokes phase. In Section \ref{sec:H^2-flow}, we prove this prediction for $k=1$ and $d$ sufficiently small by solving the Morse flow equation via perturbation method. This furnishes a nontrivial test of Conjecture \ref{conj:alien}. 

\smallskip
\subsubsection*{Acknowledgments}

The authors would like to thank Maxim Kontsevich, Ryszard Nest, David Sauzin, Shanzhong Sun, Yifan Wu, Junrong Yan, Shing-Tung Yau for helpful discussions. S.~L. was supported by NSFC No.12571068. Y.~L. was partially supported by BNSFC No.JR25001. X.~T. was supported by BNSFC No.JR25001, BNSFC Youth Project No.1254041 and NSFC Youth Project No.12501079. Part of this work was done while S.~L. was a visiting fellow of the International Centre for Mathematical Sciences  at Edinburgh in spring 2026, and he thanks the institute for its hospitality and provision of excellent working environment.

\part{Theory}

\section{Heat Kernel and Resurgence}\label{sec:heat-resurgence}

In this section, we establish local $1$--Gevrey estimates for heat-kernel short-time asymptotic series on real analytic Riemannian manifolds. This allows us to do resurgence analysis in the plane of Borel transform of time, and set up the Picard–Lefschetz problem of Morse–Floer type on path spaces. 

\subsection{Borel-Laplace transform}\label{sec:Gevrey-Borel}
We will briefly review the Gevrey asymptotics and Borel-Laplace transform that will be used in our resurgence analysis. For more details, see \cite{Balser,Costin,Ecalle1981I,Ecalle1981II,Ecalle1985III,MitschiSauzin2016,SauzinSplitting}.
\begin{defn}\label{defn:series}
Let
$\widetilde f(\hbar)
=\sum\limits_{m\ge0}a_m\hbar^m
\in \C[[\hbar]]$. Let $f(\hbar)$ be a holomorphic function defined in a sector
$$
\text{S}^{\theta_1,\theta_2,r}
=\left\{\hbar\in\C
\mid
0<|\hbar|<r,
\theta_1<\arg\hbar<\theta_2
\right\}. 
$$
\begin{itemize}
\item We say that $f$ admits $\widetilde f$ as an asymptotic
expansion in $\text{S}^{\theta_1,\theta_2,r}$, and write
\begin{equation*}
f(\hbar)\sim \widetilde f(\hbar),
\qquad
\hbar\to0,\quad \hbar\in\text{S}^{\theta_1,\theta_2,r},
\end{equation*}
if there exists constants $C_{N,\text{S}'}>0$ such that
\begin{equation}\label{eq:definitionasymptotic}
\left|f(\hbar)-\sum_{m=0}^{N-1}a_m\hbar^m\right|
\le C_{N,\text{S}'}|\hbar|^N,
\qquad
\hbar\in \text{S}' =\overline{\text{S}^{\theta_1',\theta_2',r'}},
\end{equation}  
for every $N\ge0$ and for every closed subsector $\overline{\text{S}^{\theta_1',\theta_2',r'}}\Subset\text{S}^{\theta_1,\theta_2,r}$ ($0<r'<r,
\theta_1<\theta_1'<\theta_2'<\theta_2$). 

\item We say that $\widetilde f$ is a $1$--Gevrey series if
there are constants $C,A>0$ such that 
$$
|a_m|\le C A^m m!\quad \text{for} \quad m\ge0. 
$$
It is denoted by $\widetilde f\in\C[[\hbar]]_1$.
\item We say that $f$ admits $\widetilde f$ as a $1$--Gevrey asymptotic
expansion in $\text{S}$ if there are constants $C,A>0$ such that
$C_{N,\text{S}'}= C A^N N!$ in \eqref{eq:definitionasymptotic}.
\end{itemize}
\end{defn}

\begin{defn}
Let $\alpha\in\C$ with $\operatorname{Re}\alpha>0$. For a formal series $\widetilde g(\hbar)
=\hbar^\alpha \widetilde f(\hbar)
=\sum\limits_{m\ge0}a_m\hbar^{m+\alpha}
\in \hbar^\alpha\C[[\hbar]]$,
we define its Borel transform by
\begin{equation*}
\widehat g(\xi)
:=\mathcal B\widetilde g(\xi)
=\sum_{m\ge0}\frac{a_m}{\Gamma(m+\alpha)}\xi^{m+\alpha-1}
\in\xi^{\alpha-1}\C[[\xi]].
\end{equation*}
If $\widetilde f \in\C[[\hbar]]_1$, then the above formal series converges near origin, and the condition $\operatorname{Re}\alpha>0$ ensures that $\widehat g$ is locally integrable at $\xi=0$. Throughout this definition, fractional powers of $\xi$ are defined using a principal branch of $\log\xi$. We also use the convention
\begin{equation*}
\mathcal B(a_0)=a_0\delta_0,
\end{equation*}
where $\delta_0$ denotes the Dirac delta distribution at the origin in the Borel plane. With this convention, the Borel transform extends to the case $\alpha=0$ when the constant term is allowed: for $\widetilde f =\sum\limits_{n\geq 0} a_n \hbar^n\in\C[[\hbar]]_1$,
\[
\mathcal B\widetilde f
=a_0\delta_0+\sum_{n\ge1}\frac{a_n}{(n-1)!}\xi^{n-1}
\in\C\delta_0\oplus\C\{\xi\}.
\]
For coefficient estimates we also use the unshifted analytic transform
\[
\mathcal B_0\widetilde f:=\sum_{n\ge0}\frac{a_n}{n!}\xi^n,\qquad
\mathcal B\widetilde f=a_0\delta_0+\partial_\xi\mathcal B_0\widetilde f.
\]
The inverse of $\mathcal B_0$ uses $\hbar^{-1}\int e^{-\xi/\hbar}(\cdot)\,\dd\xi$; the inverse of $\mathcal B$ below uses the integral without $\hbar^{-1}$. These conventions are not interchanged in the alien calculation.
\end{defn}

\begin{defn}\label{defn:Laplace}
Let $\theta\in\R$. Suppose that $\widehat f=\mathcal B\widetilde f$ is integrable at $0$ and it admits analytic continuation to a sectorial neighborhood of the ray $e^{i\theta}\R_{\ge0}$ and has at most exponential growth there. The Laplace transform of $\widehat f$ in direction $\theta$ is defined by
\begin{equation*}
\mathcal L^\theta\widehat f(\hbar)
:=\int_{0}^{e^{i\theta}\infty}
e^{-\xi/\hbar}\widehat f(\xi)\,\dd\xi,
\end{equation*}
whenever the integral converges. We call $\mathcal L^\theta\widehat f$ the Borel--Laplace sum of $\widetilde f$ in the direction $\theta$. Here we use the convention that $\mathcal{L}^{\theta} a_0 \delta_0 := a_0$. We denote the composition of Laplace transform and Borel transform by
\begin{equation*}
\mathcal S_\theta := \mathcal{L}^\theta \mathcal{B}.
\end{equation*}
\end{defn}

Let $\widetilde f(\hbar)=\sum\limits_{m\ge0}a_m\hbar^m\in\C[[\hbar]]_1$, and set
\begin{equation*}
\widetilde g(\hbar)=\hbar^\alpha\widetilde f(\hbar),
\end{equation*}
where either $\operatorname{Re}\alpha>0$, or $\alpha=0$ with the above $\delta_0$ convention. It is well-known that if the regular part of the Borel transform $\widehat g$ admits analytic continuation to a sectorial neighborhood of the ray $e^{i\theta}\R_{\ge0}$, with at most
exponential growth there, then 
\begin{equation*}
\hbar^{-\alpha}\mathcal S_\theta\widetilde g(\hbar) \sim \sum_{m\ge0}a_m\hbar^m  \quad \quad \text{as } \ \hbar\rightarrow 0 \ \text{ in }\  \text{S}^{\theta-\frac{\pi}{2},\theta+\frac{\pi}2,r}
\end{equation*}
as a $1$--Gevrey asymptotic expansion in the sector for sufficiently small $r>0$.

\subsection{Gevrey estimates for heat kernel asymptotics}\label{subsec:proof-borel-summability} We study the short-time coefficients and the associated local Borel germ. Throughout this subsection, the Gevrey parameter $\hbar$ used in previous subsection will be identified with the heat-time variable $\tau$\footnote{We will use $\tau$ as the time variable in the heat kernel expansion, while $t\in[0,\tau]$ denotes the parameter of path in Section \ref{sec:H^2-flow}.} (see Section \ref{sec:intro-MF}). 

Let $(M,g)$ be a smooth connected $n$--dimensional Riemannian manifold. We use the non--negative Laplace--Beltrami operator $\Delta\ge 0$, acting on functions,
\[
\Delta f=-\mathrm{div}(\nabla f).
\]
For a complete metric of bounded geometry, the heat kernel $K(\tau;x,y)$ is smooth for $\tau>0$ and $x,y\in M$. The fundamental local result is
\[
K(\tau;x,y)\sim(4\pi \tau)^{-n/2}\exp\Big(-\tfrac{d(x,y)^2}{4\tau}\Big)
\sum_{k\ge 0} u_k(x,y)\tau^k,\quad \tau\rightarrow0^+.
\]

\begin{rmk} The precise meaning of  this asymptotic expansion in the sense of Definition \ref{defn:series} is understood as
\begin{equation*}
(4\pi \tau)^{n/2}
\exp\!\left(\frac{d(x,y)^2}{4\tau}\right)
K(\tau;x,y)
\sim\sum_{k\ge0}u_k(x,y)\tau^k \in \C[[\tau]].
\end{equation*}
\end{rmk}

On a local chart $\Omega$ with coordinate $(x^i)_{1\le i\le n}$, we can write the Laplace--Beltrami operator as
\[\Delta f=-\sum_{i,j}\partial_i(g^{ij}\partial_jf)+\sum_{i}a^i(x)\partial_if(x)+b(x)f(x),\]
where $a^i(x),b(x)$ are explicitly expressed in terms of the Riemannian metric $g$. If the metric $g$ is real analytic, then $a^i(x),b(x)$ are also real analytic. In this case, we have the following result.

\begin{thm}
\label{thm:local-borel-summability}
Let $(M,g)$ be real analytic. Let $V$ be a sufficiently small compact convex neighborhood with the heat kernel asymptotic expansion 
\begin{equation*}
K(\tau;x,y)\sim (4\pi\tau)^{-n/2}e^{-d(x,y)^2/(4\tau)}\sum_{j\ge0}u_j(x,y)\tau^j, \quad (x,y)\in V\times V. 
\end{equation*}
Then $\sum\limits_{j\ge0}u_j(x,y)\tau^{j}$ is $1$--Gevrey power series uniformly for $(x,y)\in V'\times V'$ with $V'\Subset V$. 
\end{thm}

\begin{rmk}
The Gevrey estimate of Theorem \ref{thm:local-borel-summability} is a local parametrix statement. In some concrete analytic heat equations, it is supplied by Borel summability theorems such as those of Lutz--Miyake--Sch\"afke, Costin--Park--Takei and Harg\'e \cite{CPT,Har,Har2,LMS}.
\end{rmk}

The rest of this subsection is devoted to the proof of Theorem \ref{thm:local-borel-summability}. It is well known that the coefficients in
the short-time heat kernel expansion agree with the coefficients obtained from the formal solution of the heat equation after extracting the exponential factor and imposing the standard normalization $u_0(x,x)=1$. Therefore, our goal is to prove that the coefficients $u_j(x,y)$ obtained recursively from the transport equations satisfy the $1$--Gevrey estimate
\begin{equation*}
|u_j(x,y)|\leq AB^j j!
\end{equation*}
uniformly for $(x,y)\in V'\times V'$, where $V'\Subset V$ and $A,B>0$ are constants independent of $j$. 

\vskip 1em

\noindent \textbf{Parametrix, recursion relation and complex extension}
\vskip 1em

First, as in the smooth case, we have the short-time asymptotic expansion of $K(\tau;x,y)$ 
\[K(\tau;x,y)\sim (4\pi\tau)^{-n/2}e^{-d(x,y)^2/(4\tau)}\widetilde{u}(\tau;x,y),\qquad \widetilde{u}(\tau;x,y)=\sum_{i=0}^{\infty}u_i(x,y)\tau^i.\]
For simplicity, denote $H(\tau;x,y)=(4\pi \tau)^{-n/2}e^{-d(x,y)^2/(4\tau)}$. Then one can compute
\begin{equation}\label{eq:heateq}
(\partial_\tau+\Delta_x)(H\widetilde{u})=H\left(\partial_\tau \widetilde{u}+\frac{T}{\tau}\widetilde{u}+\Delta_x \widetilde{u}\right)
\end{equation}
where $T$ is a first-order differential operator of the form
\[
T\widetilde{u}=
\sum_{i,j}\frac12 g^{ij}(x)\partial_i d^2(x,y)\partial_j\widetilde{u}
+B(x,y)\widetilde{u},
\]
with an analytic scalar coefficient $B(x,y)=-\Delta_x\left(\frac{d(x,y)^2}{4}\right)-\frac n2$ on $V\times V$ determined by $g^{ij}$, $a^i(x)$, $b(x)$ and $d(x,y)$. In particular, the coefficients of the first differential operator $T$ are real analytic.

Expanding Equation \eqref{eq:heateq}, we get the recursion relation for $u_j(x,y)$:
\[
Tu_0(x,y)=0,\qquad (T+j+1)u_{j+1}(x,y)=-\Delta u_j(x,y). 
\]
In particular, the normalization $u_0(x,x)=1$ is forced by the delta initial condition of the heat kernel. These coefficient equations admit analytic solutions on $V\times V$.

\smallskip

Use holomorphic normal coordinates $\mathfrak E(y,v):=(\exp_y^{\C}v,y)$, with a fixed local trivialization of the tangent bundle. Write
\[
P_\rho(c):=\{u\in\C^n:|u-c|_\infty<\rho\}, \qquad
v_{x_0,y_0}:=(\exp_{y_0}^{\C})^{-1}(x_0).
\]
For sufficiently small $\rho>0$, define
\[
\mathcal P_\rho(V'):=\bigcup_{\substack{(x_0,y_0)\in V'\times V'\\0\le r\le1}}
\mathfrak E\bigl(P_\rho(y_0)\times P_\rho(r v_{x_0,y_0})\bigr).
\]
These are nested complex neighborhoods of $V'\times V'$. Since
\[
\sigma P_\rho(r v_{x_0,y_0})\subset P_\rho(\sigma r v_{x_0,y_0}),
\qquad 0\le\sigma\le1,
\]
they are invariant under radial transport:
\[
(x,y)\in\mathcal P_\rho(V')\quad\Longrightarrow\quad
(\Gamma_\sigma(x,y),y)\in\mathcal P_\rho(V').
\]

\begin{lem}\label{lem:recursionestimate}
After shrinking $V$ if necessary, and fixing $V'\Subset V$, choose $\rho_*>0$ so that the normal coordinates, the coefficients of $\Delta_x$, and $u_0^{\pm1}$ are holomorphic and uniformly bounded on a neighborhood of $\overline{\mathcal P_{\rho_*}(V')}$. Every coefficient $u_j$ extends holomorphically to $\mathcal P_{\rho_*}(V')$, and there is a constant $C_T>0$ such that, for every $0<\rho_1<\rho_2<\rho_*$,
\begin{equation*}
\|u_{j+1}\|_{\rho_1}
\le \frac{C_T}{j+1}\,\|\Delta_x u_j\|_{\rho_2},
\qquad
\|f\|_{\rho}:=\sup_{\mathcal P_{\rho}(V')}|f|.
\end{equation*}
In particular, restricting to the real locus gives
\[
\sup_{V'\times V'}|u_{j+1}|
\le \frac{C_T}{j+1}\sup_{V\times V}|\Delta u_j|.
\]
\end{lem}

\begin{proof}
We work in a fixed analytic coordinate chart and still write $\Delta_x$ for the holomorphic extension of the Laplace--Beltrami operator. 

After shrinking $V$, we may assume that any two points of $V$ are joined by a unique real analytic geodesic depending holomorphically on the endpoints.  Then in the chosen analytic coordinates, the real analytic dependence of the geodesic on its endpoints extends to a
holomorphic map on a small complex neighborhood of $V'\times V'$, satisfying the holomorphic geodesic equation (see Equation \eqref{eq:holomorphic-geodesic} for an explicit explanation). Denote such a holomorphic extension by
\[
\Gamma_\sigma(x,y)=\exp_y^{\C}(\sigma v),
\qquad
0\le \sigma\le 1.
\]
Thus $\Gamma_0(x,y)=y$ and $\Gamma_1(x,y)=x$. By construction of the domains,
\[
 (x,y)\in\mathcal P_\rho(V')
 \quad\Longrightarrow\quad
 (\Gamma_\sigma(x,y),y)\in\mathcal P_\rho(V')
 \quad(0\le\sigma\le1).
\]
In the holomorphic parametrix construction, the coefficients $u_j$ are still determined by the transport equations 
\[
u_{j+1}(x,y)=-u_0(x,y)
\int_0^1\sigma^j\,
u_0(\Gamma_\sigma(x,y),y)^{-1}
\bigl(\Delta_x u_j\bigr)(\Gamma_\sigma(x,y),y)
\,\dd\sigma.
\]
The holomorphicity of each $u_j$ follows from this formula:  
\begin{itemize}
    \item The function $u_0$ is holomorphic and non-vanishing on a sufficiently small complex neighborhood of the diagonal.
    \item By induction, suppose that $u_j$ extends holomorphically to the relevant complex tube. Then $\Delta_xu_j$ is holomorphic there as well, because $\Delta_x$ has holomorphic coefficients in our analytic coordinates. 
   \item The function $(x,y)\mapsto \Gamma_\sigma(x,y)$ is holomorphic and depends
continuously on $\sigma\in[0,1]$.
\end{itemize}
Then the integrand above is holomorphic in $(x,y)$ and continuous in $\sigma$. Therefore the integral defines a holomorphic function of $(x,y)$. This proves recursively that every $u_j$ extends holomorphically to $\mathcal P_{\rho_*}(V')$ on the fixed tube.

\vskip 0.1cm
It remains to prove the estimate. Since $u_0$ is holomorphic and non-vanishing on the chosen tube, the ratio $\frac{u_0(x,y)}{u_0(\Gamma_\sigma(x,y),y)}$ is uniformly bounded for $(x,y)\in \mathcal P_{\rho_1}(V'),\ 0\le \sigma\le 1$. Set
\[
C_T:=\sup_{\substack{(x,y)\in \mathcal P_{\rho_*}(V')\\0\le\sigma\le1}}
\left|\frac{u_0(x,y)}{u_0(\Gamma_\sigma(x,y),y)}\right|.
\]
Using the transport formula and the inclusion $(\Gamma_\sigma(x,y),y)\in \mathcal P_{\rho_2}(V')$, we obtain
\begin{align*}
|u_{j+1}(x,y)|\le C_T\int_0^1
\sigma^j\left|(\Delta_x u_j)(\Gamma_\sigma(x,y),y)\right|
\,\dd\sigma\le C_T\|\Delta_x u_j\|_{\rho_2}
\int_0^1 \sigma^j\,\dd\sigma
=\frac{C_T}{j+1}\|\Delta_x u_j\|_{\rho_2}.
\end{align*}
Taking the supremum over $(x,y)\in \mathcal P_{\rho_1}(V')$ gives  $\|u_{j+1}\|_{\rho_1}\le\frac{C_T}{j+1}\|\Delta u_j\|_{\rho_2}$.
Restricting to the real locus immediately gives
\[
\sup_{V'\times V'} |u_{j+1}|
\le\frac{C_T}{j+1}\sup_{V\times V}|\Delta_x u_j|,
\]
where the right-hand side may be taken over $V\times V$, because the real geodesic segments joining points of $V'$ remain in $V$ after the initial shrinking of $V$. 
\end{proof}

\vskip 1em
\noindent \textbf{Cauchy estimates and the estimates for $u_j$ and the remainder}
\vskip 1em

\begin{lem}\label{lem:Cauchy}
Let $f$ be holomorphic on $P_\rho(z_0)\subset \C^n$ and suppose that $\sup_{P_\rho(z_0)}|f|\le M$. Then for any multiindex $\alpha\in \mathbb{N}^n$,
\[
|\partial^\alpha f(z_0)|\le M~\frac{\alpha!}{\rho^{|\alpha|}}.
\]
More generally, for any $0<\rho'<\rho$, we have
\[\sup_{z\in P_{\rho'}(z_0)}|\partial^\alpha f(z)|\le M~\frac{\alpha!}{(\rho-\rho')^{|\alpha|}}\]
\end{lem}

\begin{proof}
First, let us prove the one-variable case. Assume $z_0=0$. Cauchy formula says, 
\[f^{(k)}(z)=\frac{k!}{2\pi i}\oint_{|\zeta-z|=r}\frac{f(\zeta)}{(\zeta-z)^{k+1}}\,\dd\zeta.\]
When $|z|\leq\rho'$, choose $r=\rho-\rho'$, then we have
\[|f^{(k)}(z)|\leq \frac{k!}{2\pi}\cdot(2\pi r)\cdot\frac{M}{r^{k+1}}=\frac{k!M}{r^k}=\frac{k!M}{(\rho-\rho')^k}.\]

For the multi-variable in the polydisc, just apply the one-variable Cauchy integral formula in each complex coordinate and iterate.
\end{proof}

The key lemma of the estimate for $u_j(x,y)$ is as follows:

\begin{lem}\label{lem:factorialU}
There exist constants $C_0,A_0>0$ (depending on $V'\subset V$) such that the coefficients satisfy
\begin{equation*}\label{eq:ujBound}
\|u_j\|_{\infty,V'\times V'} \le C_0~A_0^{j}~j!,
\quad j\ge0.
\end{equation*}
\end{lem}

\begin{proof}
First, for $0<\rho''<\rho'<\rho_*$, by Lemma \ref{lem:recursionestimate}, we have
\begin{equation}\label{eq:Tbound2}
\|u_{j+1}\|_{\rho''}
 \le C_T\frac{1}{j+1}\|\Delta_xu_j\|_{\rho'}.
\end{equation}

\smallskip
In the normal coordinates $(y,v)$, $\Delta_x$ is a second-order operator in $v$ with uniformly bounded holomorphic coefficients. Applying Lemma \ref{lem:Cauchy} on the nested domains $\mathcal P_\rho(V')$ gives a constant $C_L$, independent of $j$ and the radii, such that
\begin{equation}\label{eq:Lbound}
\|\Delta_xu_j\|_{\rho'}
\le C_L(\rho-\rho')^{-2}
\|u_j\|_\rho.
\end{equation}

Choose $2\rho'-\rho''\le\rho<\rho_*$, and combine inequalities \eqref{eq:Tbound2}--\eqref{eq:Lbound}, we get
\begin{align*}
\|u_{j+1}\|_{\rho''}
 \le~&\frac{C_TC_L}{j+1}(\rho-\rho')^{-2}\|u_j\|_\rho\\
 \le~&\frac{4C_TC_L}{j+1}(\rho-\rho'')^{-2}\|u_j\|_\rho.
\end{align*}
\medskip
For $N\ge1$, fix $0<\rho_N<\rho_0<\rho_*$ and choose a decreasing sequence of radii
\[\rho_0>\rho_1>\rho_2>\cdots>\rho_N\]
as follows
\[\rho_k=\rho_N+\frac{N-k}{N}\delta, \qquad \delta=\rho_0-\rho_N.\]
Then
\[\|u_N\|_{\rho_N}\leq\left(\frac{C}{\delta^2}\right)^N\frac{N^{2N}}{N!}\|u_0\|_{\rho_0}.\]
Using the weak version of Stirling formula $N!\ge\left(\frac{N}{e}\right)^N$, we obtain
\[\|u_N\|_{\rho_N}\leq\left(\frac{Ce^2}{\delta^2}\right)^NN!\|u_0\|_{\rho_0}.\]
Then choosing $C_0=\|u_0\|_{\rho_0}$, $A_0=\frac{Ce^2}{\delta^2}$, we are done. 
\end{proof}

We next estimate the remainder. Assume that $(M,g)$ is complete with bounded geometry. After shrinking the geodesically convex
neighborhood $V$, the zeroth-order heat parametrix gives
\[
C^{-1}H(t;x,y)\le K(t;x,y)\le C H(t;x,y),
\qquad x,y\in V,\quad 0<t<t_0.
\]
Write
\[
K_N^{\rm loc}(\tau;x,y):=(4\pi\tau)^{-n/2}e^{-d(x,y)^2/(4\tau)}\sum_{j=0}^N u_j(x,y)\tau^j
=H(\tau;x,y)\sum_{j=0}^N u_j(x,y)\tau^j
\]
for the truncated local parametrix. Then
\[
(\partial_\tau+\Delta_x)K_N^{\rm loc}(\tau;x,y)
=H(\tau;x,y)\tau^N E_N(x,y),
\]
where $E_N=\Delta_xu_N$. Using Lemma \ref{lem:factorialU} and Cauchy bounds for derivatives, we get
\[
\|E_N\|_{\infty,V_1\times V_1}\leq C_2 A_2^{N+1}(N+1)!,
\]
where $V'\Subset V_1\Subset V$.

We now localize the parametrix. Choose a cutoff $\chi\in C_c^\infty(V\times V)$ which is identically one in a neighborhood of the diagonal over
$V_1\times V_1$, and in particular equals one on $V'\times V'$ after shrinking $V'$ if necessary. Set
\[
\widehat K_N(\tau;x,y):=\chi(x,y)K_N^{\rm loc}(\tau;x,y).
\]
Then $\widehat K_N=K_N^{\rm loc}$ on $V'\times V'$. Moreover,
\[
(\partial_\tau+\Delta)\widehat K_N=
\chi H\tau^N E_N+C_N,\qquad C_N:=[\Delta,\chi]K_N^{\rm loc}.
\]
Here $C_N$ is supported where derivatives of $\chi$ are nonzero; thus it is supported away from the diagonal.

Let
\[
R_N(\tau;x,y):=K(\tau;x,y)-\widehat K_N(\tau;x,y).
\]
Since $\chi=1$ near the diagonal, $R_N$ has zero initial data. Thus
\[
(\partial_\tau+\Delta)R_N=-\chi H\tau^N E_N
-C_N,\qquad\lim_{\tau\to0^+}R_N(\tau)=0.
\]
By Duhamel's formula,
\[
\begin{aligned}
R_N(\tau;x,y)
=&-\int_0^\tau\int_M
K(\tau-t;x,z)\chi(z,y)H(t;z,y)t^N E_N(z,y)\,\dd z\,\dd t\\[0.1cm]
&-\int_0^\tau\int_MK(\tau-t;x,z)C_N(t;z,y)\,\dd z\,\dd t.
\end{aligned}
\]
The cutoff term has an extra Gaussian factor because $x,y$ remain in the inner neighborhood while derivatives of $\chi$ are supported outside it. Estimating each term of the truncated sum gives
\[
 \left|\int_0^\tau\int_M K(\tau-t;x,z)C_N(t;z,y)\,\dd z\,\dd t\right|
 \le C H(\tau;x,y)e^{-c/\tau}
       \sum_{j=0}^N A^{j+1}(j+1)!\tau^{j-1}.
\]
For $0\le j\le N$, note that
\[e^{-c/\tau}\le (N-j+2)!(\tau/c)^{N-j+2}\qquad\text{and}\qquad (j+1)!(N-j+2)!\le(N+3)!.\] 
By enlarging $A$, the preceding bound is at most
\[
C B^{N+1}(N+1)!\tau^{N+1}H(\tau;x,y).
\]

For the first integral, the cutoff restricts the $z$-integration to a compact subset of $V$. Comparing $H(t;z,y)$ with $K(t;z,y)$ and using the heat semigroup identity gives
\[
\int_MK(\tau-t;x,z)\chi(z,y)H(t;z,y)\,\dd z
\leq C H(\tau;x,y),\qquad 0<t<\tau.
\]
Therefore, for $x,y\in V'$,
\begin{align*}
|R_N(\tau;x,y)|&\leq C A_2^{N+1}(N+1)!H(\tau;x,y)
\int_0^\tau t^N\,\dd t +C B^{N+1}(N+1)!\tau^{N+1}H(\tau;x,y)\\
&\leq C' A'^{N+1}(N+1)!\tau^{N+1-n/2}e^{-d(x,y)^2/(4\tau)}.
\end{align*}
Since $\widehat K_N=K_N^{\rm loc}$ on $V'\times V'$, this proves the Gevrey remainder estimate.\qed

\subsection{Complexification}\label{sec: complexification}

Theorem \ref{thm:local-borel-summability} suggests that we could perform resurgence analysis of heat kernel on the plane of Borel transform of time. As explained in the introduction, this phenomenon is expected to reflect a version of Picard–Lefschetz theory on the complexified path space. A natural way to introduce such complexification is to extend the real-analytic manifold $(M,g)$ into $(M_\C,g_{\C})$, where $M_{\C}$ is a complex manifold, $g_{\C}$ is a holomorphic metric such that $M$ is a connected component of the real locus of $M_{\C}$ and $g$ is the restriction of $g_{\C}$.

There are several closely related notions of complexification in the literature. Classical results of Bruhat--Whitney and Grauert imply that every paracompact real-analytic manifold admits a local/germ-level complexification, and one may often choose $M_{\C}$ Stein. There are two classical constructions for such complexifications, known as Grauert tubes. 
\begin{itemize}
\item[(1)] The first one, by Bruhat and Whitney \cite{BW}, is based on a gluing construction: complex analytic extensions of real analytic charts are glued together in a suitable way. 
\item[(2)] The second one was devised by Grauert \cite{Gr}, and developed by Akhiezzer and Gindikin \cite{AG}, Guillemin and Stenzel \cite{GS}, Lempert and Sz\"oke \cite{LS,Sz}, Patrizio and Wong \cite{PW}.
\end{itemize}

When $M$ is equipped with a real-analytic Riemannian metric $g$, and in local real-analytic coordinates $x=(x^1,\dots,x^n)$, write
\[
g=\sum_{i,j} g_{ij}(x)\,\dd x^i\dd x^j.
\]
Since the coefficients $g_{ij}$ are real analytic, they admit holomorphic extensions $g_{ij}^{\C}(z)$ to a small complex neighborhood. Thus one obtains a holomorphic symmetric $2$-tensor
\[
g_{\C}:=\sum_{i,j} g_{ij}^{\C}(z)\,\dd z^i\dd z^j ,
\]
which restricts to $g$ on the real locus. After shrinking the neighborhood if necessary, $\det(g_{ij}^{\C}(z))\neq 0$, so $g_{\C}$ is a holomorphic metric. This extension is canonical at the level of germs: different analytic charts or different local models give biholomorphically equivalent holomorphic metrics along $M$.

As in the real case, there is a unique torsion-free holomorphic connection $\nabla^{\C}$, also known as holomorphic Levi-Civita connection \cite{BiswasDumitrescuHRM},  satisfying
\[
\nabla^{\C} g_{\C}=0.
\]
In local coordinates, its Christoffel symbols are the holomorphic extensions of the usual one:
\[
\Gamma^{k,\C}_{ij}
=\frac12\sum_{\ell}(g_{\C})^{k\ell}
\bigl(\partial_i g^{\C}_{j\ell}
+\partial_j g^{\C}_{i\ell}
-\partial_\ell g^{\C}_{ij}\bigr).
\]
The geodesic equation also extends holomorphically:
\begin{equation}\label{eq:holomorphic-geodesic}
\ddot z^k+\sum_{i,j}\Gamma^{k,\C}_{ij}(z)\dot z^i\dot z^j=0. 
\end{equation}

Their solutions will be called holomorphic geodesics. 

\begin{rmk} 
For Lie groups, one also has the universal complexification, characterized by a universal property for holomorphic homomorphisms into complex Lie groups; see, for example, \cite{HilgertNeeb2012,Hochschild1965,Hochschild1966,Neeb2000}. For Riemannian manifolds with negative sectional curvature, the structure of a Grauert domain is obstructed to extend to the whole tangent bundle. There is an upper bound on the radius of a Grauert tube. See Lempert-Sz\"oke \cite{LS,Sz}. 
\end{rmk}

In this paper, we use a fixed admissible embedded complexification
\[
\iota:M\hookrightarrow M_{\mathbb C}.
\]
This means that $M$ is a connected component of the real locus of $M_{\C}$, locally modeled on $(\mathbb C^n,\mathbb R^n)$, and that the metric $g$ extends to a holomorphic metric $g_{\mathbb C}$ on $M_{\mathbb C}$.

\begin{assumption}\label{assumption}
We assume this complexification is sufficiently large for the analysis below: the holomorphic Levi-Civita connection and the holomorphic geodesic flow are defined on a domain containing all real endpoints and complexified initial directions considered in the paper. In this
relative sense, the complexification is holomorphically geodesically complete for the family of geodesics under study. We do not impose any
global maximality or uniqueness condition on $M_{\mathbb C}$ unless this is explicitly stated. 
\end{assumption}

\begin{eg}[Compact symmetric spaces of rank 1]
The compact symmetric spaces of rank 1 are classified as follows:
\begin{itemize}
    \item[$(1)$] the $n$-sphere $S^n$, $n\geq2$;
    \item[$(2)$] the real projective space $\mathbb{RP}^n$, $n\geq2$;
    \item[$(3)$] the complex projective space $\mathbb{CP}^n$, $n\geq1$;
    \item[$(4)$] the quaternionic projective space $\mathbb{HP}^n$, $n\geq1$;
    \item[$(5)$] the Cayley projective plane $\mathbb{P}^2(\mathrm{Cayley})$.
\end{itemize}
Denote them by $M_i$ respectively according to the order. Their (good) complexifications\footnote{A good complexification is a smooth affine algebraic variety $U$ defined over $\R$ such that $U(\R)\cong M$ and $U(\R)\hookrightarrow U(\C)$ is a homotopy equivalence. Complexification $M_{\C}$ of $M$ with the condition that
\[\iota:M\hookrightarrow M_{\C} \quad\text{is a homotopy equivalence},\]
is also known as a minimal complexification.} are well known as follows (Morimoto-Nagano \cite{MN}, Huckleberry-Snow \cite{HS}, Azad \cite{Az}, Patrizio-Wong \cite{PW}).
\begin{itemize}
    \item[(I)] The standard models $M_{i,\C}$ of their complex affine algebraic manifolds are: 
\begin{itemize}
    \item[$(1)$] $M_{1,\C}=Q^n=$ complex affine hyperquadric, $n\geq2$;
    \item[$(2)$] $M_{2,\C}=\mathbb{P}^n(\C)\setminus\bar Q^{n-1}$, the complement of the isotropic projective quadric, $n\geq2$;
    \item[$(3)$] $M_{3,\C}=\{([v],[\alpha])\in\mathbb P(V)\times\mathbb P(V^*):\alpha(v)\ne0\}$, where $V=\C^{n+1}$;
    \item[$(4)$] $M_{4,\C}=\{W\in\mathrm{Gr}(2,\C^{2n+2}):\Omega|_W\ne0\}$, where $\Omega$ is a fixed nondegenerate complex symplectic form;
    \item[$(5)$] $M_{5,\C}$ is a 16 dimensional Stein manifold.
\end{itemize}
   \item[(II)] There exists a real analytic strictly plurisubharmonic exhaustion $\rho_i:M_i^{\C}\rightarrow[1,\infty)$ such that 
\[M_i=\{\rho_i=1\},\qquad\dim_{\R}M_i=\dim_{\C}M_{i,\C}.\] Furthermore, the plurisubharmonic function $u_i=\cosh^{-1}\rho_i$ satisfies the homogeneous Monge-Amp\`ere equation on $M_{i,\C}-M_i$. 
\end{itemize}

\end{eg}

\begin{eg}[Kulkarni \cite{Ku}]
Let $G$ be a compact Lie group and $H$ be a closed subgroup of $G$. Then the homogeneous space $G/H$ admits a minimal affine algebraic complexification $G^{\C}/H^{\C}$. Furthermore, the hodge numbers of $G^{\C}/H^{\C}$ have the property
\[h^{p,q}(G^{\C}/H^{\C})=0,\quad p\neq q.\]
\end{eg}

\begin{eg}[Complexification of the real hyperbolic plane $\HH_{\R}^2$] This example will play a major role in this paper, and we will discuss it in some detail. We consider the following three models.
\begin{itemize}
\item[$(1)$] The upper half plane model: 
\[\HH_{\R}^2=\{z=x+iy\in\C\mid y>0\}\] 
with the metric
\[g=\frac{\dd x^2+\dd y^2}{y^2}=-\frac{4\dd z\,\dd\bar{z}}{(z-\bar{z})^2}.\]
\item[$(2)$] The symmetric space model: \[\HH_{\R}^2=SL(2,\R)/SO(2,\R)\]
with the ($SL(2,\R)$-invariant) metric.
\item[$(3)$] The hyperboloid model: 
\[\HH_{\R}^2=\{-x_0^2+x_1^2+x_2^2=-1\mid x_i\in\R, x_0>0\}\] 
with the metric
\[\dd s^2=-\dd x_0^2+\dd x_1^2+\dd x_2^2\big|_{\HH_{\R}^2}.\]
\end{itemize}

There are two well-known complexifications of $\HH_{\R}^2$.
\vskip 0.1cm
\noindent$\bullet~$ \textbf{Minimal complexification}

\vskip 0.1cm
For the upper half plane model, we consider the space 
\[X:=\HH_{\R}^2\times\HH_{\R}^{2-}=\{(z,w)\in\C^2\mid \im z>0,\im w<0\}\]
with the conjugation
\[\tau:X\rightarrow X, \quad (z,w)\mapsto(\bar{w},\bar{z})\]
and the holomorphic metric
\[g_{\C}=-2(\dd z\otimes \dd w+\dd w\otimes \dd z)/(z-w)^2.\]
We have the embedding 
\[\iota:\HH_{\R}^2\hookrightarrow X,\quad z\mapsto(z,\bar{z}),\]
satisfying $\iota^*g_{\C}=g$. Such $X$ is a $SL(2,\R)$-invariant Stein domain in $\HH_{\C}^2$ (see the definition below). $SL(2,\R)$ acts on $X$ properly, and the metric $g_{\C}$ is a (the unique) $SL(2,\C)$-invariant holomorphic metric.

\vskip 0.2cm
\noindent$\bullet~$ \textbf{Affine complexification}
\vskip 0.1cm

Upon complexifying $SL(2,\R)$, $SO(2,\R)$, we obtain the affine complexification
\[\HH_{\C}^2=SL(2,\C)/SO(2,\C)\cong\{-z_0^2+z_1^2+z_2^2=-1\mid z_i\in\C\}\]
with the holomorphic metric
\[\dd s_{\C}^2=-\dd z_0^2+\dd z_1^2+\dd z_2^2\big|_{\HH_{\C}^2}.\]
The map
\[SL(2,\R)/SO(2,\R)\hookrightarrow SL(2,\C)/SO(2,\C),\qquad g\cdot SO(2,\R)\mapsto g\cdot SO(2,\C)\]
constitutes a $SL(2,\R)$-equivariant embedding which realizes $\HH_{\R}^2=SL(2,\R)/SO(2,\R)$ as one connected component of the totally real submanifold of $\HH_{\C}^2$. 

The two complexifications play rather different roles. The domain
\[
X=\HH_{\R}^2\times \HH_{\R}^{2-}
\]
is the minimal complexification: it is a Stein neighborhood of the real slice and corresponds to choosing one single-valued branch of the complexified geodesic geometry. By contrast, the affine quadric $\HH_{\C}^2$ is larger and contains the extra complex branches responsible for monodromy and the infinitely many holomorphic geodesics with the same real endpoints (see Example \ref{eg:complex-geodesics-H2}). In particular, $\HH^2_{\C}$ satisfies Assumption \ref{assumption} while $X$ does not. We have
\begin{itemize}
    \item $\HH_{\R}^2$ is a connected component of fixed point set of the complex conjugation in $\HH_{\C}^2$.
    \item $\iota:\HH_{\R}^2\hookrightarrow\HH_{\C}^2$ is not homotopy equivalence.
\end{itemize}

More concretely, consider the map
\[
\Psi:X\longrightarrow \HH_{\C}^2,\qquad
\Psi(z,w)=\left(
i\frac{1+zw}{z-w},\,
i\frac{z+w}{z-w},\,
i\frac{1-zw}{z-w}
\right)
\]
One checks that
\[
\Psi^*(\dd s_{\C}^2)= -\,\frac{2(\dd z\otimes \dd w+\dd w\otimes \dd z)}{(z-w)^2}=g_{\C}.
\]
Thus the holomorphic metric on the minimal model is simply the pullback of the affine holomorphic metric; in particular the holomorphic Levi-Civita connection on $X$ is the pullback of that of $\dd s_{\C}^2$. 
\end{eg}

\subsection[Picard--Lefschetz theory on path space]{Picard--Lefschetz theory on complexified path space}\label{sec:PL-problem-path}
We are ready to discuss the Picard–Lefschetz problem on the complexified path space and connect it to the resurgence analysis of heat kernel. 

Let $M_{\C}$ be a complexified manifold of $M$ with an anti-holomorphic involution $\tau$, complex structure $J$, and let $a,b\in M\subset M_{\C}$ be real endpoints. For variational problems, it is natural to work on the Sobolev path space
\[
\mathcal P_{a,b}(M_{\C})
=\left\{
\gamma\in H^1([0,\tau],M_{\C}) :
\gamma(0)=a,\ \gamma(\tau)=b
\right\}.
\]
Since $H^1([0,\tau])\hookrightarrow C^0([0,\tau])$, $\gamma$ is continuous and the endpoint conditions are well-defined. The complex structure on $M_{\C}$ induces a pointwise almost complex structure on the path space. For 
$$
\xi\in T_\gamma\mathcal P_{a,b}(M_{\C})
=H^1_0([0,\tau],\gamma^*TM_{\C}),
$$
the induced complex structure is defined by 
\[
(\mathbb J_\gamma \xi)(t)
:=J_{\gamma(t)}\xi(t).
\]
This satisfies $\mathbb J_\gamma^2=-\mathrm{id}$. 

Let $g_{\C}$ be a holomorphic complex-bilinear metric on $M_{\C}$. We define the complexified energy by
\[
\mathcal E_{\C}(\gamma)
:=\frac{1}{4}\int_0^\tau
g_{\C,\gamma(t)}(\dot\gamma(t),\dot\gamma(t))\,\dd t,\quad \gamma\in H^1.
\]
We choose $\frac{1}{4}$ for convenience of resurgence analysis. In local holomorphic coordinates $z=(z^1,\ldots,z^n)$,
$$
\mathcal E_{\C}(z)
=\frac{1}{4}\int_0^\tau \sum_{i,j}
g_{ij}(z(t))\,\dot z^i(t)\dot z^j(t)\,\dd t,
$$
where the functions $g_{ij}$ are holomorphic. $\mathcal E_{\C}$ is holomorphic in the Banach-manifold sense: in $H^1$-charts it is a holomorphic map from an open subset of a complex Hilbert space to $\C$. Equivalently, its first variation is complex-linear in the variation. Indeed, for $\xi\in T_\gamma\mathcal P^1_{a,b}(M_{\C})$, one has
\[
\dd\mathcal E_{\C,\gamma}(\xi)
=\frac{1}{4}\int_0^\tau\sum_{i,j,k}
(\partial_k g_{ij})(\gamma(t))\,\xi^k(t)\dot\gamma^i(t)\dot\gamma^j(t)\,\dd t
+\frac{1}{2}\int_0^\tau \sum_{i,j}
g_{ij}(\gamma(t))\,\dot\xi^i(t)\dot\gamma^j(t)\,\dd t.
\]
Moreover,
\[
\dd\mathcal E_{\C,\gamma}(\mathbb J_\gamma\xi)
=i\, \dd \mathcal E_{\C,\gamma}(\xi),
\]
since the coefficients $g_{ij}$ are holomorphic and the metric is complex-bilinear. Thus we have arrived at the infinite-dimensional Picard–Lefschetz problem associated to the action 
$$\mathcal E_{\C}: \PP_{a,b}(M_{\C})\to \C.$$
\begin{rmk}
As explained in the introduction, we can rescale the time by $\tau$ to normalize it to the unit interval $[0,1]$. Then $\tau$ plays the role of $\hbar$, and it is nature to use $\tau \mathcal E_{\C}$ for the holomorphic Morse function in order to match our convention in the finite dimensional case. Nevertheless, they lead to the same critical paths and Lefschetz thimbles. In the following discussions, we will use $\mathcal E_{\C}$ to derive the Morse flow equation and use the critical values of $\tau \mathcal E_{\C}$ to identify the Borel singularities. 
\end{rmk}

Let
$$
\gamma: [0,\tau]\to M_{\C}, \qquad \gamma \in \PP_{a,b}(M_{\C})
$$
be a critical path of the energy functional $\mathcal E_{\C}$. By the standard calculus of variation, $\gamma$ satisfies 
\[
\ddot z^k+\sum_{i,j}\Gamma^{k,\C}_{ij}(z)\dot z^i\dot z^j=0.
\]
where $z^i(t)$ represents $\gamma$ in local holomorphic coordinates. This is the complex analogue of geodesic equation, and solutions are holomorphic geodesics. Thus critical paths of $\mathcal E_{\C}$ corresponds precisely to holomorphic geodesics from $a$ to $b$.

Let us now consider the downward gradient Morse flow associated to $\mathcal E_{\C}$. First we choose a hermitian metric $h$ on $M_\C$. It induces a $L^2$ hermitian inner product $<\cdot,\cdot>_\hbar$ on $\mathcal P_{a,b}(M_{\C})$ by 
$$
<\delta \gamma_1, \delta \gamma_2>_h:=\int_0^\tau h(\delta \gamma_1, \delta \gamma_2) \dd t, \quad  \text{for}\quad \delta \gamma_i\in T_\gamma\mathcal P_{a,b}(M_{\C}).
$$
Write
$
h=\sum\limits_{i,j}h_{i\bar j}\dd z^i \dd\bar z^{\bar j}
$ in a local holomorphic coordinates $z=(z^1,\ldots,z^n)$. Then 
$$
<\delta \gamma_1, \delta \gamma_2>_h=\int_0^\tau \sum_{i,j} h_{i\bar j}(\gamma(t))\delta \gamma_1^i(t) \overline{\delta \gamma^{j}_2(t)}\,\dd t. 
$$

Let $\theta$ be a chosen phase. Consider the the downward gradient flow of 
$$
\re(e^{-i\theta}\mathcal E_\C). 
$$
It can be described by a map $\gamma(t, s): [0,\tau] \times [0,+\infty) \to M_{\C}$ satisfying the flow equation 
\[\partial_s\gamma=-\nabla_{L^2}\re(e^{-i\theta}\mathcal E_\C)(\gamma),\]
with the fixed endpoint conditions $\gamma(0,s)=a,\,\gamma(\tau,s)=b$. Here $\nabla_{L^2}$ is the gradient with respect to the above $L^2$-hermitian metric.

In a local holomorphic coordinate $\{z^i\}$ of $M_{\C}$, we have 
$
g_{\C}=\sum\limits_{i,j}g_{ij}\dd z^i \dd z^j,  \ h=\sum\limits_{i,j}h_{i\bar j}\dd z^i \dd\bar z^j. 
$
We also write the flow locally as $z^i(t,s)$. Then the above downward gradient Morse flow is explicitly 
\begin{equation*}
\frac{\partial z^j}{\partial s}=\frac{1}{2}\sum_{k,m}e^{i\theta}h^{j\bar k}\overline{g_{mk}\nabla^{\C}_{\dot z}\dot z^m}.
\end{equation*}
Here $\dot z^m=\frac{\partial z^m}{\partial t}$ and $\nabla^{\C}$ is the holomorphic Levi-Civita connection of the holomorphic metric $g_\C$. Both $h^{j\bar k}$ and $g_{mk}$ are evaluated at $\gamma(t,s)$.

\begin{prop}\label{prop:FG}
Let $z(s,\cdot)$ be a smooth solution of the flow equation. Then
\begin{equation*}
\frac{\dd}{\dd s}F_{\theta}\bigl(z(s,\cdot)\bigr)
=-\norm{\partial_s z(s,\cdot)}_{L^2_t}^2\le 0,
\qquad
\frac{\dd}{\dd s}G_{\theta}\bigl(z(s,\cdot)\bigr)=0.
\end{equation*}
Here $\norm{- }_{L^2_t}^2=<-,->_{h}$. In particular, $G_{\theta}$ is constant along every finite-energy half-trajectory.
\end{prop}

\begin{proof}
The first identity follows from the basic property of downward gradient flow. Since $\mathcal E_{\C}$ is holomorphic on $\PP_d^{\C}$, the Cauchy--Riemann equations imply that $\nabla G_{\theta(d)}=\mathbb J\nabla F_{\theta(d)}$. Therefore,
\begin{align*}
\frac{\dd}{\dd s}G_{\theta(d)}(z(s))
=&~\dd\,G_{\theta(d)}(z(s))(\partial_s z)
=\re<\nabla G_{\theta(d)},\partial_s z>_h\\
=&-\re<\mathbb J\nabla F_{\theta(d)},\nabla F_{\theta(d)}>_h=0.\qedhere
\end{align*}
\end{proof}

\begin{eg}
Take $M=\R^n$ with Euclidean metric $g$, and $M_{\C}=\C^n$ with holomorphic coordinate $z=(z^1,\ldots,z^n)$. The holomorphic metric and the standard hermitian metric are 
$g_{\C}=\sum\limits_{i=1}^n\dd z^i \dd z^i,\ h=\sum\limits_{i=1}^n\dd z^i \dd \bar{z}^i$. The holomorphic energy function is 
\[\mathcal E_{\C}(z)=\frac{1}{4}\int_0^{\tau}\sum_{k=1}^n\left(\dot{z}^k\right)^2\dd t.\]
In the flat case, the critical path (geodesic) is unique and given by the affine segment
$z(t)=a+\frac{t}{\tau}(b-a)$. The downward gradient Morse flow of $ \re(e^{-i\theta}\mathcal E_\C)$ is 
\[\partial_s z(s,t)=\frac{1}{2}e^{i\theta}\partial_t^2 \bar{z}(s,t),\qquad z(s,0)=a,~ z(s,\tau)=b.\]
\end{eg}

\smallskip
\begin{eg}\label{eg:H^2-PL} Consider the hyperboloid model: 
\[\HH_{\R}^2=\{-x_0^2+x_1^2+x_2^2=-1\mid x_i\in\R, x_0>0\}\]
with the complexification 
\[\HH_{\C}^2=\{-z_0^2+z_1^2+z_2^2=-1\mid z_i\in\C\}\]
and the holomorphic metric
\[\dd s_{\C}^2=-\dd z_0^2+\dd z_1^2+\dd z_2^2\big|_{\HH_{\C}^2}.\]

Let us denote\footnote{We regard $z,w\in\mathbb C^3$ as column vectors in matrix formulas. Later,
when writing points or paths explicitly, we often use the horizontal notation
$z=(z_0,z_1,z_2)$ only for readability. This should not cause confusion:
all products involving $\eta$, projections, and tangent equations are computed with column vectors.}
\[\langle z,w\rangle_{\eta}=z^T\eta w, \qquad z,w\in \C^3, \quad \eta:=\diag(-1,1,1)
\]
for the complex bilinear extension of the Lorentz form. Then
\[\HH_{\C}^2=\{z\in\C^3: \langle z, z\rangle_\eta=-1\}.\]
The holomorphic Levi-Civita connection on the quadric $\HH_{\C}^2$ is obtained by tangential projection from that on  $\C^3$. If $\gamma(t)$ is a holomorphic curve in $\HH_{\C}^2$, then
\[
\nabla^{\C}_{\dot\gamma}\dot\gamma=0
\quad\Longleftrightarrow\quad
\ddot\gamma \ \text{is normal to } T_{\gamma}\HH_{\C}^2.
\]
Since
\[
T_{\gamma}\HH_{\C}^2=\{\xi\in \C^3:\ \langle\gamma,\xi\rangle_{\eta}=0\},
\]
the normal line is $\C\gamma$, and then the geodesic equation is
\[
\ddot\gamma=\lambda(t)\gamma.
\]
Differentiating $\langle\gamma,\gamma\rangle_{\eta}=-1$ twice gives
\[
\langle\gamma,\dot\gamma\rangle_{\eta}=0,
\qquad
\langle\gamma,\ddot\gamma\rangle_{\eta}=-\langle\dot\gamma,\dot\gamma\rangle_{\eta},
\]
thus $\lambda(t)=\langle\dot\gamma,\dot\gamma\rangle_{\eta}$. Along a geodesic this quantity is constant:
\[
\frac{d}{\dd t}\langle \dot\gamma,\dot\gamma\rangle_\eta
=2\langle \ddot\gamma,\dot\gamma\rangle_\eta
=2\lambda \langle \gamma,\dot\gamma\rangle_\eta=0.
\]
Denote this constant by $c$, then every holomorphic geodesic satisfies the linear ODE
\[
\ddot\gamma=c\,\gamma.
\]
Therefore, after choosing a branch $\mu$ with $\mu^2=c$, the holomorphic geodesics are
\[
\gamma(t)=\cosh(\mu t)\,a+\frac{\sinh(\mu t)}{\mu}\,v,
\qquad
v=\dot\gamma(0)\in T_a\HH_{\C}^2,\quad \langle v,v\rangle_{\eta}=\mu^2.
\]

Now let $a,b\in \HH_{\R}^2$ and let $d=d(a,b)>0$ be their real hyperbolic distance, so that
\[
\langle a,b\rangle_{\eta}=-\cosh d.
\]
Define the unit tangent vector at $a$ pointing toward $b$ by
\[
u=\frac{b-\cosh d\, a}{\sinh d},
\qquad
\langle a,u\rangle_{\eta}=0,\quad \langle u,u\rangle_{\eta}=1.
\]
Then every holomorphic geodesic in the affine complexification joining $a$ to $b$ is contained in the complex $2$-plane $\C a\oplus \C u$ and is given by
\begin{align*}
\gamma_k(t)
=~&\cosh\!\Big(\frac{d+2\pi i k}{\tau}t\Big)\,a
+\sinh\!\Big(\frac{d+2\pi i k}{\tau}t\Big)\,u,
\qquad k\in\Z \\
=~&\cosh\!\Big(\frac{d+2\pi i k}{\tau}t\Big)\,a
+\frac{\sinh\!\big(\frac{d+2\pi i k}{\tau}t\big)}{\sinh d}\,
\bigl(b-\cosh d\,a\bigr).
\end{align*}
The corresponding energy is
\[
\mathcal E_{\C}(\gamma_{k})
=\frac{1}{4}\int_0^{\tau}g_{\HH_{\C}^2}(\gamma_k'(t),\gamma_k'(t))\,\dd t
=\frac{(d+2\pi i k)^2}{4\tau}.
\]

Let us now describe the Morse flow. Let $h$ denote the restriction of the standard ambient Hermitian metric on $\C^3$. Thus for any $\xi,\zeta\in T_{z}\HH_{\C}^2\,\subset\,\C^3$, their hermitian inner product is
\[
h(\xi,\zeta)=\xi\cdot \overline{\zeta}.
\]
The induced $L^2$-hermitian metric on the path space is
\[<\xi,\zeta>_h:=\int_0^{\tau} \xi(t)\cdot \overline{\zeta(t)}\,\dd t.
\]

Let $P_z$ be the $h$-orthogonal projection onto $T_z\HH_{\C}^2$, given by
\begin{equation*}\label{eq:projection}
P_z(v)=v-\frac{z^T\eta v}{|z|^2}\,\overline{\eta z},\qquad v\in\C^3.
\end{equation*}
Then the downward gradient Morse flow of $ \re(e^{-i\theta}\mathcal E_\C)$ becomes 
\[
\partial_s z=P_z\bigl(e^{i\theta}\eta\,\overline{z_{tt}}\bigr),
\qquad
z(s,0)=a,
\qquad
z(s,\tau)=b,
\]
with the constraint
\[\langle z(s,t), z(s,t)\rangle_{\eta}=-1.\]
\end{eg}

\subsection{Holomorphic geodesics and Borel singularities}
Now we consider the resurgence analysis of heat kernel. For $x,y$ in a sufficiently small compact convex neighborhood $V\subset M$, the expansion in Theorem \ref{thm:local-borel-summability} is the usual short-time expansion on the real
manifold. It is attached to the real minimal branch
\[
\Phi_0(x,y)=\frac{d(x,y)^2}{4},
\]
coming from the unique minimal real geodesic in a sufficiently small convex
neighborhood.

More generally, in the complexification $M_{\mathbb C}$, one may also consider asymptotic solutions of the heat equation attached to other holomorphic geodesic branches.

Throughout this subsection, unless explicitly stated otherwise, the variables $x$ and $y$ are allowed to be complex points of
$M_{\mathbb C}$. We write $\Delta_x^{\mathbb C}$ for the holomorphic extension of the Laplace operator in the $x$-variable. When $x,y\in M$, the notation $\Delta_x$ refers to the original real Laplacian.

Let
\[
E:TM_{\mathbb C}\longrightarrow M_{\mathbb C}\times M_{\mathbb C},
\qquad E(x,v)=\bigl(x,\exp_x^{\mathbb C}(v)\bigr),
\]
be the holomorphic endpoint map associated with the holomorphic
Levi-Civita connection of $g_{\mathbb C}$. For later use, we introduce the following definitions.

\begin{defn}
Let $x_0\in M_{\mathbb C}$ and $v_0\in T_{x_0}M_{\mathbb C}$. Set
\[
y_0:=\exp_{x_0}^{\mathbb C}(v_0),\qquad
\gamma_0(t):=\exp_{x_0}^{\mathbb C}(t v_0),\qquad 0\leq t\leq 1.
\]

We say that $v_0$, or equivalently the geodesic segment $\gamma_0$, is non-conjugate at the endpoint if
\[
\dd_v\exp_{x_0}^{\mathbb C}\big|_{v=v_0}:
T_{x_0}M_{\mathbb C}\longrightarrow T_{y_0}M_{\mathbb C}
\]
is an isomorphism. We say that $\gamma_0$ is non-conjugate (or non-caustic) along the segment, if
\[
\dd_v\exp_{x_0}^{\mathbb C}\big|_{v=t v_0}:
T_{x_0}M_{\mathbb C}\longrightarrow T_{\gamma_0(t)}M_{\mathbb C}
\]
is an isomorphism for every $0<t\leq 1$.
\end{defn}

By the holomorphic inverse function theorem, if
$(x_0,v_0)$ is no-conjugate, then there exists a connected open set $U_\gamma\subset TM_{\mathbb C}$ such that
\[
E_\gamma:=E|_{U_\gamma}:U_\gamma\longrightarrow \Omega_\gamma
\]
is biholomorphic onto its image
\[
\Omega_\gamma:=E_\gamma(U_\gamma)\subset M_{\mathbb C}\times M_{\mathbb C}.
\]
{If, moreover, the central segment $\gamma_0$ is non-conjugate along the
whole interval $0<t\leq 1$, then we shall shrink $U_\gamma$, if
necessary, so that for every $(x,v')\in U_\gamma$,
\[
\dd_v\exp_x^{\mathbb C}\big|_{v=t v'}
\]
is an isomorphism for all $0<t\leq 1$. In this case we call $\Omega_\gamma$ a no-conjugate (or non-caustic) geodesic branch chart.}

\begin{defn}
With the above notations. We denote the inverse branch in the second variable by
\[
v_\gamma:\Omega_\gamma\longrightarrow TM_{\mathbb C},
\qquad v_\gamma(x,y)\in T_xM_{\mathbb C},
\]
so that
\[
y=\exp_x^{\mathbb C}\bigl(v_\gamma(x,y)\bigr).
\]
The associated holomorphic geodesic branch is defined to be
\[
\gamma_{x,y}(t):=\exp_x^{\mathbb C}\bigl(t\,v_\gamma(x,y)\bigr),\qquad
0\leq t\leq 1.
\]
Thus
\[
\gamma_{x,y}(0)=x,\qquad\gamma_{x,y}(1)=y,\qquad
\dot\gamma_{x,y}(0)=v_\gamma(x,y).
\]
\end{defn}

Furthermore, we define its squared complex length to be
\[
\ell_\gamma(x,y)^2:=\int_0^1g_{\mathbb C,\gamma_{x,y}(t)}
\bigl(\dot\gamma_{x,y}(t),\dot\gamma_{x,y}(t)\bigr)\,\dd t.
\]
By construction, the integrand is constant
in $t$:
\[
\frac{\dd}{\dd t}g_{\mathbb C,\gamma_{x,y}(t)}\bigl(\dot\gamma_{x,y}(t),\dot\gamma_{x,y}(t)\bigr)
=2g_{\mathbb C,\gamma_{x,y}(t)}\bigl(\nabla^{\mathbb C}_{\dot\gamma_{x,y}}\dot\gamma_{x,y},\dot\gamma_{x,y}\bigr)=0.
\]
Thus
\[
\ell_\gamma(x,y)^2=g_{\C,x}\bigl(v_\gamma(x,y),v_\gamma(x,y)\bigr).
\]

\begin{lem}\label{lem:eikonal-equation}
Let $\Omega_{\gamma}$, $v_{\gamma}$ be defined as above and let $\Phi_\gamma=\frac{\ell_{\gamma}(x,y)^2}{4}$, then it satisfies the holomorphic eikonal equation
\[
\Phi_\gamma=g_{\mathbb C}\bigl(\nabla_x^{\mathbb C}\Phi_\gamma,\nabla_x^{\mathbb C}\Phi_\gamma\bigr)
\qquad\text{on }~ \Omega_\gamma.
\]
Here $\nabla_x^{\mathbb C}\Phi_\gamma$ denotes the holomorphic gradient
with respect to $g_{\mathbb C}$ in the
$x$-variable. If $x,y\in M$ and
the branch is the real geodesic branch, then $\nabla_x^{\mathbb C}$ restricts to the usual real gradient with respect to $g=g_{\mathbb C}|_M$.
\end{lem}

\begin{proof}
Fix $y$, and vary the initial endpoint $x$. Let
$x_\varepsilon$ be a holomorphic curve with
\[
x_0=x, \qquad \left.\frac{\dd x_\varepsilon}{\dd\varepsilon}\right|_{\varepsilon=0}=\xi \in T_xM_{\C}.
\]
After shrinking the $\varepsilon$-disc, we may assume
\[
(x_\varepsilon,y)\in \Omega_\gamma,
\] 
then the chosen geodesic branch $v_{\gamma_{x_{\varepsilon},y}}$ depends
analytically on $\varepsilon$. Thus there is an analytic $\varepsilon$--family of
geodesics
\[
\gamma_\varepsilon:[0,1]\to M_{\C},\qquad
\gamma_\varepsilon(0)=x_\varepsilon,\quad\gamma_\varepsilon(1)=y.
\]
Set
\[
J(t):=\left.\frac{\partial \gamma_\varepsilon(t)}
{\partial \varepsilon}\right|_{\varepsilon=0}.
\]
Then
\[
J(0)=\xi, \qquad J(1)=0.
\]
Differentiating the energy
\[
\ell_{\gamma_\varepsilon}(x_\varepsilon,y)^2=\int_0^1
g_{\C}(\dot\gamma_\varepsilon(t),
\dot\gamma_\varepsilon(t))\,\dd t
\]
at $\varepsilon=0$, we obtain
\begin{align*}
  \dd_x(\ell_\gamma^2)(\xi)
&=2\int_0^1g_{\C}(\nabla_t^{\C} J,\dot\gamma)\,\dd t  \\
&=2\Big[g_{\C}(J,\dot\gamma)
\Big]_{t=0}^{t=1}-2\int_0^1g_{\C}(J,\nabla_t^{\C}\dot\gamma)\,\dd t =-2\,g_{\C}(\xi,\dot\gamma(0)).
\end{align*}

It follows that
\[
\dd_x\Phi_\gamma(\xi)
=-\frac12 g_{\C}(\xi,\dot\gamma(0)),\quad\text{
equivalently,}\quad
\nabla_x^{\C}\Phi_\gamma=-\frac12 \dot\gamma(0).
\]
Thus,
\[
g_{\mathbb C}\bigl(\nabla_x^{\mathbb C}\Phi_\gamma,
\nabla_x^{\mathbb C}\Phi_\gamma\bigr)
=\frac14 g_{\mathbb C,x}\bigl(v_\gamma(x,y),v_\gamma(x,y)\bigr)
=\frac{\ell_\gamma(x,y)^2}{4}=\Phi_\gamma(x,y).\qedhere
\]
\end{proof}

\begin{thm}
\label{thm:heat-kernel-expansion-complex-geodesic}
Let $\Omega_\gamma$ be a no-conjugate holomorphic geodesic branch
chart, with phase $\Phi_\gamma$. Fix a branch of
$J_\gamma^{-1/2}$. Away from the isotropic locus, choose a radial
transverse section and prescribe analytic transport data for every
coefficient of positive order. There is then a unique formal
solution
\[
 \widetilde K_\gamma(\tau;x,y)
 =(4\pi\tau)^{-n/2}e^{-\Phi_\gamma(x,y)/\tau}
       \sum_{j\ge0}u_{\gamma,j}(x,y)\tau^j,
 \qquad u_{\gamma,0}=J_\gamma^{-1/2},
\]
of $(\partial_\tau+\Delta_x^{\C})\widetilde K_\gamma=0$ with these data. 
\end{thm}

\begin{proof}
Substitution and the eikonal identity give
\[
 (\partial_\tau+\Delta_x^{\C})
 \bigl((4\pi\tau)^{-n/2}e^{-\Phi_\gamma/\tau}u\bigr)
 =(4\pi\tau)^{-n/2}e^{-\Phi_\gamma/\tau}
 \left(\partial_\tau u+\tau^{-1}T_\gamma u+\Delta_x^{\C}u\right),
\]
where $T_\gamma=2\nabla_x^{\C}\Phi_\gamma\cdot\nabla_x^{\C}
-\Delta_x^{\C}\Phi_\gamma-n/2$. Thus
\[
 T_\gamma u_{\gamma,0}=0,
 \qquad (T_\gamma+j)u_{\gamma,j}
       =-\Delta_x^{\C}u_{\gamma,j-1}\quad(j\ge1).
\]
To keep the differentiated variable fixed correctly, use the
reversed geodesic and write
\[
 x=\exp_y^{\C}(r\omega),\qquad
 \omega\in T_yM_\C,\quad g_{\C,y}(\omega,\omega)=1.
\]
Then $\Phi_\gamma=r^2/4$ and
\[
 T_\gamma=r\partial_r+\frac r2\partial_r\log J_\gamma.
\]
Set $U_j=J_\gamma^{1/2}u_{\gamma,j}$. On a transverse section
$r=r_*\ne0$, specify $U_j(r_*,\omega,y)=c_j(\omega,y)$ for
$j\ge1$. The exact solution is
\[
 U_j(r,\omega,y)=r^{-j}\left[
 r_*^jc_j(\omega,y)
 -\int_{r_*}^r t^{j-1}J_\gamma(t,\omega,y)^{1/2}
       \Delta_x^{\C}u_{\gamma,j-1}(t,\omega,y)\,\dd t\right].
\]
The integral is along the specified local radial path, and the analytic ODE theorem gives existence and uniqueness at every order.
\end{proof}

\begin{rmk}
We call the resulting formal expansion a WKB-type formal solution or a WKB branch. The one-loop factor is the usual WKB factor; the higher-order transport data remain part of its normalization. 
\end{rmk}

\begin{rmk}
If the branch chart is reached by analytic continuation from a regular diagonal germ, the continued transport coefficients provide
the transverse data in the theorem. Thus a specified continuation path also determines a normalized formal sector. In the hyperbolic
case considered below, the continued Borel germ fixes this normalization coefficientwise.
\end{rmk}

To state a converse rigid statement, we introduce the following definition.
\begin{defn}\label{def:geometric-WKB}
Let $\Omega_{\gamma}$ be a no-conjugate geodesic branch chart, assume that there exist holomorphic functions
\[
\Psi_\kappa,\ u_{\kappa,j}\in \mathcal O(\Omega_{\gamma}),
\qquad j\geq 0,
\]
with $u_{\kappa,0}$ nowhere vanishing, such that
\[
\widetilde K_\kappa(\tau;x,y)=\frac{1}{(4\pi\tau)^{n/2}}
\exp\left(-\frac{\Psi_\kappa(x,y)}{\tau}\right)
\sum_{j\geq 0}u_{\kappa,j}(x,y)\tau^j
\]
solves
\[
\left(\partial_\tau+\Delta_x^{\mathbb C}\right)\widetilde K_\kappa=0
\]
as a formal series in $\tau$. We say that it is a geometric branch if the complex Lagrangian generated by
$\Psi_\kappa$,
\[
\Lambda_{\Psi_\kappa}:=\left\{\bigl(x,\dd_x\Psi_\kappa(x,y);y,-\dd_y\Psi_\kappa(x,y)\bigr)
:(x,y)\in\Omega_{\gamma}\right\},
\]
is contained in the holomorphic geodesic relation. That is, on the no-conjugate chart $\Omega_\gamma$, there exists a holomorphic inverse branch $v_\gamma(x,y)$ of the endpoint map $E(x,v)=(x,\exp_x^{\mathbb C}v)$ such that
\[
y=\exp_x^{\mathbb C}v_\gamma(x,y),\quad\text{and}\quad
\nabla_x^{\mathbb C}\Psi_\kappa(x,y)=-\frac12 v_\gamma(x,y).
\]
\end{defn}

\begin{prop}
\label{prop:phase-rigidity-geodesic-branch}
Let $\widetilde K_\kappa$ be a geometric branch in Definition~\ref{def:geometric-WKB}, associated with the holomorphic
geodesic branch $\gamma$. Then
\[
 \Psi_\kappa(x,y)
 =
 \frac{\ell_\gamma(x,y)^2}{4}
 \qquad\text{on }\Omega_\gamma.
\]
\end{prop}

\begin{proof}
Substituting
\[
\widetilde K_\kappa=(4\pi\tau)^{-n/2}e^{-\Psi_\kappa/\tau}
\sum_{j\geq 0}u_{\kappa,j}\tau^j
\]
into $(\partial_\tau+\Delta_x^{\mathbb C})\widetilde K_\kappa=0$ and
taking the coefficient of $\tau^{-2}$ gives
\[
\left(\Psi_\kappa-g_{\mathbb C}(\nabla_x^{\C}\Psi_\kappa,\nabla_x^{\C}\Psi_\kappa)\right)u_{\kappa,0}=0.
\]
Since $u_{\kappa,0}$ is nowhere vanishing, $\Psi_\kappa$ satisfies the holomorphic eikonal equation
\[
\Psi_\kappa=g_{\mathbb C}(\nabla_x^{\C}\Psi_\kappa,\nabla_x^{\C}\Psi_\kappa).
\]
Combined with the geometric condition, in particular,
\[\nabla_x^{\mathbb C}\Psi_\kappa(x,y)=-\frac12 v_\gamma(x,y),\]
we get
\[
\Psi_{\kappa}=g_{\mathbb C}(\nabla_x^{\C}\Psi_\kappa,\nabla_x^{\C}\Psi_\kappa)=\frac14 g_{\mathbb C,x}(v_\gamma(x,y),v_\gamma(x,y))
=\frac{1}{4}\ell_{\gamma}(x,y)^2.\qedhere\]
\end{proof}

Once the transport normalization of Theorem \ref{thm:heat-kernel-expansion-complex-geodesic} is imposed, this phase determines the normalized formal sector associated with
$\gamma$.

Now we fix $x,y\in M\subset M_{\C}$ and fix a minimal real geodesic branch $\gamma_0$ from $x$ to $y$, with phase
\[
\Phi_0(x,y)=\frac{\ell_0(x,y)^2}{4}.
\]
We have proven in Theorem \ref{thm:local-borel-summability} that the heat kernel has a $1$--Gevrey formal expansion
\[
\widetilde K_0(\tau;x,y)
=\frac{1}{(4\pi\tau)^{n/2}}
e^{-\Phi_0(x,y)/\tau}\widetilde{u}_0(\tau;x,y),\qquad \widetilde{u}_0(\tau;x,y)=
\sum_{k\geq 0}u_{0,k}(x,y)\tau^k.
\]
Let
\[
\widehat u_0(\zeta;x,y)
=\sum_{k\geq 0}u_{0,k}(x,y)\frac{\zeta^k}{k!}
\]
be the unshifted Borel transform $\mathcal B_0\widetilde u_0$; its Laplace reconstruction includes the factor $\tau^{-1}$.

\begin{thm}\label{thm:Borel-singularity-complex-geodesic}
With the above notations, suppose that $\widehat u_0$ admits analytic continuation in the Borel
plane, and suppose that $\zeta=\omega(x,y)$ is a non-zero {integrable} singularity.
After taking a small Hankel contour
$\Gamma_\omega$ around $\zeta=\omega(x,y)$, assume that the expression
\[
\widetilde K(\tau;x,y)
=
\frac{1}{(4\pi \tau)^{n/2}}
e^{-\Phi_0(x,y)/\tau}
\frac{1}{\tau}
\int_{\Gamma_\omega}
e^{-\zeta/\tau}
\widehat u_0(\zeta;x,y)\,\dd\zeta
\]
has a non-zero geometric branch of the complexified heat equation attached to a holomorphic geodesic $\gamma$
\[
\widetilde K(\tau;x,y)
\sim
\frac{1}{(4\pi\tau)^{n/2}}
e^{-\Psi_\kappa(x,y)/\tau}
\sum_{j\geq0}u_{\kappa,j}(x,y)\tau^j,
\]
Then
\[
\omega(x,y)=\Psi_\kappa(x,y)-\Phi_0(x,y)=\frac{1}{4}\bigl(l_{\gamma}^2(x,y)-l_0^2(x,y)\bigr).
\]
\end{thm}

\begin{proof}
The exponential factor in the Hankel contribution is
\[
e^{-\Phi_0(x,y)/\tau}e^{-\zeta/\tau}.
\]
Since the Hankel contour $\Gamma_\omega$ is taken around the singularity
$\zeta=\omega(x,y)$, write locally $\zeta=\omega(x,y)+\xi$, 
\[
e^{-\Phi_0(x,y)/\tau}e^{-\zeta/\tau}=e^{-(\Phi_0(x,y)+\omega(x,y))/\tau}e^{-\xi/\tau}.
\]
Thus the WKB phase carried by the Hankel contribution is
\[
\Psi_\kappa(x,y)=\Phi_0(x,y)+\omega(x,y).
\]
By Proposition \ref{prop:phase-rigidity-geodesic-branch},
\[
\omega(x,y)=\Psi_\kappa(x,y)-\Phi_0(x,y)=\frac{1}{4}\bigl(l_{\gamma}^2(x,y)-l_0^2(x,y)\bigr).\qedhere
\]
\end{proof}

We next turn to several examples to clarify how the above correspondence between Borel-plane singularities and complexified geodesic branches is realized in concrete situations.
\smallskip
\begin{eg}[The hyperbolic plane $\HH_{\R}^2$]
\label{eg:complex-geodesics-H2}
The exact heat kernel formula on $\HH^2_{\R}$ is known
\begin{equation*}
K_{\HH_{\R}^{2}}(\tau;a,b)
=
\frac{2\sqrt{2}e^{-\tau/4}}{(4\pi \tau)^{3/2}}\int_{d^2/4}^{\infty}\frac{e^{-w/\tau}}{\sqrt{\cosh\sqrt{4w}-\cosh d}}~\dd w
\end{equation*}
which has already the form of Laplace-transform. The Borel coordinate is related to $w$ by 
$$
\xi=w-d^2/4. 
$$
The Borel singularities appear at
$$
{\cosh\sqrt{4w}}=\cosh d, \quad \text{i.e.}\quad \xi=\frac{(d+2\pi i k)^2}{4}-\frac{d^2}{4}, \quad k\in \Z.
$$
They correspond precisely to the critical values of $\tau\mathcal E_{\C}$ relative to the real geodesic $\gamma_{{0},d}$
$$
\tau\mathcal E_{\C}(\gamma_{{k},d})-\tau\mathcal E_{\C}(\gamma_{{0},d})
$$
as computed in Example \ref{eg:H^2-PL}. A detailed resurgence and Picard–Lefschetz analysis of this example will be given in Section \ref{sec:H^2-resurgence} and Section \ref{sec:H^2-flow}.

For general even-dimensional hyperbolic spaces, explicit Borel summation / analytic continuation statements were analyzed in \cite{Du}. A resurgence analysis of Borel singularities of hyperbolic metric on Riemann surfaces, within the Picard–Lefschetz framework,  is also studied in \cite{Wu-thesis}. 

\end{eg}

\begin{eg}[The sphere $S^2$, $0<d(x,y)<\pi$]
The heat kernel of $S^2$ is radial, and has the spectral expansion
\[
K_{S^2}(\tau;r)
=
\frac{1}{4\pi}
\sum_{\ell=0}^{\infty}
(2\ell+1)e^{-\tau\ell(\ell+1)}P_\ell(\cos r),\qquad r:=d_{S^2}(x,y)\in(0,\pi).
\]
To read off the Borel-plane singularities, it is more transparent to use the following Abel--Poisson integral representation:
\[
K_{S^2}(\tau;r)=\frac{\sqrt{2}\,e^{\tau/4}}{(4\pi\tau)^{3/2}}\sum_{k\in\mathbb Z}(-1)^k
\int_r^\pi \frac{(\phi+2\pi k)
\exp\!\left[-\frac{(\phi+2\pi k)^2}{4\tau}\right]}
{\sqrt{\cos r-\cos\phi}}
\,\dd\phi.
\]
Indeed, the exponential weights in the Abel--Poisson formula are of the form
\[
\exp\!\left[-\frac{(\phi+2\pi k)^2}{4\tau}\right].
\]
Near the lower endpoint $\phi=r$, the $k$-th summand has action
\[
S_k(r)=\frac{(r+2\pi k)^2}{4}.
\]
Thus the shortest geodesic sector corresponds to $k=0$, and the other terms correspond to the other geodesic branches on the same great circle, with lengths
\[
\ell_k=|r+2\pi k|,\qquad k\in\mathbb Z.
\]
Therefore the action difference between the $k$-th geodesic branch and the shortest one is
\[
\omega_k(r)=S_k(r)-S_0(r)=
\frac{(r+2\pi k)^2-r^2}{4}
=\pi k r+\pi^2 k^2.
\]
These are precisely the expected locations of the Borel singularities of the perturbative sector associated with the shortest geodesic.

\begin{rmk}
Unlike the affine complexification of $\HH_{\R}^2$, the affine complexification of $S^2$ does not produce additional holomorphic geodesics between (no-conjugate) real endpoints. The branches indexed by $k$ are simply the real winding geodesics on $S^2$, so this example records the real multi-geodesic contribution rather than a new complexified one. 
\end{rmk}

\begin{rmk}
The restriction $0<r<\pi$ avoids the diagonal $r=0$ and the cut locus $r=\pi$. At $r=0$, closed geodesic degeneracies occur and the explicit computation relating geodesic information to Borel singularities will be computed in Example \ref{eg:r=0}. At $r=\pi$, the shortest geodesic is no longer isolated, since $x$ and $y$ are antipodal. 
\end{rmk}
\end{eg}

\begin{eg}[The sphere $S^2$, $d(x,y)=0$]\label{eg:r=0}
For $r=0$, we can compute explicitly as follows. The heat trace of $S^2$ is known by
\[
Z_{S^2}(\tau)
:=
\operatorname{Tr}(e^{-\tau\Delta_{S^2}})
=
\sum_{\ell=0}^{\infty}(2\ell+1)e^{-\ell(\ell+1)\tau}.
\]
It is convenient to work with the shifted heat trace
\[
\Phi_{S^2}(\tau)
:=e^{-\tau/4}Z_{S^2}(\tau)
=\sum_{\ell=0}^{\infty}
2\left(\ell+\frac12\right)e^{-(\ell+\frac12)^2\tau}.
\]
Set $f_\tau(x):=2x e^{-\tau x^2}$, then
\[
\Phi_{S^2}(\tau)=\sum_{\ell=0}^{\infty}
f_\tau\left(\ell+\frac12\right).
\]
We apply the midpoint Euler--Maclaurin formula
\[
\sum_{\ell=0}^{\infty} f\left(\ell+\frac12\right)
\sim
\int_0^\infty f(x)\,\dd x
-\sum_{m\geq 1}
\frac{B_{2m}(\frac12)}{(2m)!}
f^{(2m-1)}(0),
\]
where $B_{2m}(x)$ denotes the Bernoulli polynomial. In the present case,
\[
\int_0^\infty f_\tau(x)\,\dd x=\frac1\tau,\qquad 
f_\tau^{(2m-1)}(0)
=2\frac{(-\tau)^{m-1}}{(m-1)!}(2m-1)!.
\]
Substituting this into the Euler--Maclaurin formula gives
\[
\Phi_{S^2}(\tau)
\sim\frac1\tau+\sum_{m\geq 1}
\frac{(-1)^m B_{2m}(\frac12)}{m!}\tau^{m-1}=\frac{1}{\tau}\left(1+\sum_{m\ge1}\frac{(-1)^m B_{2m}(\frac12)}{m!}
\tau^{m}\right).
\]
Set $\phi_{S^2}(\tau):=\tau\Phi_{S^2}(\tau)$, then
\[
\widehat{\phi}_{S^2}(\xi)
=1+\sum_{m\geq 1}
\frac{(-1)^m B_{2m}(\frac12)}{(m!)^2}\xi^m.
\]

Using the identities of Bernoulli polynomials
\[
B_{2m}\left(\frac12\right)
=\left(2^{1-2m}-1\right)B_{2m},\qquad
B_{2m}
=(-1)^{m+1}
\frac{2(2m)!}{(2\pi)^{2m}}\zeta(2m).
\]
and the formula of the Dirichlet eta function
\[
\eta(s):=(1-2^{1-s})\zeta(s)
=\sum_{k\geq 1}\frac{(-1)^{k-1}}{k^s},
\]
we obtain
\[
\widehat{\phi}_{S^2}(\xi)=1+2\sum_{m\geq 1}
\frac{\binom{2m}{m}}{4^m}
\left(\sum_{k\geq 1}\frac{(-1)^{k-1}}{(k^2\pi^2)^m}\right)\xi^m.
\]
Interchanging the two sums for $|\xi|<\pi^2$,
gives
\[
\widehat{\phi}_{S^2}(\xi)
=1+2\sum_{k\geq 1}(-1)^{k-1}
\sum_{m\geq 1}\frac{\binom{2m}{m}}{4^m}\left(\frac{\xi}{k^2\pi^2}\right)^m.
\]
Using the generating function
\[
\sum_{m\geq 0}\frac{\binom{2m}{m}}{4^m}x^m
=(1-x)^{-1/2},
\]
we arrive at the explicit Borel expression
\[
\widehat{\phi}_{S^2}(\xi)=1+2\sum_{k\geq 1}(-1)^{k-1}\left[\left(1-\frac{\xi}{k^2\pi^2}\right)^{-1/2}
-1\right].
\]
It shows that $\widehat{\phi}_{S^2}(\xi)$ has algebraic branch points at $k^2\pi^2$, $k\in\Z_{\ge1}$. Furthermore, 
\[
\operatorname{Sing}\bigl(\widehat{\phi}_{S^2}\bigr)
=\{\,k^2\pi^2:\ k\in\Z_{\ge1}\,\}.
\]

With the Laplace kernel $e^{-\xi/\tau}$, the
corresponding exponentially small terms have the form
\[
e^{-k^2\pi^2/\tau}=e^{-(2k\pi)^2/(4\tau)}.
\]
where $(2k\pi)^2/(4\tau)$ is the geodesic action associated with the $k$-fold iterate of a
closed geodesic of length $2\pi$ on the unit sphere. Thus in the Borel plane, we can see the correspondence
\[
\text{closed geodesic length }2\pi k
\quad\Longleftrightarrow\quad
\text{Borel singularity } \xi=k^2\pi^2
\quad\Longleftrightarrow\quad
e^{-k^2\pi^2/\tau}.
\]
\end{eg}

\begin{rmk} For the compact rank-one symmetric spaces, the method relating heat-kernel
exponential sectors to Borel singularities is similar to the above on $S^2$, i.e the basic action-difference formula for the singularity locations. The difference may appear in the amplitudes and in the local type of the Borel singularities, such as the transverse Jacobi field mulitiplicities and the Stokes coefficients.
\end{rmk}

\begin{eg}
The higher-rank analogue of the hyperbolic plane
$SL(2,\R)/SO(2,\R)$ is
\[
M_n=SL(n,\R)/SO(n,\R).
\]
We use the standard real model
\[
M_n\cong \mathcal P_n^{\R}=
\{Q\in \operatorname{Sym}_n(\R): Q>0,\ \det Q=1\},
\qquad gSO(n,\R)\longmapsto gg^T.
\]
Its affine complexification is
\[
SL(n,\C)/SO(n,\C)
\simeq\mathcal P_n^{\C}
:=\{Q\in \operatorname{Sym}_n(\C):\det Q=1\}.
\]
Let
\[
\mathfrak a=
\left\{H=\operatorname{diag}(h_1,\ldots,h_n):\sum_{j=1}^n h_j=0\right\}.
\]
Up to the Weyl group, a point of $M_n$ may be written as
\[
Q=\exp(2H_0),\qquad H_0\in\mathfrak a^+.
\]
With the normalization
\[
\langle A,B\rangle=\operatorname{tr}(AB),\qquad A,B\in\mathfrak a,
\]
the principal real geodesic from $I$ to $\exp(2H_0)$ is
\[
Q_0(t)=\exp(2tH_0),\qquad 0\leq t\leq 1,
\]
and its squared length is, up to the fixed global normalization of the metric,
\[
L_0^2=\operatorname{tr}(H_0^2).
\]

The affine complexification remembers the multi-valuedness of the matrix
logarithm. If
\[
m=(m_1,\ldots,m_n)\in\mathbb Z^n,\qquad\sum_{j=1}^n m_j=0,
\]
then
\[
Q_m(t)=\exp(2t H_m),
\qquad H_m:=H_0+\pi i\,\operatorname{diag}(m_1,\ldots,m_n),\quad 0\leq t\leq 1,
\]
is a holomorphic geodesic branch joining the same endpoints in
$\mathcal P_n^{\C}$. Its complex squared length is
\[
L_m^2=\operatorname{tr}(H_m^2).
\]
Consequently the affine model contains the full lattice of holomorphic geodesic
actions
\[
\left\{\frac14\,\operatorname{tr}(H_m^2):m\in\mathbb Z^n,\ \sum_j m_j=0\right\}.
\]
Relative to the principal real branch, the corresponding action
differences are
\[
\omega_m=\frac14\left(\operatorname{tr}(H_m^2)-\operatorname{tr}(H_0^2)\right).
\]

The same lattice is visible from the heat kernel. Let $K_\tau(Q)$ be the
heat kernel based at the identity, and write its radial part as
$K_\tau(\exp(2H))$. For a noncompact symmetric space $G/K$, the Abel
transform sends the radial heat equation to a Euclidean heat equation on
$\mathfrak a$, with the usual $\rho$-shift:
\[
\mathcal A K_\tau(H)
=e^{-|\rho|^2\tau}
(4\pi\tau)^{-r/2}
\exp\left(-\frac{|H|^2}{4\tau}\right),
\qquad r=\dim\mathfrak a=n-1.
\]
Thus the heat kernel is recovered by applying the inverse Abel transform:
\[
K_\tau(\exp(2H))
=\mathcal A^{-1}
\left[e^{-|\rho|^2\tau}
(4\pi\tau)^{-r/2}
\exp\left(-\frac{|\cdot|^2}{4\tau}\right)
\right](H).
\]
For $SL(n,\R)/SO(n,\R)$, the inverse Abel transform has
Abel-type kernels with square-root branch factors
\[
\bigl(\sinh(\xi_a-h_j)\bigr)^{-1/2},
\]
so its complex singular hyperplanes are
\[
\xi_a-h_j\in \pi i\mathbb Z.
\]
These are exactly the hyperplanes produced by the same logarithmic
multi-valuedness $H_0\mapsto H_m$.
Equivalently, the possible Borel singular locations, relative to
the principal branch, are governed by the same action differences
\[
\omega_m=\frac14\left(\operatorname{tr}(H_m^2)-\operatorname{tr}(H_0^2)\right),\qquad
H_m=H_0+\pi i\,\operatorname{diag}(m_1,\ldots,m_n).
\]
Similarly, this lattice of extra branches is not visible in the minimal or germ-level complexification. 
\end{eg}

\subsection{Heat asymptotics and Alien operator}\label{sec:heat-alien}

We briefly recall the alien-operator convention used below, following
\cite{LiLiTang2026}. For systematic accounts of alien calculus and resurgence
theory, we refer the reader to Écalle's foundational trilogy
\cite{Ecalle1981I,Ecalle1981II,Ecalle1985III} and to the monograph of
Mitschi and Sauzin \cite{MitschiSauzin2016}. 

Suppose that $\widehat f(\xi)$ is a holomorphic germ near $0$ obtained by the Borel transform of a $1$--Gevrey series. Let $\gamma_\omega^+$ be a continuation path approaching a small
neighborhood of $\omega$ from the positive side of the relevant Stokes ray,
and let $\gamma_\omega^-$ be obtained from $\gamma_\omega^+$ by adjoining a
small clockwise loop around $\omega$. Suppose that $\widehat{\phi}$ admits analytic
continuation along $\gamma_\omega^{\pm}$ to a singular point $\omega\in\C^*$ and denote the resulting holomorphic germ near $\omega$ by $\cont_{\gamma_\omega^{\pm}} f$. The local variation is
\[
\var_{\omega}^+(\widehat f)
:=\cont_{\gamma_\omega^{+}}\widehat f-
\cont_{\gamma_\omega^{-}}\widehat f.
\]
It measures the local monodromy of the analytically continued Borel germ at $\omega$. See Figure \ref{fig:alien-operator-schematic} below.

\begin{figure}[H]
\centering
\begin{tikzpicture}[>=Latex, scale=0.62]

\coordinate (O) at (0,0);
\coordinate (W) at (5.2,3.5);
\coordinate (WW) at (7.8,5.27);
\coordinate (M1) at (2.6,1.75);
\coordinate (M2) at (1.71,1.13);

\fill[green!20] (5.6,3.7) circle (0.38);
\draw[green!50!black] (5.6,3.7) circle (0.38);
\node[green!40!black] at (6,4.5) {$D$};

\filldraw[black] (O) circle (1.3pt);
\node[below left] at (O) {$0$};

\draw[dashed, thick] (O) -- (WW);
\draw[->,thick]
    (0.38,0.00)
  .. controls (1.2,0.05) and (3.1,0.85)..
    (4,1.95)
  .. controls (4.62,2.55) and (4.70,5.00).. (5.4,3.9)
    ;
\draw[->,thick]
    (0.38,0.00)
  .. controls (1.2,0.05) and (3.2,0.75)..
    (5.35,2.7)
  .. controls (5.62,3) and (5.8,3.2).. (5.6,3.7)
    ;
\draw[red, line width=1.1pt]
    ($(W)+(-0.12,-0.12)$) -- ($(W)+(0.12,0.12)$);
\draw[red, line width=1.1pt]
    ($(W)+(-0.12,0.12)$) -- ($(W)+(0.12,-0.12)$);
\draw[red, line width=1.1pt]
    ($(M1)+(-0.12,-0.12)$) -- ($(M1)+(0.12,0.12)$);
\draw[red, line width=1.1pt]
    ($(M1)+(-0.12,0.12)$) -- ($(M1)+(0.12,-0.12)$);
\draw[red, line width=1.1pt]
    ($(M2)+(-0.12,-0.12)$) -- ($(M2)+(0.12,0.12)$);
\draw[red, line width=1.1pt]
    ($(M2)+(-0.12,0.12)$) -- ($(M2)+(0.12,-0.12)$);
\node[below,red] at (W) {$w$};
\node[above] at (6.05,1.58) {$\gamma_w^+$};
\node[above] at (4.65,4.58) {$\gamma_w^-$};
\end{tikzpicture}
\caption{The path $\gamma_\omega^+$ analytically continues the Borel germ
from the point near origin to a small disc $D$ behind the singular point $\omega$,
passing to the right of the preceding singularities on the same ray. The
path $\gamma_\omega^-$ is obtained from $\gamma_\omega^+$ by adjoining a
small clockwise loop around $\omega$ inside $D$.}
\label{fig:alien-operator-schematic}
\end{figure}
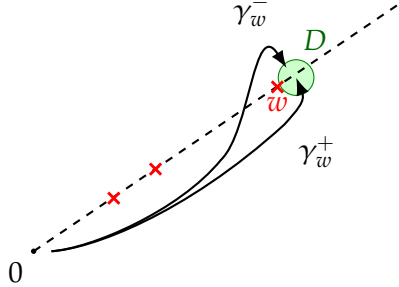

\begin{defn}
\label{def:positive-alien-operator}
With the above notation and assuming that $\hat f$ has an integrable singularity at $\omega$, we define
\[
\Delta_{\omega}^{+}
=
\mathcal B^{-1}\circ
\tau_{-\omega}\circ
\operatorname{var}_{\omega}^{+}\circ
\mathcal B : \ \hbar^\alpha\C[[\hbar]]
\longrightarrow
\hbar^\beta\C[[\hbar]]
\]
if $\tau_{-\omega}\operatorname{var}_\omega^+(\widehat f)
\in
\xi^{\beta-1}\C\{\xi\}$ where $(\tau_{-\omega}\widehat\varphi)(\xi)=\widehat\varphi(\xi+\omega)$. Here $\beta\in\C$ satisfies that $\re\beta>0$.
\end{defn}

The corresponding pointed alien operator is
\begin{equation*}
\dot\Delta_{\omega}^{+}
:=e^{-\omega/\hbar}\Delta_{\omega}^{+}.
\end{equation*}

For formal objects such as short-time heat-kernel WKB expansions, we use the
following convention. The alien calculation is performed on the 
$1$--Gevrey power series, while the exponential factor is treated as the action
label and the common algebraic heat-kernel prefactor is carried through the
calculation unchanged: \footnote{In \'Ecalle's general theory
of resurgent functions, the restriction to integrable singularities is
removed by working with a larger space of singularities, equipped with a
suitable convolution product; see, for example,
\cite{Ecalle1981I,SauzinSplitting}. In this larger setting, the exponential
factor is regarded as a formal action label in the transseries grading, and
is not part of the local Borel singularity on which the alien calculation is
performed. Algebraic factors such as $\hbar^{-n/2}$ do have Borel
transforms, but these are generally non-integrable singularities supported at
the Borel origin, for instance derivatives of the convolution unit $\delta$
in the case of negative integral powers, or more general 
singularities for fractional powers. Multiplication by such algebraic
factors on the $\hbar$-side corresponds to convolution with these origin
singularities on the Borel side. This convolution is compatible with the
operations used in the alien calculation: analytic continuation, taking the
local variation at a non-zero point $\omega$, translation back to the
origin, and the inverse Borel transform. Since the algebraic factor carries
no additional monodromy at $\omega\neq0$, it is recovered unchanged after
these operations. In the applications in this paper, this full singularity formalism gives the same
result as the simplified convention.
We therefore do not spell out the full formalism here.}

\[
\dot\Delta_\omega^+\big(
(4\pi\hbar)^{-n/2}
e^{-\Phi_0/\hbar}
\widetilde u(\hbar)\big)
=
(4\pi\hbar)^{-n/2}
e^{-(\Phi_0+\omega)/\hbar}
\Delta_\omega^+\widetilde u(\hbar),
\]
whenever the right-hand side is defined. The exponential factor restores the action shift carried by the Borel
singularity. For instance, if a formal solution with exponential weight
$e^{-\Phi_0/\hbar}$ has a Borel singularity at
\[
\omega=\Phi_1-\Phi_0,
\]
then $\dot\Delta_{\omega}^{+}$ produces a contribution with weight $e^{-\Phi_1/\hbar}$.

\begin{lem}
\label{lem:pointed-alien-commutation}
Let $\re\alpha,\re\beta>0$. Let
$\widehat\phi(\xi,x)\in \xi^{\alpha-1}\C\{\xi\}$ and
assume that its Borel inverse $\widetilde\phi(\hbar,x)$ satisfies $\mathcal{B}\Delta^+_{\omega(x)}\widetilde\phi
\in
\xi^{\beta-1}\C\{\xi\}$ for fixed $x$. Assume that $\widehat\phi$, $\omega(x)$, and the chosen continuation paths
$\gamma_x^\pm$ depend analytically on $x$, in the sense that
\begin{equation}\label{eq:commuteswithpartialx}
\partial_x\circ \operatorname{cont}_{\gamma_x^\pm}
=
\operatorname{cont}_{\gamma_x^\pm}\circ \partial_x,
\qquad
\partial_x\circ \mathcal{B}
=
\mathcal{B} \circ \partial_x.
\end{equation}
Then
\begin{equation}
\bigl[\hbar^2\partial_\hbar,\dot\Delta_{\omega(x)}^{+}\bigr]\widetilde\phi=0,
\qquad
\bigl[\partial_x,\dot\Delta_{\omega(x)}^{+}\bigr]\widetilde\phi=0.
\end{equation}
\end{lem}

\begin{proof}
The first commutation is immediate from the definition of the alien operator and
\[
\mathcal B\bigl(\hbar^2\partial_\hbar\widetilde\phi\bigr)
=\xi\,\widehat\phi.
\]
We prove the second commutation. Here we use the enlarged space of
Borel singularities discussed in the preceding footnote, so that for any $\widetilde f \in \hbar^{\beta}\C[[\hbar]]_1$,
\[
\partial_\xi\mathcal B\widetilde f
=
\mathcal{B}\hbar^{-1} \widetilde f.
\]
By the assumptions \eqref{eq:commuteswithpartialx} and the chain rule, 
\[
\partial_x
\mathcal B(\Delta_{\omega(x)}^+\widetilde\phi)
=\mathcal B\bigl(\Delta_{\omega(x)}^+(\partial_x\widetilde\phi)\bigr)
+(\partial_x\omega(x))\,\partial_\xi
\mathcal B(\Delta_{\omega(x)}^+\widetilde\phi).
\]
Here the second term on the right comes from the shift $\xi \to \xi+ \omega(x)$. Applying $\mathcal B^{-1}$, we obtain
\[
\partial_x\Delta_{\omega(x)}^+\widetilde\phi
=\Delta_{\omega(x)}^+(\partial_x\widetilde\phi)
+\frac{\partial_x\omega(x)}{\hbar}\Delta_{\omega(x)}^+\widetilde\phi.
\]
Therefore
\begin{align*}
\partial_x
\bigl(\dot\Delta_{\omega(x)}^+\widetilde\phi\bigr)
&=\partial_x\bigl(e^{-\omega(x)/\hbar}\Delta_{\omega(x)}^+\widetilde\phi\bigr)  \\
&=-\frac{\partial_x\omega(x)}{\hbar}
e^{-\omega(x)/\hbar}
\Delta_{\omega(x)}^+\widetilde\phi
+e^{-\omega(x)/\hbar}
\partial_x\Delta_{\omega(x)}^+\widetilde\phi  \\
&=-\frac{\partial_x\omega(x)}{\hbar}
e^{-\omega(x)/\hbar}
\Delta_{\omega(x)}^+\widetilde\phi
+e^{-\omega(x)/\hbar}\left(\Delta_{\omega(x)}^+(\partial_x\widetilde\phi)
+\frac{\partial_x\omega(x)}{\hbar}
\Delta_{\omega(x)}^+\widetilde\phi
\right)  \\
&=e^{-\omega(x)/\hbar}
\Delta_{\omega(x)}^+(\partial_x\widetilde\phi)
=\dot\Delta_{\omega(x)}^+(\partial_x\widetilde\phi). \qedhere
\end{align*} 
\end{proof}

\begin{cor}
\label{cor:alien-heat-solution-action-gap}
Let
\[
\widetilde K_0(\hbar;x,y)
=(4\pi\hbar)^{-n/2}e^{-\Phi_0(x,y)/\hbar}
\widetilde u_0(\hbar;x,y)
\]
be the formal heat solution associated with the minimal real geodesic. Assume that $\dot\Delta_\omega^+$ is defined on $\widetilde K_0$ and that the hypotheses of Lemma \ref{lem:pointed-alien-commutation}
hold for the coordinate derivatives entering $\Delta_x^\C$. Then
\[
\left(\partial_\hbar+\Delta_x^\C\right)
\bigl(\dot\Delta_\omega^+\widetilde K_0\bigr)=0.
\]
Let $\gamma$ be a nonconjugate holomorphic geodesic branch with $\omega=\Phi_\gamma-\Phi_0$. If $\dot\Delta_\omega^+\widetilde K_0$ has the transport normalization
of $c\,\widetilde K_\gamma$ for some $c\in\C$, where $\widetilde K_\gamma$ is normalized as in
Theorem \ref{thm:heat-kernel-expansion-complex-geodesic}, then
\[
\dot\Delta_{\Phi_\gamma-\Phi_0}^{+}\widetilde K_0=c\,\widetilde K_\gamma.
\]
\end{cor}

\begin{proof}
Lemma \ref{lem:pointed-alien-commutation} gives
\[
\left(\partial_\hbar+\Delta_x^\C\right)
\bigl(\dot\Delta_\omega^+\widetilde K_0\bigr)=0.
\]
The pointed factor changes the phase from $\Phi_0$ to
$\Phi_0+\omega=\Phi_\gamma$.
The transport uniqueness of Theorem \ref{thm:heat-kernel-expansion-complex-geodesic}
then gives
\[
\dot\Delta_{\Phi_\gamma-\Phi_0}^{+}\widetilde K_0=c\,\widetilde K_\gamma.\qedhere
\]
\end{proof}

Whenever the proportionality relation
\[
\dot\Delta_{\Phi_\gamma-\Phi_0}^{+}\widetilde K_0
=
S_{0\gamma}^{+}\,\widetilde K_\gamma
\]
holds, we call $S_{0\gamma}^{+}$ the Stokes coefficient from $\widetilde K_0$ to 
$\widetilde K_\gamma$. Whenever this proportionality has been proved and the nonzero sector
$\widetilde K_\gamma$ is fixed by its full normalization, the coefficient is unique.

The expected integrality of $S_{0\gamma}^{+}$ is a jump number of Lefschetz thimble rather than a purely formal consequence of alien calculus. In
finite-dimensional exponential integrals, the number computed from alien calculus agrees with the Picard--Lefschetz wall-crossing coefficient; after
choosing orientations, this coefficient is an intersection number, or
equivalently the signed count of direct connecting trajectories at the
corresponding Stokes phase \cite{LiLiTang2026}. 

\begin{conj}
\label{conj:alien-PL-integrality-heat}
Let $(M,g)$ be real analytic, and let $(M_{\C},g_{\C})$ be a
complexification satisfying our assumption \ref{assumption}. Fix two endpoints $x,y$, and let
$\{\gamma_i\}_{i\in I}$ be the relevant no-conjugate holomorphic geodesic from $x$ to $y$. Denote their phases by $\Phi_i(x,y):=\frac{\ell_{\gamma_i}(x,y)^2}{4}$. We conjecture that each $\gamma_i$ determines a normalized formal heat
solution of the form
\[
\widetilde K_i(\hbar;x,y)
=
(4\pi\hbar)^{-n/2}
e^{-\Phi_i(x,y)/\hbar}
\widetilde u_i(\hbar;x,y),
\qquad
\widetilde u_i(\hbar;x,y)\in \C[[\hbar]],
\]
where $\widetilde u_i$ is a $1$--Gevrey formal series. Moreover, there exists coefficient $S_{ij}^{+}\in\Z$ such that
\[
\dot\Delta_{\Phi_j-\Phi_i}^{+}\widetilde K_i
=
S_{ij}^{+}\,\widetilde K_j.
\]
\end{conj}

The hyperbolic heat-kernel model studied in the next section provides the
first test case for Conjecture~\ref{conj:alien-PL-integrality-heat}. Starting
from the principal real WKB sector $\widetilde K_0$, we will compute the
pointed alien contribution at the first nontrivial action gap. The
resurgent calculation gives the coefficient (see Section \ref{sec:H^2-resurgence})
\[
S_{01}^{+}=2.
\]
On the Picard--Lefschetz side, for two sufficiently close real endpoints, the
path-space Morse analysis shows that the corresponding wall crossing is
governed by two direct connecting trajectories. In this sense, the
hyperbolic model gives a concrete local verification of the expected
integrality mechanism.

\part{A case study of $\mathbb{H}_{\R}^2$ }

The rest of this paper is devoted to a case study of  the hyperbolic space $\mathbb{H}_{\R}^2$ in the context of the Picard–Lefschetz problem and heat kernel resurgence outlined in Section \ref{sec:heat-resurgence}. We start with the resurgence analysis in Section \ref{sec:H^2-resurgence}. In Section \ref{sec:H^2-flow}, we perform the Picard–Lefschetz analysis and compare with the resurgence computations.

\section{Resurgence analysis}\label{sec:H^2-resurgence}

In this section, we analyze the resurgent structure of the hyperbolic heat kernel
through the complex length branches of the hyperbolic distance. We first
construct, directly from the radial heat equation, a normalized formal solution $\widetilde{K}_k$ associated with each branch $\ell_k=d+2\pi i k$, as described in Theorem \ref{thm:heat-kernel-expansion-complex-geodesic}. 
We then return to the principal real-geodesic expansion and study its Borel
singularities. At the $k$-th singular point $\omega_k(d)$, the local Borel-plane variation produces the same normalized formal object $\widetilde{K}_k$, with an overall coefficient $2$. Equivalently,
\[
\dot\Delta_{\omega_k(d)}^{+}\,
\widetilde{K}_0=2\,\widetilde{K}_k.
\]
Thus the hyperbolic Stokes constant is $2$, in direct analogy with the
finite-dimensional alien/Picard--Lefschetz correspondence discussed in
\cite{LiLiTang2026}.

\subsection{Short time asymptotic expansion}

Let $a,b\in\HH_{\R}^{2}$, and write $d$ for their geodesic distance $d=d(a,b)$. By homogeneity and isotropy, the heat kernel depends only on the time $\tau$ and $d$. 

\smallskip
As we explained in the introduction, $\tau$ plays the role of $\hbar$ in the saddle problem of exponential integral. In the remainder of this section, we shall use $\hbar$ to denote the time $\tau$  to align with our finite-dimensional discussion in \cite{LiLiTang2026}. Thus the heat kernel will be denoted instead by 
\[K_{\HH_{\R}^{2}}(\hbar,d).\]

A classical integral representation is
\begin{equation*}\label{eq:H2-heat-kernel-s}
K_{\HH_{\R}^{2}}(\hbar,d)
=\frac{\sqrt{2}\,e^{-\hbar/4}}{(4\pi\hbar)^{3/2}}
\int_{d}^{\infty}
\frac{s\,e^{-s^{2}/4\hbar}}
{\sqrt{\cosh s-\cosh d}}
\,\dd s.
\end{equation*}
After the change of variables
\[
w=\frac{s^{2}}{2},
\qquad
\dd w=s\,\dd s,
\]
this becomes
\begin{equation}\label{eq:H2-heat-kernel-w}
K_{\HH_{\R}^{2}}(\hbar,d)
=\frac{\sqrt{2}\,e^{-\hbar/4}}{(4\pi\hbar)^{3/2}}
\int_{d^{2}/2}^{\infty}
\frac{e^{-w/2\hbar}}
{\sqrt{\cosh\sqrt{2w}-\cosh d}}
\,\dd w.
\end{equation}
We refer to \cite{Du} for a resurgence analysis of this representation.

\smallskip

The general $1$--Gevrey discussion of Section \ref{subsec:proof-borel-summability} provides the
formal framework. In the present hyperbolic model, the exact representation
\eqref{eq:H2-heat-kernel-w} makes the relevant Borel transform explicit.

We use the half-integer Borel convention
\[
\widetilde\phi(\hbar)
=\sum_{n\ge0}c_n\,\hbar^{n+\frac12}
\quad \mapsto\quad
\widehat\phi(\xi)
=\mathcal B\widetilde\phi(\xi)
:=\sum_{n\ge0}
\frac{c_n}{\Gamma\!\left(n+\frac12\right)}
\,\xi^{\,n-\frac12}.
\]
If $\widehat\phi$ admits analytic continuation along the ray 
$e^{i\theta}\R_+$, without singularity on that ray, and has at most
exponential growth there, then we set
\begin{equation*}\label{eq:directional-laplace}
\bigl(\mathcal L^\theta\widehat\phi\bigr)(\hbar)
:=\int_0^{e^{i\theta}\infty}
e^{-\xi/\hbar}\widehat\phi(\xi)\,\dd\xi.
\end{equation*}
Then, on
\[
S_\theta=\left\{\hbar\in\C^*:\left|\arg\hbar-\theta\right|<\frac{\pi}{2}\right\},
\]
we have uniformly on every proper subsector,
\begin{equation}\label{eq:asymptoticBL}
\bigl(\mathcal L^\theta\widehat\phi\bigr)(\hbar)\sim\sum_{n\ge0}c_n\,\hbar^{n+\frac12},
\qquad\hbar\to0.
\end{equation}

\smallskip

Returning to \eqref{eq:H2-heat-kernel-w}, set
\[w=\frac{d^2}{2}+2\xi.\]
Then
\begin{equation}\label{eq:H2-heat-kernel-borel-laplace}
K_{\HH_{\R}^{2}}(\hbar,d)
=\frac{2\sqrt{2}\,e^{-\hbar/4}}{(4\pi\hbar)^{3/2}}
\exp\!\left(-\frac{d^2}{4\hbar}\right)
\bigl(\mathcal L^0\widehat\phi_d\bigr)(\hbar),
\end{equation}
where
\begin{equation}\label{eq:H2-borel-germ}
\widehat\phi_d(\xi)
=\left(\cosh\sqrt{d^2+4\xi}-\cosh d\right)^{-1/2}.
\end{equation}
Near $\xi=0$,
\[
\widehat\phi_d(\xi)
=\left(\frac{d}{2\sinh d}\right)^{1/2}
\xi^{-1/2}\sum_{n\ge0}c_n(d)\,\xi^n,
\qquad c_0(d)=1.
\]
Applying \eqref{eq:asymptoticBL} gives
\begin{equation*}\label{eq:H2-short-time-expansion}
K_{\HH_{\R}^{2}}(\hbar,d)
\sim
\frac{e^{-\hbar/4}}{4\pi\hbar}
\exp\!\left(-\frac{d^{2}}{4\hbar}\right)
\left(\frac{d}{\sinh d}\right)^{1/2}
\sum_{n\ge0}a_n(d)\,\hbar^n,
\qquad
\hbar\to0^+,
\end{equation*}
with
\begin{equation*}\label{eq:H2-an-cn-relation}
a_n(d)=\frac{\Gamma\left(n+\frac12\right)}{\sqrt{\pi}}\,c_n(d).
\end{equation*}

The coefficients $a_n(d)$ admit a closed generating formula. Define
\[
q(d):=\frac{\sinh d}{d},
\qquad\mathcal D:=\frac1d\frac{\dd}{\dd d}.
\]
Since
\[
\cosh\sqrt{d^2+4\xi}=e^{2\xi\mathcal D}\cosh d,
\]
we obtain
\[
\cosh\sqrt{d^2+4\xi}-\cosh d
=2\xi\,q(d)\left(1+\sum_{m\ge1}
\frac{(2\xi)^m}{(m+1)!}
\frac{\mathcal D^m q(d)}{q(d)}\right).
\]
Consequently,
\begin{equation*}\label{eq:H2-an-closed-formula}
a_n(d)=\frac{\Gamma\!\left(n+\frac12\right)}{\sqrt{\pi}}\,[\xi^n]\,\left(1+\sum_{m\ge1}
\frac{(2\xi)^m}{(m+1)!}\frac{\mathcal D^m q(d)}{q(d)}\right)^{-1/2},\qquad n\ge0.
\end{equation*}
In particular,
\[a_0(d)=1.\]
For $d=0$, the formula is understood by regular continuation from $d>0$.

\subsection{Formal objects from the radial heat equation}

We now define the formal objects associated with the complex length branches
directly from the radial heat equation.

Let
\[
\ell_k:=d+2\pi ik,\qquad\Phi_k(d):=\frac{\ell_k^2}{4},
\qquad k\in\Z.
\]
The branch $k=0$ is the real branch, for which $\ell_0=d$ and
$\Phi_0(d)=d^2/4$.

In accordance with Theorem \ref{thm:heat-kernel-expansion-complex-geodesic}, we have the following explicit construction. 

For each nonzero $\ell$ with $\sinh\ell\ne0$, use the analytic
coefficients defined by
\begin{equation}\label{eq:H2-normalized-germ-coefficients}
 \left(\frac{\cosh\sqrt{\ell^2+4\zeta}-\cosh\ell}
                    {2\zeta\sinh\ell/\ell}\right)^{-1/2}
 =\sum_{n\ge0}c_n(\ell)\zeta^n,
 \qquad
 A_n(\ell)=\frac{\Gamma(n+\tfrac12)}{\sqrt\pi}c_n(\ell).
\end{equation}
The square root here has constant term $1$. These coefficients
extend regularly to $\ell=0$; at other zeros of $\sinh\ell$ they
may be singular. Their evaluation at $\ell_k=d+2\pi ik$ is
unambiguous for $d>0$. The remaining one-loop square root is
chosen by the same Borel continuation path as the local sector.

\begin{prop}
\label{prop:H2-formal-heat-recursion}

For each branch $\ell=\ell_k$, there is a formal object of the form
\begin{equation}\label{eq:H2-formal-heat-object-k}
\widetilde{K}_k(\hbar,d)
=
\frac{e^{-\hbar/4}}{4\pi\hbar}
\exp\!\left(-\frac{\Phi_k(d)}{\hbar}\right)
\left(\frac{\ell_k}{\sinh\ell_k}\right)^{1/2}
\sum_{n\ge0}A_n(\ell_k)\hbar^n,
\end{equation}
obtained by substituting it 
into the radial heat equation 
\[
(\partial_\hbar+\Delta_{\mathrm{rad},\ell})
\widetilde{K}=0,\qquad 
\Delta_{\mathrm{rad},\ell}=-\partial_\ell^2-\coth \ell\,\partial_\ell.
\]
Among radial germs regular at $\ell=0$, the coefficients are
uniquely determined by
\begin{equation}\label{eq:H2-transport-recursion-An}
A_0(\ell)=1,\qquad \left(\ell\partial_\ell+n\right)A_n(\ell) 
=\mathcal R_\ell A_{n-1}(\ell),\ \text{for} \  n\geq1,
\end{equation}
where
\[
\mathcal R_\ell:=\partial_\ell^2+\frac1\ell\partial_\ell
+\frac{1}{4\sinh^2\ell}-\frac{1}{4\ell^2}.
\]
\end{prop}

\begin{proof}
Work first with the auxiliary variable $\ell$, and write
\[
\widetilde{K}(\hbar,\ell)
=\frac{e^{-\hbar/4}}{4\pi\hbar}
e^{-\ell^2/4\hbar}\left(\frac{\ell}{\sinh\ell}\right)^{1/2}A(\hbar,\ell).
\]
Substitution into
\[
(\partial_\hbar+\Delta_{\rad,\ell})\widetilde{K}=0
\]
gives, after cancellation of the exponential and leading amplitude factors,
\[
\left(\hbar\partial_\hbar+\ell\partial_\ell\right)A=\hbar\,\mathcal R_\ell A.
\]
Writing
\[
A(\hbar,\ell)=\sum_{n\ge0}A_n(\ell)\hbar^n
\]
and comparing powers of $\hbar$ gives
\[
\ell\partial_\ell A_0=0,\qquad
(\ell\partial_\ell+n)A_n=\mathcal R_\ell A_{n-1}\quad(n\ge1).
\]
We fix the normalization by taking $A_0=1$.

For uniqueness of the regular radial germ, suppose $A_{n-1}$ is already fixed. Then
\[
(\ell\partial_\ell+n)A_n=F_n(\ell)
\]
has the general solution
\[
A_n(\ell)=\ell^{-n}
\left(C_n+\int_0^\ell u^{n-1}F_n(u)\,du\right).
\]
Regularity at $\ell=0$ forces $C_n=0$. For real $\ell>0$, the exact heat-kernel integral and Watson's lemma identify these coefficients with \eqref{eq:H2-normalized-germ-coefficients}.
Holomorphic continuation then determines the same coefficients on the nonsingular branches, in particular at $\ell=\ell_k$ for $d>0$.
\end{proof}

\begin{rmk}\label{rmk:Normalizationconvetion}
The reduced Jacobian of the complexified exponential map is $J(\ell_k)=\sinh\ell_k/\ell_k$; the leading factor is its inverse square root. This agrees with the normalization convention in Theorem \ref{thm:heat-kernel-expansion-complex-geodesic}. It also agrees with the usual WKB normalization. In
the WKB expansion at the holomorphic geodesic branch $\gamma_k$, one first
extracts the exponential factor $e^{-\Phi_k(d)/\hbar}$ and the one-loop
determinant
\[
\left[\frac{\det_D(-\partial_t^2)}
{\det_D(-\partial_t^2+\ell_k^2)}\right]^{1/2}
=\left(\frac{\ell_k}{\sinh\ell_k}\right)^{1/2}.
\]
The remaining WKB series is then normalized to be $1+O(\hbar)$. 
\end{rmk}

\subsection{Borel singularities and pointed alien operators in the hyperbolic case}

In this subsection assume $d>0$. We apply the alien-operator convention above to the principal Borel germ
$\widehat\phi_d$ in \eqref{eq:H2-borel-germ}. 
We have seen in Example \ref{eg:complex-geodesics-H2} that the nonzero Borel singularities are
\[\omega_k(d)=\frac{(d+2\pi i k)^2-d^2}{4},\qquad k\in\Z\setminus\{0\}.\]
Using the notation introduced in Section \ref{sec:heat-alien},
\[\omega_k(d)=\Phi_k(d)-\Phi_0(d).\]
Thus the Borel singularities occur precisely at the phase differences between
the principal real geodesic and the non-principal holomorphic geodesic branches.

Writing
\[\xi=\omega_k(d)+\zeta,\]
for any $k\in\mathbb{Z}\setminus\{0\}$, the translated local germ of $\widehat\phi_d$ at $\omega_k(d)$ is therefore
\[\widehat\phi_{d,k}(\zeta)
:=\left(\cosh\sqrt{\ell_k^2+4\zeta}
-\cosh\ell_k\right)^{-1/2}.\]
The branch is fixed by analytic continuation of $\widehat\phi_d$ along the
positive path $\gamma_{\omega_k}^{+}$. (This notation should not be confused with our notation for geodesics.
) See Figure \ref{fig:hyperbolicsingularity} below.

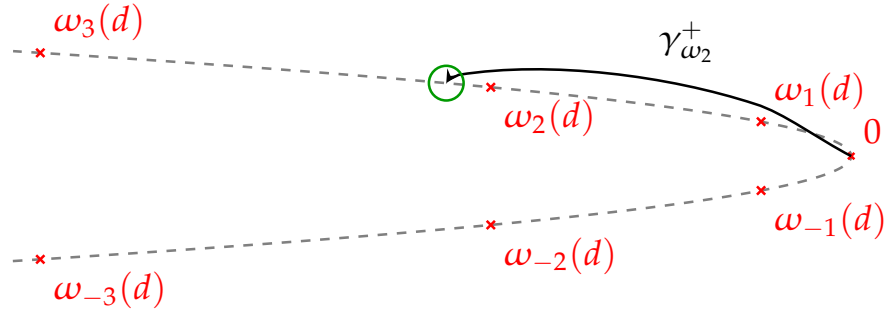
\begin{figure}[H]
\centering
\resizebox{0.67\textwidth}{!}{%
\begin{tikzpicture}[scale=0.1,>=stealth]
\def\dpar{1.2}

\coordinate (w0)  at (0,0);
\coordinate (w1)  at ({-pi*pi*1^2},{ pi*\dpar*1});
\coordinate (w2)  at ({-pi*pi*2^2},{ pi*\dpar*2});
\coordinate (w3)  at ({-pi*pi*3^2},{ pi*\dpar*3});
\coordinate (wm1) at ({-pi*pi*1^2},{-pi*\dpar*1});
\coordinate (wm2) at ({-pi*pi*2^2},{-pi*\dpar*2});
\coordinate (wm3) at ({-pi*pi*3^2},{-pi*\dpar*3});

\def\tend{2.12}
\coordinate (zend) at ({-pi*pi*\tend*\tend},{pi*\dpar*\tend});

\draw[gray,dashed,thick,domain=-11.5:11.5,samples=200,smooth]
  plot ({-(\x*\x)/(\dpar*\dpar)}, {\x});

\draw[red,thick] ($(w0)+(-0.35,-0.35)$) -- ($(w0)+(0.35,0.35)$);
\draw[red,thick] ($(w0)+(-0.35,0.35)$) -- ($(w0)+(0.35,-0.35)$);

\foreach \P in {w1,w2,w3,wm1,wm2,wm3}{
  \draw[red,thick] ($(\P)+(-0.45,-0.45)$) -- ($(\P)+(0.45,0.45)$);
  \draw[red,thick] ($(\P)+(-0.45,0.45)$) -- ($(\P)+(0.45,-0.45)$);
}

\node[red,above right] at (w0) {$0$};

\node[red,above right] at (w1) {$\omega_1(d)$};
\node[red,below right] at (w2) {$\omega_2(d)$};
\node[red,above right] at (w3) {$\omega_3(d)$};

\node[red,below right] at (wm1) {$\omega_{-1}(d)$};
\node[red,below right] at (wm2) {$\omega_{-2}(d)$};
\node[red,below right] at (wm3) {$\omega_{-3}(d)$};

\draw[black,thick,->]
  (w0)
.. controls (-3.5,1.8) and (-7.5,4.8)..
  (-10.5,5.7)
.. controls (-20,8.8) and (-33,10.5)..
  (-43.0,8.9)
.. controls (-44.0,8.45) and (-44.25,8.18)..
  (zend);

\node[black,above] at (-18,8.9) {$\gamma_{\omega_2}^+$};

\draw[green!60!black,thick] (zend) circle [radius=1.9];

\fill[black] (zend) circle (1.1pt);

\end{tikzpicture}}
\caption{Borel singularities of $\widehat{\phi}_d$ on the parabola
$\omega_k(d)=\frac{(d+2\pi i k)^2-d^2}{4}$, together with the path
$\gamma_{\omega_2}^+$.}
\label{fig:hyperbolicsingularity}
\end{figure}

Since $\widehat\phi_{d,k}$ has a square-root singularity at
$\zeta=0$, a small clockwise loop around $\omega_k(d)$ changes its sign.
With the positive variation convention,
\begin{equation}\label{eq:H2-local-variation}
\tau_{-\omega_k(d)}\,
\var_{\omega_k(d)}^{+}\widehat\phi_d
=
2\,\widehat\phi_{d,k}.
\end{equation}
Applying the inverse Borel transform gives
\[
\Delta_{\omega_k(d)}^{+}\widetilde\Psi_0
=2\,\widetilde\Psi_{d,k},\qquad
\widetilde\Psi_{d,k}
:=\mathcal B^{-1}\widehat\phi_{d,k},
\qquad \widetilde\Psi_{0}
:=\mathcal B^{-1}\widehat\phi_{d}.
\]

The leading term of $\widehat\phi_{d,k}$ is
\[
\widehat\phi_{d,k}(\zeta)=\left(\frac{\ell_k}{2\sinh\ell_k}\right)^{1/2}\zeta^{-1/2}\bigl(1+O(\zeta)\bigr),
\]
and thus
\[
\widetilde\Psi_{d,k}(\hbar)
=\sqrt{\pi}\left(\frac{\ell_k}{2\sinh\ell_k}
\right)^{1/2}\hbar^{1/2}\bigl(1+O(\hbar)\bigr).
\]
After restoring the common prefactor in
\eqref{eq:H2-heat-kernel-borel-laplace}, this reproduces the one-loop
normalization of $\widetilde{K}_k$ in
Remark \ref{rmk:Normalizationconvetion}. To compare the full WKB formal
solutions, we shall use the following commutation property.

\begin{prop}
\label{prop:H2-pointed-alien-stokes-constant}
Fix $d>0$ and $k\in\Z\setminus\{0\}$. Let
$\widetilde{K}_k$ be the WKB-type formal solution attached to
$\gamma_k$, normalized as in
\eqref{eq:H2-formal-heat-object-k} and \eqref{eq:H2-transport-recursion-An}. Then
\[
\dot\Delta_{\omega_k(d)}^{+}\,
\widetilde{K}_0=2\,\widetilde{K}_k.
\]
Equivalently, the positive Stokes constant from the principal real branch to
the $k$-th holomorphic geodesic branch is
$
S_{0k}^{+}=2.
$ In particular,
\[
\dot\Delta_{\omega_1(d)}^{+}\,
\widetilde{K}_0=2\,\widetilde{K}_1.
\]
\end{prop}

\begin{proof}
\eqref{eq:H2-local-variation} already gives
\[
 \tau_{-\omega_k(d)}\var^+_{\omega_k(d)}\widehat\phi_d
       =2\widehat\phi_{d,k}.
\]
Using \eqref{eq:H2-normalized-germ-coefficients}, we obtain
\[
 \widehat\phi_{d,k}(\zeta)
 =\left(\frac{\ell_k}{2\sinh\ell_k}\right)^{1/2}
       \sum_{n\ge0}c_n(\ell_k)\zeta^{n-1/2}.
\]
The half-integer Borel inverse is therefore
\[
 \mathcal B^{-1}\widehat\phi_{d,k}
 =\sqrt\pi\left(\frac{\ell_k}{2\sinh\ell_k}\right)^{1/2}
       \hbar^{1/2}\sum_{n\ge0}A_n(\ell_k)\hbar^n.
\]
Restoring the common factor
$2\sqrt2\,e^{-\hbar/4}(4\pi\hbar)^{-3/2}$ and the pointed
exponential $e^{-(\Phi_0+\omega_k)/\hbar}$ yields exactly
$2\widetilde K_k$ at every order. Both the one-loop sign and the
continuation path are the same on the two sides.
\end{proof}

\begin{cor}
\label{cor:H2-pointed-alien-general-transition}
Fix $d>0$. Let $j,m\in\Z$ with $j\neq m$, and set
\[
\omega_{j\to m}(d):=\Phi_m(d)-\Phi_j(d).
\]
Then
\begin{equation*}
\dot\Delta_{\omega_{j\to m}(d)}^{+}\,
\widetilde{K}_j=2\,\widetilde{K}_m.
\end{equation*}
Thus every direct pointed alien transition between two distinct complex
geodesic branches carries the same Stokes constant $2$.
\end{cor}

\begin{proof}
Let $\widehat\phi_{d,j}$ be the reduced Borel germ associated with
$\widetilde K_j$. We use Proposition \ref{prop:H2-pointed-alien-stokes-constant} after re-centering the
construction at the branch $\gamma_j$. Then the branch $\gamma_m$ appears at the Borel-plane
displacement
\[
\omega_{j\to m}(d)=\Phi_m(d)-\Phi_j(d).
\]
The radial heat equation, the transport recursion, and the normalization
$A_0=1$ are unchanged under this re-centering. Therefore
Proposition \ref{prop:H2-pointed-alien-stokes-constant} gives
\[
\tau_{-\omega_{j\to m}(d)}\operatorname{var}^{+}_{\omega_{j\to m}(d)}\widehat\phi_{d,j}
=2\,\widehat\phi_{d,m}.
\]
In particular, the singularity at $\omega_{j\to m}(d)$ is of square-root type by
Proposition \ref{prop:H2-pointed-alien-stokes-constant}, with $\gamma_j$ as the reference branch.

After inverse Borel transform and insertion of the pointed factor, the exponential weight changes from $e^{-\Phi_j(d)/\hbar}$ to
$e^{-\Phi_m(d)/\hbar}$.
Now the translated reduced germ is $2\widehat\phi_{d,m}$, and the one-loop
factor and the transport normalization are the same as those used in the
definition of $\widetilde K_m$. Thus
\[
\dot\Delta^{+}_{\omega_{j\to m}(d)}\widetilde K_j
=2\,\widetilde K_m.\qedhere
\]
\end{proof}

\begin{rmk}
At $d=0$, we have
\[
 \omega_k(0)=\omega_{-k}(0)=-\pi^2k^2,
 \qquad
 \widehat\phi_0(\xi)
 =\frac{1}{\sqrt2\,\sinh\sqrt\xi}.
\]
The two square-root singularities for $k$ and $-k$ therefore coalesce into a simple pole. This degenerate analytic limit is the counterpart of the Bott degeneration appearing in the geometric
analysis below.
\end{rmk}

The coefficient $2$ is the resurgent quantity that will later be compared
with the connecting-trajectory count for the corresponding pair of critical
paths in the complexified path-space Morse theory.

\begin{rmk}\label{rmk:alien-orientation}
The coefficients $+2$ above refer to the leading-term normalization
fixed by the heat-kernel initial condition and analytically continued
along the prescribed paths; see
Proposition \ref{prop:H2-formal-heat-recursion} and
Remark \ref{rmk:Normalizationconvetion}.
In a putative thimble interpretation, these leading terms would
specify the relative orientations, since reversing an integration
cycle reverses its entire formal expansion.
\end{rmk}

\section{Picard–Lefschetz analysis}\label{sec:H^2-flow}
In this section, we will study the $\HH_{\R}^2$ model from Picard--Lefschetz side on the complexified path space. Let us first recall necessary notations from Example \ref{eg:H^2-PL}. Let
\[
\HH_{\C}^2:=\bigl\{z=(z_0,z_1,z_2)\in \C^3:\ -z_0^2+z_1^2+z_2^2=-1\bigr\},
\qquad
\eta:=\diag(-1,1,1).
\]
After a real Lorentz isometry, we normalize the two real endpoints as
\[
a=(1,0,0),
\qquad
b_d=(\cosh d,\sinh d,0).
\]
We consider the endpoint-fixed Sobolev path space
\[
\PP_d^{\C}:=\PP_{a,b_d}(\mathbb{H}_{\C}^2)=\bigl\{z\in H^1([0,\tau],\HH_{\C}^2)\mid z(0)=a,\ z(\tau)=b_d\bigr\}
\]
equipped with the action by the holomorphic energy functional
$$
S(z)=\frac12\int_0^{\tau} z_t^T\eta z_t\,\dd t=\frac{1}{2}\int_0^{\tau}\langle z_t,z_t\rangle_{\eta}\,\dd t.
$$
\begin{rmk} The factor $1/2$ is the traditional normalization for the energy. In Section \ref{sec:PL-problem-path}, we have used the factor $1/4$ to define $\mathcal E_{\C}$ so that the exponential $\exp(-\mathcal E_{\C})$ leads to the heat kernel. Since
$$
S=2\mathcal E_{\C},
$$
they give rise to equivalent Morse flows. We choose to work with $S$ instead in this section for Picard–Lefschetz and Morse flow analysis. It simplifies expressions in various computations. 
\end{rmk}

Let $\theta$ be a chosen phase. Define the twisted real part and imaginary part of $S$ by
\begin{align*}
F_{\theta}(z)=&\re\bigl(e^{-i\theta}S(z)\bigr),
\qquad\quad
G_{\theta}(z)=\im\bigl(e^{-i\theta}S(z)\bigr). 
\end{align*}
The downward Morse flow system of $F_{\theta}$ is computed in Example \ref{eg:H^2-PL} and given by
\begin{equation}\label{eq:flow}
\begin{cases}
\partial_sz=P_z(e^{i\theta}\eta \bar{z}_{tt})=e^{i\theta}\left(\eta\bar{z}_{tt}-\frac{z^T\bar{z}_{tt}}{|z|^2}\eta\bar{z}\right)\\[0.1cm]
\langle z(s,t),z(s,t)\rangle_{\eta}=z(s,t)^T\eta z(s,t)=-1\\[0.1cm]
z(s,0)=a,~z(s,\tau)=b_d
\end{cases}.
\end{equation}

\subsection{The holomorphic geodesics and the complex Jacobi operator}

For each $j\in\Z$, set
\[
c_{j,d}:=\frac{d+2\pi i j}{\tau}.
\]
As shown in Example \ref{eg:H^2-PL}, for each $j$, there corresponds a holomorphic geodesics joining $a$ to $b_d$ by  
\begin{equation*}
\gamma_{j,d}(t)=\bigl(\cosh(c_{j,d}t),\,\sinh(c_{j,d}t),\,0\bigr),
\qquad 0\le t\le \tau.
\end{equation*}
Moreover,
\[
\frac{\dd^2}{\dd t^2}\,\gamma_{j,d}=c_{j,d}^2\gamma_{j,d}.
\]

Along $\gamma_{j,d}$, introduce the tangent and normal fields
\[
T_{j,d}(t):=\frac1{c_{j,d}}\,\frac{\dd}{\dd t}\gamma_{j,d}(t)
=\bigl(\sinh(c_{j,d}t),\,\cosh(c_{j,d}t),\,0\bigr),
\qquad
N:=(0,0,1).
\]
We have
\[
\ip{\gamma_{j,d}}{T_{j,d}}_{\eta}=0,
\qquad
\ip{\gamma_{j,d}}{N}_{\eta}=0,
\qquad
\ip{T_{j,d}}{T_{j,d}}_{\eta}=1,
\qquad
\ip{N}{N}_{\eta}=1.
\]
Therefore every tangent variation along $\gamma_{j,d}$ can be written as
\[
\xi(t)=\alpha(t)T_{j,d}(t)+\beta(t)N,
\qquad
\alpha(0)=\alpha(\tau)=\beta(0)=\beta(\tau)=0.
\]

\begin{prop}\label{prop:nondegenerate}
The $\C$-linear second-variation operator of the holomorphic action at $\gamma_{j,d}$ is
\begin{equation*}\label{eq:Ajd}
L_{j,d}^{\C}(\alpha,\beta)=\bigl(-\alpha_{tt},\,-\beta_{tt}+c_{j,d}^2\beta\bigr)
\end{equation*}
on the Dirichlet domain
\[
\mathrm{Dom}(L_{j,d}^{\C})=\bigl(H^2\cap H^1_0([0,\tau],\C)\bigr)\oplus \bigl(H^2\cap H^1_0([0,\tau],\C)\bigr).
\]
Equivalently,
\[
L_{j,d}^{\C}(\alpha T_{j,d}+\beta N)=(-\alpha_{tt})T_{j,d}+(-\beta_{tt}+c_{j,d}^2\beta)N.
\]
Furthermore, for $d>0$, $\ker L_{j,d}^{\C}=\{0\}$, thus every critical path $\gamma_{j,d}$ is nondegenerate. 
\end{prop}

\begin{proof} 
Consider an arbitrary endpoint-fixed variation of $\gamma_{j,d}$ parametrized by small $\varepsilon$
\[u(\varepsilon,t)\in \HH_{\C}^2,\qquad u(0,t)=\gamma_{j,d}(t).\]
Let
\[\xi(t)=\partial_{\varepsilon}u(0,t),\qquad \zeta(t)=\partial_{\varepsilon}^2u(0,t)\]
be the first and second order variations, which satisfy the boundary condition
\[\xi(0)=\xi(\tau)=0,\qquad \zeta(0)=\zeta(\tau)=0.\]
Then we can compute
\begin{align*}
    \frac{\dd^2}{\dd\varepsilon^2}S(u(\varepsilon))\big|_{\varepsilon=0}=~&\int_0^{\tau}\left(\langle\xi_t,\xi_t\rangle_{\eta}+\langle(\gamma_{j,d})_t,\zeta_t\rangle_{\eta}\right)\,\dd t\\
    =~&\int_0^{\tau}\langle\xi_t,\xi_t\rangle_{\eta}\dd t-\int_0^{\tau}\langle(\gamma_{j,d})_{tt},\zeta\rangle_{\eta}\,\dd t\\
    =~&\int_0^{\tau}\langle\xi_t,\xi_t\rangle_{\eta}\dd t-c_{j,d}^2\int_0^{\tau}\langle\gamma_{j,d},\zeta\rangle_{\eta}\,\dd t\\
       =~&\int_0^{\tau}(\alpha_t^2+\beta_t^2-c_{j,d}^2 \alpha^2)\ \dd t-c_{j,d}^2\int_0^{\tau}\langle\gamma_{j,d},\zeta\rangle_{\eta}\,\dd t.
\end{align*}

Since $u(\varepsilon,t)\in\HH_{\C}^2$, 
\[\langle u(\varepsilon,t),u(\varepsilon,t)\rangle_{\eta}=-1.\]
Taking the derivatives with respect to $\varepsilon$ twice, we obtain
\[\langle\gamma_{j,d},\zeta\rangle_{\eta}+\langle\xi,\xi\rangle_{\eta}=0.\]
Thus for $\xi(t)=\alpha(t)T_{j,d}(t)+\beta(t)N$, we arrive at
\[\delta^2S_{\gamma_{j,d}}(\xi,\xi)=\int_0^{\tau}\left(\alpha_t^2+\beta_t^2+c_{j,d}^2\beta^2\right)\,\dd t=\int_0^{\tau}[\alpha(-\alpha_{tt})+\beta(-\beta_{tt}+c_{j,d}^2\beta)]\,\dd t.\]
This shows the formula for $L_{j,d}^{\C}$. 

We next show $\ker L_{j,d}^{\C}=\{0\}$. The tangential equation $-\alpha_{tt}=0$ with Dirichlet boundary conditions has only the trivial solution. Solutions of the normal equation $\beta_{tt}=c_{j,d}^2\beta$
take the form 
\[\beta(t)=c_1 e^{c_{j,d}t}+c_2 e^{-c_{j,d}t}.\] 
Imposing $\beta(0)=0$ gives $c_2=-c_1$, thus
\[
\beta(t)=2c_1\sinh(c_{j,d}t).
\]
Now $\beta(\tau)=0$ implies
\[
2c_1\sinh(c_{j,d}\tau)=2c_1\sinh(d+2\pi ij)=2c_1\sinh d=0.
\]
Since $d>0$, $\sinh d\ne0$, we get $c_1=0$, that is, $\beta\equiv 0$.
\end{proof}

\begin{rmk}
The normal Dirichlet spectra of $L_{j,d}^{\C}$ are
\[\left(\frac{n\pi}{\tau}\right)^2+c_{j,d}^2,\quad n\ge1.\]
At $d=0$, the normal block becomes degenerate exactly when
\[
\Bigl(\frac{n\pi}{\tau}\Bigr)^2+\Bigl(\frac{2\pi ij}{\tau}\Bigr)^2=0,
\]
that is, when $n=2|j|$. In this case 
$$
\gamma_{j,d}(t)=\bigl(\cos\left(\frac{2\pi jt}{\tau}\right),\,i \sin\left(\frac{2\pi jt}{\tau}\right),\,0\bigr)
$$
lies on some sphere in $\HH_{\C}^2$, and this is the familiar Morse--Bott phenomenon on $S^2$. This observation will play an essential role in our perturbation analysis in Section \ref{sec:case-d0}. 
\end{rmk}

\subsection{Linearized flow equation}
\label{sec:linearized-flow}

In this subsection, we will analyze linearization of the downward gradient flow equation of $F_{\theta}$ at the critical path $\gamma_{j,d}$, which plays an important role in the study of local
structure of the Lefschetz thimble in the following subsections.

To distinguish the flow variable from the path variable, we write $\|\cdot\|_{L_t^2}$ and $\|\cdot\|_{H_t^m}$ for norms on a fixed $s$-slice. Norms over an $s$-interval $I$ are written, for example, as $\|\cdot\|_{L_s^2(I;L_t^2)}$. The symbol $\|\cdot\|_H$ below denotes the weighted real Hilbert norm defined in Proposition \ref{prop:linearization}.

\begin{prop}\label{prop:linearization}
At the critical path $\gamma_{j,d}$, the linearization of
\eqref{eq:flow} is
\begin{equation*}
 \partial_s\xi+A_{j,d,\theta}\xi=0,
 \qquad
 A_{j,d,\theta}=e^{i\theta}\mathcal C
                 M_{\gamma_{j,d}}^{-1}L^{\C}_{j,d},
\end{equation*}
where $\mathcal C$ is coefficientwise complex conjugation in the frame
$(T_{j,d},N)$ and
\[
 M_{\gamma_{j,d}}(t)=\begin{pmatrix}w_d(t)&0\\0&1\end{pmatrix},
 \qquad w_d(t)=|\gamma_{j,d}(t)|^2
                    =\cosh(2dt/\tau).
\]
The operator is self-adjoint on the real Hilbert space
\[
 H=L^2((0,\tau);\C^2),\qquad
 (u,v)_H=\re\int_0^\tau\bar u^TM_{\gamma_{j,d}}v\,\dd t,
\]
with domain $D_2=H^2((0,\tau);\C^2)\cap H^1_0((0,\tau);\C^2)$.
It has compact resolvent, and for $d>0$ its kernel is zero.
\end{prop}

\begin{proof}
Set
\[
V(z):=P(\eta\bar{z}_{tt})=\eta\bar{z}_{tt}-\lambda(z)\eta\bar{z},\qquad \lambda(z):=\frac{z^T\bar{z}_{tt}}{|z|^2}.\]
For brevity, write
\[\gamma(t)=\gamma_{j,d}(t),\quad c=c_{j,d},\quad T=T_{j,d}(t),\quad \xi=\alpha T+\beta N.\]

Linearizing \eqref{eq:flow} at $\gamma=\gamma_{j,d}$ along $\xi$, we obtain
\[DV_{\gamma}(\xi)=\eta\,\bar{\xi}_{tt}-D\lambda_{\gamma}(\xi)\eta\,\bar\gamma-\bar{c}^2\,\eta\,\bar{\xi}.\]
Using $\gamma_{tt}=c^2\gamma$, we compute $D\lambda_{\gamma}(\xi)$ as follows:
\begin{align*}
    D\lambda_{\gamma}(\xi)=~&\frac{\gamma^T\bar{\xi}_{tt}+\bar{c}^2\,\xi^T \bar \gamma}{|\gamma|^2}-\frac{\bar c^2(\xi^T \bar \gamma+ \gamma^T \bar \xi)}{|\gamma|^2}=~\frac{\gamma^T\bar{\xi}_{tt}-\bar{c}^2\,\gamma^T\,\bar{\xi}}{|\gamma|^2}\\
    =~&\frac{\gamma^T(\bar{\alpha}_{tt}+\bar{c}^2\bar{\alpha})\bar{T}+2\gamma^T\bar{c}\bar{\alpha}_t\bar{\gamma}+\gamma^T\bar{\beta}_{tt}N-\bar{c}^2\bar \alpha\,\gamma^T\,\bar{T}}{|\gamma|^2}\\
    =~&\frac{\gamma^T\bar{T}}{|\gamma|^2} \bar{\alpha}_{tt}+2\bar{c}\,\bar{\alpha}_t.
\end{align*}
Using the identity
\[\cosh^2(ct)-\sinh^2(ct)=1,\]
we can compute
\[\eta\bar{T}-\frac{\gamma^T\bar{T}}{|\gamma|^2} \eta \bar \gamma=\frac{1}{|\gamma|^2}(\sinh(ct),\cosh(ct),0)=\frac{1}{|\gamma|^2}T.\]
Substituting these results into $DV_{\gamma}(\xi)$, 
we obtain
\[DV_{\gamma}(\xi)=\frac{\bar{\alpha}_{tt}}{|\gamma|^2}T+\left(\bar{\beta}_{tt}-\bar{c}^2\bar{\beta}\right)N.\]

The tangential weighted Dirichlet eigenfunctions
and the normal sine eigenfunctions give a complete real orthogonal
mode decomposition. On each complex coefficient line, the operator
has the form $b\mapsto e^{i\theta}\overline{\phi b}$, whose two real
eigenvalues are $\pm|\phi|$. Its self-adjoint domain is $D_2$ by the
one-dimensional Dirichlet elliptic estimate. The eigenvalues tend to
infinity in absolute value, proving compact resolvent. Finally,
$\ker A_{j,d,\theta}=\ker L^{\C}_{j,d}=0$ for $d>0$.
\end{proof}

\begin{prop}\label{prop:linear-spectrum}
For $d>0$, define
\[
 \mu_{j,d,*}=\min\left\{
\inf_{n\ge1}\phi^{\mathrm{tan}}_{n,d},\,\inf_{n\ge1}|\phi^{\mathrm{nor}}_{n,j,d}|\right\}>0,
\]
where $\phi^{\mathrm{tan}}_{n,d}$ are the positive eigenvalues of
\[
 -\alpha_{tt}=\phi\,w_d(t)\alpha,
 \qquad\alpha(0)=\alpha(\tau)=0,
\]
and
\[
 \phi^{\mathrm{nor}}_{n,j,d}=(n\pi/\tau)^2+c_{j,d}^2.
\]
The spectral projections $\Pi^s=1_{(0,\infty)}(A_{j,d,\theta})$ and
$\Pi^u=1_{(-\infty,0)}(A_{j,d,\theta})$ are bounded on
\[
V=(H,D_2)_{1/2,2}=H^1_0((0,\tau);\C^2)
\]
as well as on $H$ and $D_2$.
\end{prop}

\begin{proof}
Set
\[B_{j,d}=M_{\gamma_{j,d}}^{-1}L^{\C}_{j,d},\qquad B_{j,d}^{\#}=\mathcal{C}B_{j,d}\mathcal{C}.\]
Then as $\R$-linear operators,
\[A_{j,d,\theta}^2=B_{j,d}^{\#}B_{j,d},\]
and in particular, the real spectrum of $A_{j,d,\theta}$ is obtained from the complex spectrum of $B_{j,d}$ as follows:
\[\phi\in \sigma_{\C}(B_{j,d})~~\Rightarrow~~ \pm|\phi|\in \sigma_{\R}(A_{j,d,\theta}).
\]

\begin{itemize}
    \item The tangential eigenvalues of $B_{j,d}$ are the eigenvalues
$\phi^{\mathrm{tan}}_{n,j,d}$ of the weighted Dirichlet problem
\[
-\alpha_{tt}(t)=\phi^{\mathrm{tan}}_{n,j,d}
|\gamma_{j,d}(t)|^2\alpha(t)=\phi_{n,j,d}^{\tan}\cosh\left(\frac{2dt}{\tau}\right)\alpha(t),
\qquad\alpha(0)=\alpha(\tau)=0.
\]
They are independent of $j$ and positive. In particular, the minimal eigenvalue satisfies
\[\phi_{1,d}^{\tan}:=\phi_{1,j,d}^{\tan}\ge\frac{1}{\cosh(2d)}\left(\frac{\pi}{\tau}\right)^2>0.\]
\item The normal eigenvalues are
\[
\phi^{\mathrm{nor}}_{n,j,d}=
\left(\frac{n\pi}{\tau}\right)^2+c_{j,d}^2,
\qquad n\ge1.
\]
In particular, $|\phi^{\mathrm{nor}}_{n,j,d}|\rightarrow+\infty$ as $n\rightarrow\infty$. By Proposition \ref{prop:nondegenerate}, none of the normal eigenvalues vanishes when $d>0$. 
\end{itemize}

Then
\[
\operatorname{dist}\bigl(0,\sigma_{\R}(A_{j,d,\theta})\bigr)\ge \mu_{j,d,*}:=\min\left\{\inf_{n\ge1}\phi^{\mathrm{tan}}_{n,j,d},
\inf_{n\ge1}\left|\phi^{\mathrm{nor}}_{n,j,d}\right|\right\}>0.
\]

The graph norm of $A_{j,d,\theta}$, with the $H$ norm included, is equivalent to the $D_2$ norm. Spectral projections are bounded in both norms and
thus on their real interpolation space $V$.
\end{proof}

\begin{defn}
Let $A$ be a $\R$-linear closed operator on a real
Hilbert space $H$. For the equation $u_s+Au=0$, an exponential dichotomy with gap $\mu_*>0$ means a splitting $H=E^s\oplus E^u$ with bounded
spectral projections such that
\[
 \|e^{-sA}\Pi^s\|_{H\to H}\le Ce^{-\mu_*s},\qquad
 \|e^{sA}\Pi^u\|_{H\to H}\le Ce^{-\mu_*s}
 \quad(s\ge0).
\]
\end{defn}

For the self-adjoint operator $A_{j,d,\theta}$ above, we can take $C=1$ and $\mu_*=\mu_{j,d,*}$. The corresponding splitting is $V=\Pi^sV\oplus\Pi^uV$.

Choose the following local chart of the quadric $X$ around $\gamma_{j,d}$:
\[
\Phi(u)=\gamma_{j,d}\sqrt{1+u_1^2+u_2^2}
+u_1T_{j,d}+u_2N,
\]
where the square-root branch is chosen near $1$. In the corresponding homogeneous Dirichlet coordinates, the equation takes the form
\begin{equation}\label{eq:abstract-hyperbolic-end}
u_s+Au=\mathcal N(u),
\qquad\mathcal N(0)=0,
\qquad D\mathcal N(0)=0.
\end{equation}
Here
\[
\mathcal N:D_2\cap B_V(0,\rho)\longrightarrow H
\]
is smooth for some sufficiently small $\rho>0$.

The equation is linear in $u_{tt}$, at most quadratic in $u_t$, and has coefficients smooth in $(t,u)$. The standard one-dimensional Sobolev product estimates, together with
$H^1\hookrightarrow L^\infty$ and $H^2\hookrightarrow W^{1,\infty}$, give \begin{equation}\label{eq:hyperbolic-tame}
\|\mathcal N(u)\|_H\le C\|u\|_V\|u\|_{D_2},\qquad
\|D\mathcal N(u)v\|_H\le C\bigl(\|u\|_V\|v\|_{D_2}
   +\|u\|_{D_2}\|v\|_V\bigr).
\end{equation}
Higher derivatives satisfy analogous estimates, with one $D_2$ norm factor. 

For an interval $I$, set
\[
\mathcal E(I)=W^{1,2}(I;H)\cap L^2(I;D_2),\qquad
\|u\|_{\mathcal E(I)}^2
=\int_I(\|u_s\|_H^2+\|u\|_{D_2}^2)\,\dd s.
\]
Endpoint evaluation takes values in $V$: $\mathcal E(I)\hookrightarrow C(\overline I;V)$ for bounded $I$.
We also denote
\begin{equation}\label{eq:weighted-half-spaces}
\begin{aligned}
\mathcal E_a^\pm&=\{u:e^{\pm as}u\in\mathcal E(\R_\pm)\},&
\mathcal F_a^\pm&=L^2(\R_\pm,e^{\pm2as}\dd s;H),\\
\|u\|_{\mathcal E_a^\pm}^2
 &=\int_{\R_\pm}e^{\pm2as}
(\|u_s\|_H^2+\|u\|_{D_2}^2)\,\dd s.
 \end{aligned}
\end{equation}

\begin{lem}
\label{lem:mixed-regularity}
Let $K\Subset\mathbb H_{\C}^2$ be a compact coordinate set and
$I'\Subset I$. If $z(I\times[0,\tau])\subset K$ is a solution of the flow equation \eqref{eq:flow} and
\[
 \|z\|_{W_s^{1,2}(I;L_t^2)\cap L_s^2(I;H_t^2)}\le M,
\]
then
\[
\|z\|_{C^k(\overline{I'}\times[0,\tau])}
 \le C_{k,M,K,I',I},\qquad k\ge0.
\]
For the linearized operator $A$, set $\Pi^h=\Pi^s+\Pi^u$.
If the spectrum on $\Pi^hH$ has gap $\mu_*>0$, then
\[
 \|\mathcal G^\pm f\|_{\mathcal E_a^\pm}
 \le C_a\|f\|_{\mathcal F_a^\pm},\qquad
 f=\Pi^hf,\quad |a|<\mu_*,
\]
where $\mathcal G^\pm$ are the one-sided Green operators on $\R_\pm$, and the constants $C_a$ are uniform for bounded coefficients, a fixed
graph-norm equivalence and gap, and $a$ in a compact subinterval
of $(-\mu_*,\mu_*)$.
\end{lem}

\begin{proof}
The proof is given in Appendix \ref{appendix_A}.    
\end{proof}

Suppose that, in a fixed critical-path chart, a hafl-solution $u(s)$ satisfies
\[
\sup_{s\in\R_\pm}\|u(s)\|_V\le\varepsilon.
\]
For $\varepsilon$ and the parameters sufficiently small, the gradient
identity and the quadratic expansion of $F_\theta$ at the critical path give
\[
 \int_{\R_\pm}\|u_s\|_H^2\,\dd s\le C\varepsilon^2.
\]
Moreover, from $u_s+Au=\mathcal N(u)$ and the graph-norm estimate for $A$,
\[
\|u\|_{D_2}
 \le C\bigl(\|u_s\|_H+\|u\|_V\bigr)
      +C(|d|+\varepsilon)\|u\|_{D_2}.
\]
After absorbing the last term and integrating over unit intervals,
\begin{equation}\label{eq:small-half-energy}
\sup_{J\subset\R_\pm, |J|=1}\|u\|_{\mathcal E(J)}
 \le C\varepsilon.
\end{equation}
Thus every such half-solution belongs to the corresponding weight space $u\in\mathcal E_{\mp\omega}^{\pm}$, $\omega>0$.

\begin{lem}
\label{lem:lyapunov-perron}
For the fixed-gap equation \eqref{eq:abstract-hyperbolic-end}, there
are $\varepsilon,\rho>0$ and a $C^1$ map
\[
 \Psi:B_\rho(\Pi^sV)\longrightarrow\Pi^uV,
 \qquad \Psi(0)=D\Psi(0)=0,
\]
such that
\[
 W^s_{\loc}(0)=\{h+\Psi(h):h\in B_\rho(\Pi^sV)\}.
\]
Here $W^s_{\loc}$ consists of the values at $s=0$ of half-solutions with $\sup_{s\ge0}\|u(s)\|_V<\varepsilon$ and
$\Pi^su(0)\in B_\rho(\Pi^sV)$. Each such solution is uniquely
determined by $h=\Pi^su(0)$ and satisfies, for every $0<\mu<\mu_*$,
\[
 \|u(s)\|_V\le C_\mu e^{-\mu(s-s_\mu)},\qquad s\ge s_\mu.
\]
\end{lem}

\begin{proof}
This is the standard Lyapunov--Perron stable-manifold construction
for an exponential dichotomy in weighted maximal-regularity spaces. The details are given in Appendix \ref{appendix_A}.    
\end{proof}

\begin{rmk}
At a nonconstant Bott loop the gap tends to zero as $d\to0$.
Proposition \ref{prop:complete-tails} treats the resulting slow modes
separately and constructs tail charts uniform in $d$.
\end{rmk}

\subsection{Lefschetz thimble for generic $\theta$}
We now describe Lefschetz thimbles of the energy functional on $\mathcal{P}_d^{\C}$ for generic phase $\theta$ to prepare for the Picard–Lefschetz analysis. 

\begin{defn}
Fix $d>0$, $j\in\Z$, and a phase $\theta$. Let $\mathcal M_j^{\theta,+}$ consist of maps
$z:[0,\infty)\to\mathcal P_d^{\C}$, continuous in $H_t^1$ and locally
of class $\mathcal E$, satisfying
\[
\left\{
\begin{aligned}
& \partial_s z
=P_z\!\left(e^{i\theta}\eta\,\overline{z_{tt}}\right),
 \qquad s>0,\\
 &z(s,0)=a,\qquad z(s,\tau)=b_d,
 \qquad s\ge0,\\
 &\lim_{s\to+\infty}
 \|z(s,\cdot)-\gamma_{j,d}\|_{H_t^1}=0,\\
 &\int_0^\infty
 \|\partial_s z(s,\cdot)\|_{L_t^2}^2\,\dd s<\infty.
\end{aligned}
\right.
\]
The corresponding Lefschetz thimble attached to $\gamma_{j,d}$ is defined by
\begin{equation*}
\mathcal
 J_j^\theta=\operatorname{ev}_0
(\mathcal M_j^{\theta,+})=\{z(0,\cdot):z\in\mathcal{M}_j^{\theta,+}\}\subset\PP_d^{\C}.
\end{equation*}
\end{defn}

Along every such solution, the gradient identity and holomorphicity
give
\[
\frac{\dd}{\dd s}F_\theta(z(s))=-\|z_s(s,\cdot)\|_{L_t^2}^2,
\qquad \frac{\dd}{\dd s}G_\theta(z(s))=0.
\]
The action is continuous on $H^1_t$, so these identities integrate to
\begin{equation}\label{eq:thimble-action-identities}
G_\theta(z(0))=G_\theta(\gamma_{j,d}),\qquad
F_\theta(z(0))-F_\theta(\gamma_{j,d})
=\int_0^\infty\|z_s(s,\cdot)\|_{L_t^2}^2\,\dd s.
\end{equation}
We impose the genericity assumption that the twisted phase separates the imaginary critical values:
\[
 G_\theta(\gamma_{j,d})\ne G_\theta(\gamma_{k,d})
\quad(j\ne k).
\]
By \eqref{eq:thimble-action-identities}, distinct thimbles lie in
distinct $G_\theta$ level sets and are therefore disjoint. 

\begin{prop}\label{prop:expconv}
If $d>0$ and $z\in\mathcal M_j^{\theta,+}$, then for every
$0<\mu<\mu_{j,d,*}$ there are $s_0$ and $C_\mu$ such that
\[
 \|z(s,\cdot)-\gamma_{j,d}\|_{H_t^1}
                \le C_\mu e^{-\mu(s-s_0)}\quad(s\ge s_0).
\]
\end{prop}

\begin{proof}
For $s_0$ sufficiently large, the tail lies in the chart $\Phi$ and
its coordinate $u(s)$ converges to $0$ in $V$.
Lemma \ref{lem:lyapunov-perron} gives the exponential decay of $u$,
and the local chart transfers it to $z-\gamma_{j,d}$.
\end{proof}

Immediately, we have the local graph result for the thimble germ. 
\begin{cor}
\label{cor:localthimble}
Fix $j\in\mathbb Z$, $d>0$, and a phase $\theta$. Using
$D\Phi(0)$ to identify $T_{\gamma_{j,d}}\mathcal P_d^{\C}$ with $V$, we have
\[
T_{\gamma_{j,d}}\mathcal P_d^{\C}
=E^s_{j,d,\theta}\oplus E^u_{j,d,\theta},\qquad
E^s_{j,d,\theta}=\Pi^sV,\qquad
E^u_{j,d,\theta}=\Pi^uV.
\]
There exist $r,\rho>0$ and a $C^1$ map
\[
\Psi_{j,d,\theta}: B_\rho(E^s_{j,d,\theta})\longrightarrow E^u_{j,d,\theta},\qquad
\Psi_{j,d,\theta}(0)=
D\Psi_{j,d,\theta}(0)=0.
\]
Define
\[
\mathcal M_{j,\loc}^{\theta,+}
:=\{
 z\in\mathcal M_j^{\theta,+}:
 z(s)\in\Phi(B_r(V))~\text{for }s\ge0,~
 \|\Pi^s\Phi^{-1}(z(0))\|_V<\rho
\}.
\]
Then
\begin{equation*}
 \mathcal J_{j,\loc}^{\theta}
 :=
 \operatorname{ev}_0
 \bigl(\mathcal M_{j,\loc}^{\theta,+}\bigr)
 =
 \left\{
 \Phi\bigl(h+\Psi_{j,d,\theta}(h)\bigr):
 h\in B_\rho(E^s_{j,d,\theta})
 \right\}.
\end{equation*}
In particular, $\mathcal J_{j,\loc}^{\theta}$ is a $C^1$
submanifold of $\mathcal P_d^{\C}$ through $\gamma_{j,d}$, and
\[
T_{\gamma_{j,d}}\mathcal J_{j,\loc}^{\theta}=E^s_{j,d,\theta}.
\]
Moreover, for every $z\in\mathcal M_j^{\theta,+}$ there exists
$s_0\ge0$ such that the shifted half-line solution
\[
(\sigma_{s_0}z)(s):=z(s+s_0)\in \mathcal M_{j,\loc}^{\theta,+}.
\]
\end{cor}

\begin{rmk}
The spectral normal form is an
infinite-dimensional local analogue of the holomorphic Morse normal
form, with tangent space given by the positive real spectrum of
$A_{j,d,\theta}$.
\end{rmk}

\subsection{Real path space as a thimble}
\label{subsec:real-thimble}

We identify the thimble attached to the real geodesic $\gamma_{0,d}$
at phase $\theta=0$. Write
\[
\mathcal J_0:=\mathcal J_0^0
=\{z(0):z\in\mathcal M^{0,+}_0\}.
\]
Thus the starting real path space $\mathcal{P}_d^{\R}$ for the heat kernel can be viewed as one Lefschetz thimble on $\mathcal{P}_d^{\C}$. It is interesting to see whether other Lefschetz thimbles $\J_j^{\theta}$ have heat type interpretations.

\subsubsection{The real flow}

Recall that
\[
\PP_d^{\R}=
\{x\in H^1([0,\tau],\HH^2_{\R}):x(0)=a,\ x(\tau)=b_d\}.
\]
At $\theta=0$, the full equation reads
\begin{equation}
\label{eq:theta-zero-flow-corrected}
z_s=P_z(\eta\bar z_{tt})
=\eta\bar z_{tt}
-\frac{z^T\bar z_{tt}}{|z|^2}\eta\bar z.
\end{equation}
Let $h$ be the restriction of the ambient Euclidean metric to
$\HH^2_{\R}$ and let $g=\eta|_{T\HH^2_{\R}}$ be its hyperbolic metric. 

\begin{prop}
\label{prop:real-flow-parabolic}
Equation \eqref{eq:theta-zero-flow-corrected} restricts to a real
equation on $\PP_d^{\R}$, that is, the negative $L_h^2$-gradient flow of
\[
F_0(x)=\frac12\int_0^\tau g(x_t,x_t)\,\dd t.
\]
For every classical real solution,
\begin{equation}
\label{eq:real-energy-decay}
\frac{\dd}{\dd s}F_0(x(s))
=-\int_0^\tau h(x_s,x_s)\,\dd t.
\end{equation}
Its principal coefficient in real coordinates is $h^{-1}g$, which is
uniformly positive on each compact target region.
\end{prop}

\begin{proof}
For real $x$, $P_x(\eta x_{tt})$ is real and tangent.
For every endpoint-fixed real tangent variation $\xi$,
\[
\dd F_0(x)(\xi)
=-\int_0^\tau g(\nabla_t^g x_t,\xi)\,\dd t
=-\int_0^\tau h(P_x(\eta x_{tt}),\xi)\,\dd t.
\]
In local coordinates on the target, the field $\nabla_tq_t$ has principal part $q_{tt}$,
and changing the gradient metric from $g$ to $h$ multiplies it by $h^{-1}g$. The latter is is $h$-self-adjoint and positive definite. On every compact target region, its eigenvalues are bounded uniformly away from zero.
\end{proof}

\begin{lem}
\label{lem:real-flow-global}
For every $x_0\in\PP_d^{\R}$, the real equation has a unique solution
\[
x\in C^0([0,\infty);H^1_t),\qquad x(0)=x_0,
\]
which is smooth for $s>0$. It has finite energy and satisfies
\[
x(s,\cdot)\longrightarrow\gamma_{0,d}
\quad\text{in }H^1_t\text{ and smoothly as }s\rightarrow\infty.
\]
\end{lem}

\begin{proof}
Local existence and uniqueness for the solution follow from the
quasilinear Dirichlet theory; see \cite[Chapter~5]{PrussSimonett}. It gives a maximal solution
\[
 x:[0,T_{\max})\times[0,\tau]\longrightarrow \HH^2_{\R}.
\]
For every $s<T_{\max}$ and $t\in[0,\tau]$, the energy identity gives
\[
d_{\HH^2_{\mathbb R}}(x(s,t),a)\leq
\int_0^t |x_t(s,r)|_g\,\dd r\leq
\sqrt{\tau}
\left(
\int_0^\tau |x_t(s,r)|_g^2\,\dd r
\right)^{1/2}\leq
\sqrt{2\tau F_0(x_0)}.
\]
Thus the maximal solution takes values in a compact set $K$.

In local real coordinates $q=(q^1,q^2)$ on $\HH_{\R}^2$, the real equation has the form
\begin{equation}
\label{eq:real-quasilinear}
q_s^i=a^i{}_j(q)\,q_{tt}^j+Q^i(q,q_t),\qquad
a^i{}_j=(h^{-1}g)^i{}_j.
\end{equation}
To obtain the continuation bound, set $\sigma=x_t^T\eta x_t$. Since
\[
\int_0^\tau \sigma\,\dd t
=2F_0(x(s))\le 2F_0(x_0),
\qquad
\|x_t(s)\|_{L_t^2}^2
=2F_0(x(s))\le 2F_0(x_0),
\]
and
\[
\sigma_t=2x_t^Tx_s,
\qquad
\|\sigma_t\|_{L_t^1}
\le2\|x_t\|_{L_t^2}\|x_s\|_{L_t^2},
\]
the one-dimensional Sobolev estimate gives
\[
\|\sigma(s)\|_{L_t^\infty}
\le C_{F_0(x_0)}
\bigl(1+\|x_s(s)\|_{L_t^2}\bigr).
\]
The equation can be written as
\[
 x_{tt}=(\sigma+x^Tx_s)x+\eta x_s.
\]
Since $x(s,t)\in K$ for $s<T_{\max}$, then
\begin{equation}\label{eq:real-continuation-estimate}
\|x_{tt}(s)\|_{L_t^2}\le
 C_{F_0(x_0)}
\bigl(1+\|x_s(s)\|_{L_t^2}\bigr),\qquad \int_0^b\|x_s(s)\|_{L_t^2}^2\,\dd s\le F_0(x_0),~~ 0<b<T_{\max}.
\end{equation}
If $T_{\max}<\infty$, \eqref{eq:real-continuation-estimate} gives
\[
\sup_{0<b<T_{\max}}
\|x\|_{W_s^{1,2}((0,b);L_t^2)\cap L_s^2((0,b);H_t^2)}<\infty.
\]
Then the standard quasilinear continuation criterion extends the
solution beyond $T_{\max}$, a contradiction. Thus
$T_{\max}=\infty$. Now we have
\[
\int_0^\infty
\|x_s(s,\cdot)\|_{L_t^2(h)}^2\,\dd s
\leq
F_0(x_0)<\infty.
\]
Thus there exists a sequence $s_\nu\to+\infty$ such that
\[
\|x_s(s_\nu,\cdot)\|_{L_t^2(h)}\longrightarrow0.
\]
The uniform higher Sobolev bounds and compactness give, after passing to a subsequence,
\[
x(s_\nu,\cdot)\longrightarrow x_\infty(\cdot)
\]
strongly in $H_t^1$ (in fact smoothly). Passing to
the limit in the gradient equation gives
\[
\nabla^g_t(x_\infty)_t=0.
\]
Thus $x_\infty$ is a hyperbolic geodesic from $a$ to $b_d$. Since $\HH^2_{\mathbb R}$ is a Hadamard manifold, such a geodesic is unique. Thus $x_\infty=\gamma_{0,d}$.

It remains to upgrade subsequential convergence to the full limit. Since the linearization of the real flow at $\gamma_{0,d}$ is the positive
Dirichlet Jacobi operator, $\gamma_{0,d}$ is locally exponentially
stable. For $\nu$ sufficiently large, $x(s_\nu)$ lies in this stable
neighborhood, and thus
\[
 \|x(s)-\gamma_{0,d}\|_{H_t^1}
 \le
 C e^{-\mu(s-s_\nu)}
 \|x(s_\nu)-\gamma_{0,d}\|_{H_t^1},
 \qquad s\ge s_\nu.
\]
Furthermore, the uniform parabolic estimates upgrade it to smooth convergence.
\end{proof}

Proposition \ref{prop:real-flow-parabolic} and Lemma \ref{lem:real-flow-global} give 
\begin{prop}
\label{prop:real-slice-contained}
At $\theta=0$, $\PP_d^{\R}\subset\mathcal J_0$.
\end{prop}

\subsubsection{Conjugation and small forward tails}

Complex conjugation $\mathcal C(z)=\bar z$ preserves
\eqref{eq:theta-zero-flow-corrected} and fixes $\PP_d^{\R}$.
At $\gamma_{0,d}$, set
\[
B_d=M_{\gamma_{0,d}}^{-1}L^{\C}_{0,d}.
\]
This is a real positive self-adjoint Jacobi operator for the
$L_h^2$ metric, with Dirichlet domain and a positive spectral gap. Writing $\xi=u+iv$, the linearized equation is
\[
u_s=-B_du,\qquad v_s=B_dv.
\]
In particular, the forward stable and unstable spaces are given by
\[
E^s=T_{\gamma_{0,d}}\PP_d^{\R},\qquad
E^u=iT_{\gamma_{0,d}}\PP_d^{\R}.\] 

\begin{prop}
\label{prop:local-real-stable}
At phase $\theta=0$, the local thimble at $\gamma_{0,d}$ satisfies
\begin{equation}
\label{eq:local-real-thimble}
\mathcal J_{0,\loc}^{0}=\Phi\bigl(B_\rho(E^s)\bigr)
=\Phi\bigl(B_\rho(T_{\gamma_{0,d}}\mathcal P_d^{\R})
\bigr)\subset \mathcal P_d^{\R},
\end{equation}
where $\Phi$ is chosen conjugation-equivariant. Then, for every $z\in\mathcal M_0^{0,+}$, there exists $S>0$ such that
\[
z(s,\cdot)\in\mathcal P_d^{\R}
\qquad\text{for all }s\ge S.
\]
\end{prop}

\begin{proof}
By Corollary~\ref{cor:localthimble}, in a sufficiently small neighborhood of $\gamma_{0,d}$ the local thimble is the graph
\[
\mathcal J_{0,\loc}^{0}
=\left\{\Phi\bigl(u+\Psi(u)\bigr):
 u\in B_\rho(E^s)\right\},
\qquad\Psi:B_\rho(E^s)\longrightarrow E^u.
\]
Choose $\Phi$ to be equivariant under complex conjugation. Conjugation fixes $E^s$, acts as $-1$ on $E^u$, and preserves the class of local forward half-solutions. Thus, for each $u\in B_\rho(E^s)$, both
\[
\Phi\bigl(u+\Psi(u)\bigr)
\quad\text{and}\quad
\Phi\bigl(u-\Psi(u)\bigr)
\]
belong to the local stable graph with the same $E^s$-component. The uniqueness part of Corollary \ref{cor:localthimble} therefore implies
\[
\Psi(u)=-\Psi(u),
\]
and thus $\Psi\equiv0$. This proves
\eqref{eq:local-real-thimble}.

Finally, if $z\in\mathcal M_0^{0,+}$, then for $S$ sufficiently large, its shifted tail lies in the local thimble. By
\eqref{eq:local-real-thimble},
\[
z(s,\cdot)\in\mathcal P_d^{\R},
\qquad s\ge S.\qedhere
\]
\end{proof}

\subsubsection{Anti-real normal equation and backward rigidity}

We use the anti-real equation to propagate the real tail backwards.

\begin{lem}
\label{lem:w-equation}
For a solution of \eqref{eq:theta-zero-flow-corrected}, set
\[
\alpha(z)=\frac{z^T\bar z_{tt}}{|z|^2},\qquad w=z-\bar z.
\]
Then
\begin{equation}
\label{eq:w-exact}
w_s=\mathcal K(z)w_{tt}+\mathcal L(z)w,
\qquad
\mathcal K(z)=-\eta+\frac{\eta\bar z\,\bar z^T}{|z|^2},
\qquad
\mathcal L(z)=\bar\alpha\eta
-\frac{\eta\bar z\,\bar z_{tt}^{\,T}}{|z|^2},
\end{equation}
with zero Dirichlet boundary condition. At a real point $x$, the restriction of
$-\mathcal K(x)$ to $T_x\HH^2_{\R}$ is $h_x^{-1}g_x$.
\end{lem}

\begin{proof}
Subtracting the conjugate equation gives
\[
w_s=-\eta w_{tt}+\bar\alpha\eta w
+(\bar\alpha-\alpha)\eta\bar z,
\qquad
\bar\alpha-\alpha
=\frac{\bar z^Tw_{tt}-w^T\bar z_{tt}}{|z|^2},
\]
which proves \eqref{eq:w-exact}. If $\xi,\zeta$ are real tangent
vectors at $x$, then
\[
-\mathcal K(x)\xi=P_x(\eta\xi),\qquad
h_x(-\mathcal K(x)\xi,\zeta)=g_x(\xi,\zeta).\qedhere
\]
\end{proof}

\begin{lem}
\label{lem:local-backward-rigidity}
Let
\[
z:[S-\varepsilon,S]\times[0,\tau]\longrightarrow \mathbb{H}^2_{\mathbb C}
\]
be a regular solution of the $\theta=0$ flow. Assume
that its image lies in a sufficiently small complex neighborhood of a compact subset of the real slice and that
\[
z(S,\cdot)\in\mathcal P^{\R}_d.
\]
Then
\[
z(s,\cdot)\in\mathcal P^{\R}_d
\qquad\text{for every }s\in[S-\varepsilon,S].
\]
\end{lem}

\begin{proof}
Write $z=x+iy$, $x,y\in\R^3$. The constraint $z^T\eta z=-1$ gives
\[
x^T\eta y=0, \qquad x^T\eta x-y^T\eta y=-1.
\]
After shrinking the complex neighborhood, there exists $c_0>0$ such that $x^T\eta x\le -c_0$ on the strip $[S-\varepsilon,S]\times[0,\tau]$. Thus
\[
E_x:=\{\xi\in\R^3:x^T\eta\xi=0\}
\]
is a real rank-two bundle on which the Lorentz form is positive definite.

Since the strip is contractible, choose a smooth Euclidean-orthonormal frame $e_1,e_2$ of $E_x$. As $y\in E_x$, write
\[
y=v^1e_1+v^2e_2,\qquad v=(v^1,v^2)^T.
\]
Equivalently, $w=z-\bar z=2iy$. Substituting this expression into the exact equation
\eqref{eq:w-exact}, the derivatives of the moving frame contribute only first- and zeroth-order terms. Thus
\begin{equation}
\label{eq:v-backward-parabolic}
v_s=-\mathcal A(s,t)v_{tt}+B_1(s,t)v_t+B_0(s,t)v.
\end{equation}

Along the real slice, Lemma \ref{lem:w-equation} identifies the principal matrix with the matrix of $h_x^{-1}g_x$, which is positive definite. By compactness and continuity, the complex neighborhood may be chosen so that
\begin{equation}
\label{eq:A-uniform-positive}
 \left\langle
 \frac{\mathcal A+\mathcal A^T}{2}\xi,\xi
 \right\rangle
 \ge c|\xi|^2
\end{equation}
for some $c>0$ throughout the strip. The matrices $B_0$ and $B_1$ are bounded.

Reverse the $s$-variable by setting
\[
r=S-s,\qquad V(r,t)=v(S-r,t).
\]
Then \eqref{eq:v-backward-parabolic} becomes
\[
V_r-\mathcal A(S-r,t)V_{tt}=-B_1(S-r,t)V_t-B_0(S-r,t)V.
\]
By \eqref{eq:A-uniform-positive}, this is a forward uniformly
parabolic linear system. Since $z(S,\cdot)$, $a, b_d$ are real, we have
\[
V(0,\cdot)=0, \qquad V(r,0)=V(r,\tau)=0.
\]
Standard uniqueness for the forward Dirichlet problem gives $V\equiv0$. Then $y\equiv0$, and thus $z=\bar z$ on the whole strip.
\end{proof}

\begin{thm}
\label{thm:real-path-space-thimble}
For every $d>0$, at phase $\theta=0$,
\[
\mathcal J_0=\PP_d^{\R}.
\]
Thus the thimble attached to the real geodesic consists precisely
of the real endpoint-fixed $H^1$ paths.
\end{thm}

\begin{proof}
Proposition \ref{prop:real-slice-contained} gives $\PP_d^{\R}\subset\mathcal J_0$. For the opposite inclusion, let $z\in\mathcal M^{0,+}_0$. By
Proposition \ref{prop:local-real-stable}, it is real for $s\geq S$. Fix $s_0\in(0,S)$ and define
\[
\sigma_*:=\inf\left\{\sigma\in[s_0,S]:
z(s,\cdot)\in\mathcal P^{\R}_d
\text{ for every }s\in[\sigma,S]
\right\}.
\]
Continuity and closedness of the real slice show that
$z(\sigma_*,\cdot)$ is real. Suppose that $\sigma_*>s_0$, since $H^1$ continuity and $H^1_t\hookrightarrow C^0_t$, there exists
$\varepsilon>0$ such that
\[
z([\sigma_*-\varepsilon,\sigma_*]\times[0,\tau])
\]
lies in a sufficiently small complex neighborhood of the compact real path $z(\sigma_*,\cdot)$. Applying Lemma \ref{lem:local-backward-rigidity} on that strip contradicts the definition of $\sigma_*$. Thus $z$ is real for every $s>0$. Finally, its $H^1$ continuity at $s=0$ gives $z(0)\in\PP_d^{\R}$.
\end{proof}

\subsection{Prediction on two connecting trajectories}
\label{sec:connecting-trajectories}
We now start to consider the prediction (Conjecture
\ref{conj:alien}) on the connecting trajectories from resurgence analysis.
By Proposition \ref{prop:H2-pointed-alien-stokes-constant}, the Alien coefficient
2 predicts that there should be exactly two positive-signed solutions of downward Morse flow connecting $\gamma_{k,d}$ $(k\ne0)$
to $\gamma_{0,d}$ at the Stokes phase. We pick $k=1$ and verify this prediction for $d$ sufficiently small. This
provides a nontrivial test of Conjecture \ref{conj:alien}.

\vskip 0.1cm
The Stokes phase $\theta(d)$ is determined by
\[
G_{\theta(d)}(\gamma_{1,d})=G_{\theta(d)}(\gamma_{0,d}),\qquad
F_{\theta(d)}(\gamma_{1,d})>F_{\theta(d)}(\gamma_{0,d}),
\]
and thus
\begin{equation*}
 \theta(d)=\pi-\arctan\frac d\pi.
\end{equation*}
In this subsection and below, $\theta$ denotes $\theta(d)$ when no
confusion can arise. Set
\[
 V_d(z)=P_z\bigl(e^{i\theta(d)}\eta\overline{z_{tt}}\bigr),
 \qquad
 \mathcal F_d(z)=z_s-V_d(z).
\]
Let
\[
 \mathcal M_d=\mathcal M_d(a,b_d;\gamma_{0,d},\gamma_{1,d})
\]
be the set, modulo translation in $s$, of finite-energy classical solutions of $\mathcal F_d(z)=0$ with
\[
z(-\infty,\cdot)=\gamma_{1,d},\qquad
z(+\infty,\cdot)=\gamma_{0,d} \quad\text{in }
H^1([0,\tau]).\]
For every such solution, we have
\begin{equation}\label{eq:action}
\frac{\dd}{\dd s}S(z(s))=-e^{i\theta(d)}\|z_s(s)\|_{L_t^2}^2,\qquad\int_{\R}\|z_s\|_{L_t^2}^2\,\dd s
=\frac{2\pi\sqrt{\pi^2+d^2}}{\tau}.
\end{equation}

The proof has four parts: the exact count at $d=0$, continuation of the two resulting trajectories for $d>0$, a compactness argument excluding all additional trajectories and sign computation.

\subsection{The case: \texorpdfstring{$d=0$}{d=0}}
\label{sec:case-d0}

In this subsection, we prove the count
\[
 \#\mathcal M_0=\#\mathcal M_0(a,a;\gamma_{0,0},\gamma_{1,0})=2.
\]
The proof has two ingredients. On the invariant sphere slice
$\Sigma\subset\HH_{\C}^2$, the flow reduces to the based-loop heat
flow on $S^2$, and the once-round Bott loop has exactly two strong
unstable branches, both converging to the constant loop. Conversely,
every $d=0$ connector enters $\Sigma$ near its positive end, and then
backward uniqueness forces the whole trajectory to lie in
$\Sigma$. 

\vskip 0.1cm
At $d=0$, the once-round critical loops form the Bott family
\begin{equation}
 X_w(t)=(\cos(2\pi t/\tau),i\cos w\sin(2\pi t/\tau),i\sin w\sin(2\pi t/\tau))\qquad w\in\C/(2\pi\mathbb Z).
 \label{eq:bott-loops}
\end{equation}
Now the upper loop is $X_0=\gamma_{1,0}$. The lower critical path $\gamma_{0,0}=a$ is the constant loop, which has a fixed spectral gap, Lemma \ref{lem:lyapunov-perron} and parabolic regularity upgrade
the prescribed $H_t^1$ convergence to
\[
z(s,\cdot)\longrightarrow a\quad\text{in }C^\infty([0,\tau])
\quad\text{as }s\to+\infty.
\]

\begin{thm}
\label{thm:exact-two-d0}
Every bounded $d=0$ connecting trajectory ending at the constant loop lies in
\[
\Sigma=\{(q_0,-iq_1,-iq_2):q\in S^2\}.
\]
For the upper limit $\gamma_{1,0}$, there are exactly two translation classes, denoted $[v_{0,+}]$ and $[v_{0,-}]$. If the upper once-round loop is unspecified, it belongs to the real Bott circle
$w\in\R/(2\pi\mathbb Z)$ in \eqref{eq:bott-loops}.
\end{thm}

\begin{proof}
Write
\[
z_0=q_0+ip_0,\qquad z_1=p_1-iq_1, \qquad z_2=p_2-iq_2.
\]
Then
\[
q^2-p^2=1,\qquad q\cdot p=0,
\]
and the flow takes the form
\begin{equation}
\label{eq:qpequation}
q_s=q_{tt}+\lambda q+\mu p,
\qquad p_s=-p_{tt}-\lambda p+\mu q,
\end{equation}
where
\[
\lambda=-\frac{q\cdot q_{tt}+p\cdot p_{tt}}{|q|^2+|p|^2},\qquad
\mu=\frac{q\cdot p_{tt}-p\cdot q_{tt}}{|q|^2+|p|^2},\qquad
|q|^2+|p|^2\geq 1.
\]
Since $z(s)\to a$ smoothly as $s\to+\infty$, we have
$\lambda(s,\cdot)\to0$. Therefore, on a sufficiently late positive tail, Poincar\'e's inequality gives
\[
\frac{\dd}{\dd s}\frac12\int_0^\tau |p|^2\,\dd t
=\int_0^\tau |p_t|^2\,\dd t-\int_0^\tau\lambda|p|^2\,\dd t
\ge c\int_0^\tau |p|^2\,\dd t.
\]
Because $p(s)\to0$, this forces $p=0$ on that tail. Differentiating $q\cdot p=0$
twice in $t$ gives
\[
\mu=-\frac{2}{|q|^2+|p|^2}(q_t\cdot p_t+q_{tt}\cdot p),
\]
so along the given smooth solution
\begin{equation}\label{eq:normal-rigidity-linear}
 p_s=-p_{tt}+B_1p_t+B_0p,
 \qquad B_1=-\frac{2q\cdot q_t}{|q|^2+|p|^2},\qquad
 B_0=-\lambda\Id-\frac{2q \cdot q_{tt}}{|q|^2+|p|^2}.
\end{equation}
On each finite strip, these are bounded smooth coefficients. 
Choose $S$ sufficiently large such that
\[p(s,\cdot)\equiv0,\qquad s\ge S,\]
and set $P(r,t)=p(S-r,t)$.
Then \eqref{eq:normal-rigidity-linear} becomes
\[P_r=P_{tt}-B_1(S-r,t)P_t-B_0(S-r,t)P,\] with zero initial
and Dirichlet data. Integration by parts and Young's inequality
give, on every finite reversed strip,
\[
 \frac12\partial_r\|P\|_2^2+\frac12\|P_t\|_2^2
 \le C\|P\|_2^2.
\]
Gronwall's inequality yields $P=0$. Since the reversed strip was
arbitrary, $p\equiv0$, and the whole connector lies in
\[
 \Sigma=\{(q_0,-iq_1,-iq_2):q\in S^2\}.
\]

On $\Sigma$, the equation reduces to the based-loop heat flow on
$S^2$,
\begin{equation}
\label{eq:sphere-heat-flow}
 q_s=\nabla_tq_t=q_{tt}+|q_t|^2q,
 \qquad
 q(s,0)=q(s,\tau)=a.
\end{equation}
Set
\[
 E_{S^2}(q)
 :=
 \frac12\int_0^\tau |q_t|^2\,\dd t.
\]
Let $q_*$ denote the prescribed once-round geodesic and let
$q_0(t)\equiv a$. Then
\[
 E_{S^2}(q_*)=\frac{2\pi^2}{\tau},
 \qquad
 E_{S^2}(q_0)=0.
\]

At $q_*$, the tangential and normal Hessian eigenvalues are
\[
\lambda_n^{\mathrm{tan}}=\frac{n^2\pi^2}{\tau^2},
 \qquad
\lambda_n^{\mathrm{nor}}=\frac{(n^2-4)\pi^2}{\tau^2},
 \qquad n\ge1.
\]
The normal mode $n=2$ is tangent to the Bott family. After fixing the Bott point, the mode $n=1$ is the unique negative direction. Then the strong unstable fibre $W^{uu}_{\loc}(q_*)$ is one-dimensional. Its punctured fibre has two connected components,
denoted by
\[
 W^{uu}_{\loc,+}(q_*),
 \qquad
 W^{uu}_{\loc,-}(q_*).
\]
By the strong unstable-manifold theorem, these components determine
two solutions $q_+$ and $q_-$, unique up to translation in $s$, such
that
\begin{equation*}
 \|q_\pm(s)-q_*\|_{H_t^1}
 \le Ce^{\delta s},
 \qquad s\le0,\quad\text{for some } \delta>0.
\end{equation*}

For the positive end, start at any finite point on $q=q_{\pm}$ and set $e=|q_t|^2/2$. The sphere heat equation gives
\begin{equation}\label{eq:sphere-energy-density}
(\partial_s-\partial_t^2)e=-|q_s|^2\le0,\qquad
 e_t(s,0)=e_t(s,\tau)=0.
\end{equation}
Thus the maximum principle yields
\[
\|q_t(s)\|_{L_t^\infty}
\le\|q_t(s_0)\|_{L_t^\infty},\qquad s\ge s_0.
\]
This uniform $C^1$ bound, together with the standard Dirichlet
parabolic continuation and uniqueness theory, gives a unique smooth
extension for all $s\ge s_0$. On $q=q_+$ or $q=q_-$,
\[
 E_{S^2}(q(s_0))<\frac{2\pi^2}{\tau},\qquad
 \partial_sE_{S^2}(q)=-\|q_s\|_2^2.
\]
The derivative bound and parabolic smoothing give compactness on
shifted strips. Every positive-end limit is stationary, while $E_{S^2}(q_\pm(s))<2\pi^2/\tau$ excludes every nonconstant based geodesic. Then
\[
 q_\pm(s)\longrightarrow a
 \qquad\text{smoothly as }s\to+\infty,
\]
thus the two branches $q_{\pm}$ are finite-energy connectors.
\end{proof}

\subsection{At least two connectors for \texorpdfstring{$0<d\ll1$}{small positive d}}
\label{sec:at-least-two}

\begin{thm}
\label{thm:at-least-two}
For sufficiently small $d>0$, the two $d=0$ classes $[v_{0,+}]$, $[v_{0,-}]$ continue to distinct classes $[v_{d,+}],[v_{d,-}]\in\mathcal M_d$, which depend $C^1$ on $d\ge0$.
\end{thm}

Let us start with a simple observation. 

\begin{lem}\label{lem:reflection} If $\mathcal{M}_d\ne\emptyset$, then the reflection $\mathcal R(z_0,z_1,z_2)=(z_0,z_1,-z_2)$ induces a
free involution on $\mathcal M_d$.
\end{lem}

\begin{proof}
It is easy to check that $\mathcal{R}$ preserves the quadratic form, Hermitian metric,
endpoints and critical paths, and
\[P_{\mathcal Rz}\mathcal R=\mathcal RP_z.\] 
Thus it preserves the
flow and its energy. If $\mathcal Rz(s)=z(s+c)$, then $z(s)=z(s+2c)$.
Strict decrease of $F_{\theta(d)}$ gives $c=0$, thus
$z_2\equiv0$. But then $z_0+z_1$ gives an endpoint-fixed homotopy in $\C^*$ between
\[
e^{(d+2\pi i)t/\tau}
\qquad\text{and}\qquad
e^{dt/\tau}.
\]
These paths have relative winding numbers $1$ and $0$, respectively,
a contradiction.
\end{proof}

By this lemma, to prove Theorem \ref{thm:at-least-two}, we only need to consider a continuation near $v_{0,+}$. We parametrize the two tails and solve the finite-strip matching problem after adding one auxiliary forcing parameter. The conservation of $G_{\theta(d)}$ shows that this parameter vanishes, so the matched solution solves the original flow.

\subsubsection{Local half-solutions}
\label{sec:tail-charts}

To fix the endpoints, set
\[
B_\rho=\begin{pmatrix}\cosh\rho&\sinh\rho&0\\
\sinh\rho&\cosh\rho&0\\0&0&1\end{pmatrix},
 \qquad \mathcal T_d(t)=B_{dt/\tau}.
\]
Then $\mathcal T_d\gamma_{j,0}=\gamma_{j,d}$. At the upper end, we write
\[
\gamma=\gamma_{1,0},\qquad
 E_1(t)=(i\sin(2\pi t/\tau),\cos(2\pi t/\tau),0),\qquad N=(0,0,1),
\]
and use the chart from Subsection~\ref{sec:linearized-flow},
\begin{equation*}
\Phi(u)=\gamma\sqrt{1+u_1^2+u_2^2}+u_1E_1+u_2N,
 \qquad z=\mathcal T_d\Phi(u).
\end{equation*}
We use the space $H=L^2((0,\tau);\C^2)$ as Subsection \ref{sec:linearized-flow} with the $d=0$
inner product. In these coordinates, the equation is
\begin{equation}\label{eq:local}
 u_s+A_0u=\mathcal N_d(u),\qquad
 A_0=A_{1,0,\pi},\qquad \mathcal N_d(0)=0,
\end{equation}
where
\[
 A_0(u_1,u_2)=
 (\overline{u_{1,tt}},\overline{u_{2,tt}}+\frac{4\pi^2}{\tau^2}\overline{u_2}).
\]
Let $E^c=\ker A_0=\{(0,b\sin(2\pi t/\tau)):b\in\C\}$ and set
\[
 \Pi^c=\mathbf1_{\{0\}}(A_0),\qquad
 \Pi^s=\mathbf1_{(0,\infty)}(A_0),\qquad
 \Pi^u=\mathbf1_{(-\infty,0)}(A_0),\qquad
 \Pi^h=\Id-\Pi^c.
\]
The restriction to $\Pi^hH$ has a fixed spectral gap. For $d>0$,
the two centre eigenvalues split into
\[
 \pm\lambda_{\mathrm{soft}}(d),\qquad
 \lambda_{\mathrm{soft}}(d)
 =\frac{d\sqrt{d^2+16\pi^2}}{\tau^2},
\]
while the remaining spectrum stays uniformly away from zero.

\begin{prop}
\label{prop:center-fibres}
For sufficiently small $|d|$, the following hold.
\begin{itemize}
\item[(i)] Equation~\eqref{eq:local} has a locally invariant centre graph
\[
 \Gamma_d(c)=c+\psi_d(c),
 \qquad
 \Pi^c\psi_d(c)=0,
 \qquad c\in E^c,
\]
which is $C^2$ in $(d,c)$ with values in $D_2$. Moreover, $\Gamma_0$ is the local Bott manifold.

\item[(ii)] Let $U_{d,c}$ be the centre orbit with
\[
U_{d,c}(0)=\Gamma_d(c).
\]
Every sufficiently small negative half-solution has a unique representation
\[
u=U_{d,c}+v_d(c,h),\qquad
v_d(c,h)\in\mathcal E_\kappa^-,\qquad
\Pi^u v_d(c,h)(0)=h\in\Pi^uV.
\]
\item[(iii)] The endpoint map
\[
F_d(c,h):=
\Gamma_d(c)+v_d(c,h)(0)
\]
is $C^1$ and gives local coordinates on the space of small negative
half-solutions. Its inverse is uniformly bounded for small $d$.
\end{itemize}
\end{prop}

\begin{proof}
This is the standard parameter-dependent centre-manifold and
strong-fibre construction; see the details in Appendix \ref{appendix_A}.
\end{proof}

We now select the centre orbits converging to $\gamma_{1,d}$ as
$s\to-\infty$. Use the local coordinate $w$ on $\Gamma_d$ inherited
from the Bott family \eqref{eq:bott-loops}, with $w=0$ corresponding
to $\gamma_{1,d}$ in the original variables.

Let $\mathcal V_d$ denote the reduced vector field on $\Gamma_d$, so that
\[
w_s=\mathcal V_d(w).
\]
The family $\mathcal V_d$ depends $C^1$ on $d$, and
$\mathcal V_0\equiv0$ because $\Gamma_0$ is the Bott manifold of
equilibria. Then
\[
\mathcal V_d(w)=d\,\mathcal G_d(w),
\qquad
\mathcal G_d(w):=
\int_0^1\partial_d\mathcal V_{\lambda d}(w)\,\dd\lambda.
\]
For $d>0$, the time change $\rho=ds$ gives
\[
 w_\rho=\mathcal G_d(w).
\]

To compute $\mathcal G_0$, let
\[
 K_{\mathrm{rot}}
 =
 \begin{pmatrix}
  0&0&0\\
  0&0&-1\\
  0&1&0
 \end{pmatrix}.
\]
Pairing the original Morse equation \eqref{eq:flow}
with $K_{\mathrm{rot}}z$ and integrating by parts gives
\begin{equation}\label{eq:rotation-current}
\int_0^\tau z_s^T\overline{K_{\mathrm{rot}}z}\,\dd t=
e^{i\theta(d)}\sinh d\,\overline{z_{2,t}(\tau)}.
\end{equation}
Using $K_{\mathrm{rot}}X_w=\partial_wX_w$, $\theta(0)=\pi$,
and the $C^1(H_t^2)$ expansion of the centre graph as
$X_w+O(d)$, the first-order terms give
\[
\mathcal G_0(w)=
-\frac{\overline{(X_w)_{2,t}(\tau)}}{\|\partial_wX_w\|_{L_t^2}^2}
=i\frac{4\pi}{\tau^2}
\frac{\sin\bar w}{\cosh(2\operatorname{Im}w)}.
\]

The eigenvalues of $D\mathcal G_0(0)$ are $\pm4\pi/\tau^2$. Thus the finite-dimensional unstable-manifold theorem gives a $C^1$ family of local unstable curves
\[
w=b_d(\beta),\qquad |\beta|<\beta_0,\qquad b_d(0)=0,
\]
with $\beta_0$ independent of small $d$. For $d>0$, these are exactly the local initial values of centre orbits converging to $0$ as $s\to-\infty$.

Choose the parameter orientation so that
\[
 b_0'(0)=-e^{i\pi/4}.
\]
Then the corresponding upper soft direction is
\[
\left.\partial_\beta X_{b_0(\beta)}\right|_{\beta=0}=
e^{-i\pi/4}\sin(2\pi t/\tau)N=:\sigma_*^u.
\]

\begin{prop}
\label{prop:complete-tails}
Choose $S>0$ so that $v_{0,+}$ lies in the upper and lower tail charts for $s\le -S$ and $s\ge S$. Let $\Pi_+^s$ be the stable spectral projection at $\gamma_{0,0}$. For  sufficiently small
$d\ge0$, there are $C^1$ maps
\[
\mathfrak T_d^-:
B_\rho(\R)\times B_\rho(\Pi^uV)\longrightarrow V,
\qquad \mathfrak T_d^+:
B_\rho(\Pi_+^sV)\longrightarrow V,
\]
with the following properties.
\begin{itemize}
    \item[(i)] For $d>0$, $\mathfrak T_d^-(\beta,h^-)$ gives the value at $s=-S$
of the unique small upper half-solution converging to $\gamma_{1,d}$,
while $\mathfrak T_d^+(h^+)$ gives the value at $s=S$ of the unique
small lower half-solution converging to $\gamma_{0,d}$. 
\item[(ii)] Every sufficiently small half-solution is represented in this way, with parameters uniformly controlled as $d\to0^+$.
    \item[(iii)] At $d=0$, the $\beta$-direction
at the negative end has asymptotic coefficient $\sigma_*^u$.
\end{itemize}
\end{prop}

\begin{proof}
For the upper tail, choose the centre orbit through $b_d(\beta)$ and attach its strong unstable fibre with parameter $h^-\in\Pi^uV$. Evaluating at $s=-S$ gives
$\mathfrak T_d^-(\beta,h^-)$.

At the lower end, the fixed-gap stable-manifold construction (see; details in Appendix \ref{appendix_A}) of
Lemma \ref{lem:lyapunov-perron}, evaluated at $s=S$, gives
$\mathfrak T_d^+(h^+)$. The completeness, $C^1$ dependence, and uniform inverse bounds (see; details in Appendix \ref{appendix_A}) follow from Proposition \ref{prop:center-fibres} and the fixed-gap construction. The normalization of $b_0$ gives the coefficient $\sigma_*^u$.
\end{proof}

\begin{rmk}
At $d=0$, the upper half-solution with parameter $\beta$ converges to $X_{b_0(\beta)}$ as $s\to-\infty$. 
\end{rmk}

\subsubsection{The augmented matching problem}
\label{sec:augmented-matching}

Fix $v=v_{0,+}$ and set
\[
T=v_s,\qquad Z=\Krot v,\qquad
\phi_B=\sin(2\pi t/\tau)N,\qquad \sigma_B=i\phi_B.
\]
Let $D$ be the real-linear linearization along $v$, and let $J$ denote
multiplication by $i$. Then
\[
DT=DZ=0,\qquad T(\pm\infty)=0,\qquad
Z(-\infty)=\sigma_B,\qquad Z(+\infty)=0.
\]
All pairings below are real Hermitian $L^2$ pairings. Choose $0<\delta<\kappa$ the nonzero asymptotic rates and set
\[
\mathcal X_\delta=
W_\delta^{1,2}(\R;H)\cap L_\delta^2(\R;D_2),
\qquad\mathcal Y_\delta
=L_\delta^2(\R;H),
\]
with weight $e^{\delta|s|}$.

\begin{lem}\label{lem:fredholm-matching}
The operator $D:\cX_\delta\to\cY_\delta$ is Fredholm, with
\begin{equation*}
\ker D=\R T,\qquad
\coker D=\Span_{\R}\{JT,JZ\},\qquad \ind D=-1.
\end{equation*}
Adjoining the upper asymptotic direction $\sigma_*^u$ removes the $JZ$ obstruction.
\end{lem}

\begin{proof}
The sphere slice $\Sigma$ is totally geodesic, and the Hessian of $F_\pi$ anticommutes with $J$. After identifying
$N\Sigma=J(T\Sigma)$, the linearization splits as
\[
D^\top=\nabla_s+H(s),
\qquad
D^\perp=\nabla_s-H(s),
\qquad
D^*J=-JD,
\]
where
\[
H(s)\zeta=
-\nabla_t^2\zeta-|v_t|^2\zeta
+\langle v_t,\zeta\rangle v_t .
\]
At the two ends,
\[
H_-:=\lim_{s\to-\infty}H(s)=
\diag\!\left(-\partial_t^2,
-\partial_t^2-\frac{4\pi^2}{\tau^2}\right),
\qquad H_+:=\lim_{s\to+\infty}H(s)=\diag(-\partial_t^2,-\partial_t^2).
\]
\textbf{(1) We first compute the index.} For a scalar block
$\partial_s+h(s)$, conjugation by the weight $e^{\delta|s|}$ changes the two end values to
$h_-+\delta$, $h_+-\delta$. Then for the second component of $H(s)$, the Dirichlet sine modes give
\[
\begin{array}{c|cc|cc}
 &\multicolumn{2}{c|}{D^\top}
 &\multicolumn{2}{c}{D^\perp}\\
 n&\text{end signs}&\text{index}
  &\text{end signs}&\text{index}\\ \hline
 1&(-,+)&1 &(+,-)&-1\\
 2&(+,+)&0 &(+,-)&-1\\
 n\ge3&(+,+)&0 &(-,-)&0 .
\end{array}
\]
The first component has index zero.

Choose a smooth function $\chi_-$ equal to $1$ near $-\infty$ and $0$ near $+\infty$, and set
\[
H_{\rm ref}(s)=
\diag\!\left(-\partial_t^2,
-\partial_t^2-\frac{4\pi^2}{\tau^2}\chi_-(s)\right).
\]
Let
\[
D_{\rm ref}=
(\partial_s+H_{\rm ref})
\oplus
(\partial_s-H_{\rm ref})
\]
under the above splitting. The mode count gives
\[
\ind D_{\rm ref}=1-2=-1.
\]
Moreover, $D-D_{\rm ref}$ has order at most one in $t$, its coefficients tend to zero at both ends, and it is relatively compact
between the weighted spaces. Thus
\[
\ind D=-1.
\]
\textbf{(2) Next we determine the kernel.} For $D^\top$, the only admissible decaying mode at the negative end is the first normal-to-geodesic sine mode. The half-line Green estimate therefore gives
\[
\dim\ker D^\top\le1.
\]
We also know that the translation field $T=v_s$ satisfies $DT=0$ and $T\neq0$, thus
\[
\ker D^\top=\R T.
\]

For the normal block, let $D^\perp\zeta=0$. On a sufficiently late positive tail,
\[
\frac12\partial_s\|\zeta\|_2^2=\langle H(s)\zeta,\zeta\rangle_t
 \ge c\|\zeta\|_{H_t^1}^2.
\]
Weighted decay gives $\zeta=0$ on that tail. The same time-reversed uniqueness argument used in the proof of Theorem \ref{thm:exact-two-d0} then gives $\zeta\equiv0$.
Thus,
\[
\ker D^\perp=0.
\]
\textbf{(3) Now we identify the cokernel.} Here $D^*$ denotes the formal adjoint
with respect to the unweighted real $L^2$ pairing; the corresponding
adjoint fields lie in the dual weighted space $L^2_{-\delta}$.
Since
\[
D^*J=-JD\qquad\text{and}\qquad DT=DZ=0,
\]
both $JT$ and $JZ$ lie in the adjoint kernel. They are linearly independent. Since $\ind D=-1$ and $\dim\ker D=1$, the cokernel has dimension two. Therefore
\[
\ker D=\R T,\qquad
\coker D=\Span_{\R}\{JZ,JT\}.
\]
\textbf{(4) Adjoining $\sigma_*^u$ removes the $JZ$ obstruction.} Finally, choose a field $\zeta$ with
\[
\zeta(-\infty)=\sigma_*^u,
\qquad
\zeta(+\infty)=0.
\]
Integration by parts against the two adjoint kernel fields gives
\begin{equation}\label{eq:TB-obstruction}
\langle D\zeta,JT\rangle_{s,t}=0,
\qquad\langle D\zeta,JZ\rangle_{s,t}
=-\langle\sigma_*^u,J\sigma_B\rangle_t.
\end{equation}
We can further compute
\[
-\langle\sigma_*^u,J\sigma_B\rangle_t=-\langle e^{-i\pi/4\phi_B,i\phi_B}\rangle=2^{-1/2}\|\phi_B\|_2^2>0.
\]
Then we are done.
\end{proof}

We now add a forcing in the remaining cokernel direction and fix translation. On the middle strip, set
\[
\cX^c=W^{1,2}([-S,S];H)\cap L^2([-S,S];D_2),
\qquad \cY^c=L^2([-S,S];H).
\]
The endpoint-evaluation map $\xi\mapsto(\xi(-S),\xi(S))$ from $\cX^c$ to $V\times V$ has a bounded extension operator. We use the induced positive metric to write
\[
z_{d,\xi}(s,t)=\mathcal T_d(t)\exp^{\C}_{v(s,t)}\xi(s,t),
\]
and the pullback of $\mathcal{F}_d(z_{x,\xi})$ 
\[\widetilde{\mathcal{F}}_d(\xi):=\left(\mathcal{T}_d\circ d_{\xi}\exp_v^{\C}\right)^{-1}\left((z_{d,\xi})_s-P_{z_{d,\xi}}(e^{i\theta(d)}\eta\overline{(z_{d,\xi})_{tt}})\right).\]
Express the tail maps in these coordinates as $\mathfrak T^-_d(\beta,h^-)$ and
$\mathfrak T^+_d(h^+)$, with base parameters translated to
zero. Choose $\chi_0\ge0$ supported in $(-S,S)$ with
$\int\chi_0\|T\|_{L_t^2}^2\,\dd s>0$, and define
\[
\psi_{d,z}=\chi_0P_zJT,\qquad\Lambda(\xi)=
\frac{\int_{-S}^S\langle\xi,T\rangle_t\,\dd s}{\int_{-S}^S\|T\|_{L_t^2}^2\,\dd s}.
\]
Writing $\widetilde\psi_{d,\xi}$ for the pullback of the forcing, we obtain the $C^1$ matching map
\begin{equation}\label{eq:complete-match}
\mathfrak N(d,\beta,h^-,h^+,\xi,\lambda)=
\begin{pmatrix}
  \widetilde{\mathcal F}_d(\xi)+\lambda\widetilde\psi_{d,\xi}\\
  \xi(-S)-\mathfrak T^-_d(\beta,h^-)\\
  \xi(S)-\mathfrak T^+_d(h^+)\\
  \Lambda(\xi)
\end{pmatrix}
\in\cY^c\times V\times V\times\R,
\end{equation}
where $(\beta,h^-,h^+,\xi,\lambda)\in\R\times\Pi^uV\times\Pi_+^sV\times\cX^c\times\R$.

\begin{lem}
\label{lem:matching-equivalence}
Let $M$ be the derivative of the matching map
\eqref{eq:complete-match} at the base connector.
Choose a smooth field $\zeta_\beta$ such that
\[
\zeta_\beta(-\infty)=\sigma_*^u,
\qquad
\zeta_\beta(+\infty)=0,
\qquad D\zeta_\beta=0
\quad\text{on the two tails},
\]
and $\Lambda(\zeta_\beta)=0$.
Then $M$ is stably equivalent to
\[
\mathcal D_{\mathrm b}(\zeta,\beta,\lambda)=
\left(D\zeta+\beta D\zeta_\beta+\lambda\chi_0JT,\,\Lambda(\zeta)\right).
\]
The equivalence induces a reflection-equivariant isomorphism $\det M\simeq\det\mathcal D_{\mathrm b}$.
\end{lem}

\begin{proof}
This is the standard cutting-and-gluing reduction using the half-line Green operators and an endpoint extension operator. At the same time, the reduction is triangular with identity diagonal and is reflection-equivariant. Thus,
\[
\det M\simeq\det\mathcal D_{\mathrm b}.\qedhere
\]
\end{proof}

\begin{lem}\label{lem:matching-inverse}
At the $d=0$ base point, the derivative of
\eqref{eq:complete-match} in $(\beta,h^-,h^+,\xi,\lambda)$
is an isomorphism.
\end{lem}

\begin{proof}
By Lemma \ref{lem:matching-equivalence}, the matching derivative is stably equivalent to $\mathcal D_{\mathrm b}$, whose index is
$-1+1+1-1=0$. Equation~\eqref{eq:TB-obstruction} and
$\langle T(s),Z(s)\rangle_t=0$ give
\[
 \begin{pmatrix}
 \langle D\zeta_\beta,JZ\rangle&\langle\chi_0JT,JZ\rangle\\
 \langle D\zeta_\beta,JT\rangle&\langle\chi_0JT,JT\rangle
 \end{pmatrix}
 =\begin{pmatrix}
  2^{-1/2}\|\phi_B\|_2^2&0\\
  0&\displaystyle\int\chi_0\|T\|_2^2\,\dd s
 \end{pmatrix}.
\]
Thus a kernel element has $\beta=\lambda=0$ and $\zeta\in\R T$.
Since $\Lambda(T)=1$, then $\Lambda(\zeta)=0$ implies $\zeta=0$. It follows that $\mathcal{D}_{\mathrm{b}}$ is an index-zero injective operator,
thus invertible.
\end{proof}

\begin{proof}[Proof of Theorem \ref{thm:at-least-two}]
The proof has three steps.
\vskip 0.1cm

First, Lemma \ref{lem:matching-inverse} shows that the linearization
of the augmented matching map at $(d,v)=(0,v_{0,+})$ is an
isomorphism. Thus the parameterized implicit-function theorem gives,
for all sufficiently small $d>0$, a unique nearby zero
\[
(\beta_d,h_d^-,h_d^+,\xi_d,\lambda_d)
\]
of the matching map. Gluing the middle piece to the two half-tails
gives a solution of the forced equation
\[
 z_s=V_d(z)-\lambda_d\psi_{d,z}
\]
with the prescribed limits $\gamma_{1,d}$ and $\gamma_{0,d}$.

\vskip 0.1cm
Second, since $\dd G_{\theta(d)}(V_d)=0$ and the two critical
values of $G_{\theta(d)}$ agree,
\[
0=-\lambda_d\int_{-S}^{S}
\dd G_{\theta(d)}|_{z(s)}
\bigl(\psi_{d,z}\bigr)\,\dd s .
\]
At $d=0$ and $z=v_{0,+}$, the integral is
\[
-\int_{-S}^{S}\chi_0\|T\|_{L_t^2}^2\,\dd s\ne0.
\]
It remains nonzero for small $d$, and therefore $\lambda_d=0$. Thus the matched solution gives the connecting trajectory, and we denote its class by $[v_{d,+}]$.

\vskip 0.1cm
Finally, reflection gives $v_{d,-}:=\mathcal Rv_{d,+}$. By Lemma \ref{lem:reflection}, this is a distinct connector class.
Moreover,
\[
v_{d,+}\longrightarrow v_{0,+},\qquad
v_{d,-}\longrightarrow v_{0,-}
\qquad (d\to0^+).\qedhere
\]
\end{proof}

\subsection{Exactly two connectors for \texorpdfstring{$0<d\ll1$}{small positive d}}
\label{sec:exact-two}

We prove the exact count for the first pair of critical paths.

\begin{thm}
\label{thm:exact-two-small-d}
There are $d_0>0$ and $M_*<\infty$ such that
\[
\mathcal M_d=\{[v_{d,+}],[v_{d,-}]\},\qquad \#\mathcal M_d=2
\quad(0<d<d_0),
\]
and every connector in this range satisfies
\begin{equation}
\label{eq:uniform-bound}
\sup_s\norm{z_t(s)}_{L^2_t}\le M_*.   
\end{equation}
$v_{d,\pm}$ depend $C^1$ on $d\ge0$ in the matching coordinates of last subsection.
\end{thm}

The global argument has two main parts. 
\begin{itemize}
    \item First, we prove the uniform bound \eqref{eq:uniform-bound} by contradiction: 
    \begin{itemize}
        \item blow-up selects the ($n=2$) resonant mode (Lemmas \ref{lem:rescaling-identities}, \ref{lem:upper-action-concentration});
        \item the quadric constraint and the rotation current associated to $K_{\mathrm{rot}}$ give an effective amplitude equation (Lemmas \ref{lem:nonresonant-estimates}, \ref{lem:constraint-reduction}, Proposition \ref{prop:amplitude-equation});
        \item geometric contradiction on the amplitude gives the uniform bound (Proposition \ref{prop:uniform-bound}).
    \end{itemize}
    \item Second, we study the upper tail of a convergent sequence. 
    \begin{itemize}
        \item The endpoint condition selects the limiting Bott point (Lemmas \ref{lem:first-order-action}, \ref{lem:approximate-slow-dynamics}, \ref{lem:select-x0});
        \item the slow reduced dynamics excludes excursions along the Bott family (Proposition \ref{prop:no-slow-return}).
    \end{itemize} 
\end{itemize}

\subsubsection{Rescaling}
\label{sec:rescaling}

For $0<\lambda_{\rm sc}\le1$, set
\[
\eps=\lambda_{\rm sc}^2,\qquad Y(r,t)=\lambda_{\rm sc}\,z(s_0+\lambda_{\rm sc}^{-2}r,t),
\qquad q=\lambda_{\rm sc}^{-1}\sinh d,
\qquad a_Y(r)=\|Y_r(r)\|_{L_t^2}.
\]
We shall take $\lambda_{\rm sc}=H^{-1}$ near a derivative maximum of size $H$, and $\lambda_{\rm sc}=1$ on bounded tails. 

\begin{lem}
\label{lem:rescaling-identities}
The rescaled connector $Y$ satisfies
\begin{equation}
\begin{aligned}
 Y^T\eta Y&=-\eps,\qquad Y(r,0)=\lambda_{\rm sc} a,\qquad Y(r,\tau)=\lambda_{\rm sc} b_d,\\
 Y_{tt}&=(\sigma+e^{i\theta}Y^T\overline{Y_r})Y
              +\eps e^{i\theta}\eta\overline{Y_r},
\end{aligned}
\label{eq:rescaled-flow}
\end{equation}
where
\[
\sigma=\eps^{-1}Y_t^T\eta Y_t,\qquad
\sigma_t=2e^{i\theta}Y_t^T\overline{Y_r},\qquad
\frac12\int_0^\tau\sigma\,\dd t=S(z).
\]
The flow energy is invariant under this rescaling. With $\Krot$ as in \eqref{eq:rotation-current}, the rotation current is
\begin{equation}
\label{eq:rescaled-current}
\int_0^\tau Y_r^T\overline{\Krot Y}\,\dd t
=e^{i\theta}q\,\overline{Y_{2,t}(r,\tau)}.
\end{equation}
If $\|Y_t(r)\|_{L_t^2}\le M$, then there exists $C_M>0$, uniform for $0<\lambda_{\rm sc}\le1$ and bounded $d$, such that
\begin{equation}
\label{eq:rescaled-slice-bounds}
\|Y(r)\|_\infty\le C_M,\qquad
\|\sigma(r)\|_\infty+\|Y_{tt}(r)\|_{L_t^2}
\le C_M(1+a_Y(r)),
\end{equation}
and
\begin{equation}
\label{eq:rescaled-derivative-bound}
\left|\partial_r\|Y_t(r)\|_{L_t^2}^2\right|
       \le C_Ma_Y(r)(1+a_Y(r)).
\end{equation}
\end{lem}

\begin{proof}
Restoring the normal component and using $z^T\eta z=-1$ gives
\[
z_{tt}=(z_t^T\eta z_t+e^{i\theta}z^T\overline{z_s})z+e^{i\theta}\eta\overline{z_s}.
\]
Rescaling proves \eqref{eq:rescaled-flow}, the formula for
$\sigma_t$, and the current identity. Since
$Y_r=\lambda_{\rm sc}^{-1}z_s$ and
$\dd r=\lambda_{\rm sc}^2\dd s$, the flow energy is unchanged.

The fixed endpoint and $\|Y_t\|_2\le M$ give
$\|Y\|_\infty\le|a|+\sqrt\tau M$. By \eqref{eq:action},
the action $S(z)$ is uniformly bounded. Thus
\[
\|\sigma\|_\infty\le C\bigl(|\textstyle\int\sigma|+\|\sigma_t\|_1\bigr)\le C_M(1+a_Y).
\]
The equation gives $\|Y_{tt}\|_2\le C_M(1+a_Y)$.
Finally, $Y_r=0$ at $t=0,\tau$, so
\[
\left|\partial_r\|Y_t\|_2^2\right|
=2\left|\re\int_0^\tau Y_r^T\overline{Y_{tt}}\,\dd t\right|
\le C_Ma_Y(1+a_Y).\qedhere
\]
\end{proof}

The next lemma describes what a derivative blow-up would look like: after rescaling, the solution stays for a long time near the once-round action level $-2\pi^2/\tau$ with very small flow energy.

\begin{lem}
\label{lem:upper-action-concentration}
Suppose $d_j\to0$ and
\[
H_j:=\sup_s\|(z_j)_t(s)\|_2\longrightarrow\infty.
\]
After translating a maximizing slice to $s=0$, set
\[
Y_j(r,t):=H_j^{-1}z_j(H_j^2r,t).
\]
Then there are $a_j<0<b_j$ such that
\begin{equation}
\begin{aligned}
 E_j:=\int_{a_j}^{b_j}\|(Y_j)_r\|_{L_t^2}^2\,\dd r&\longrightarrow0,\\
 \|(Y_j)_t(a_j)\|_2+\|(Y_j)_t(b_j)\|_2&\longrightarrow0,\\
 \sup_{[a_j,b_j]}
 \left|\frac1\tau\int_0^\tau\sigma_j\,\dd t+\frac{4\pi^2}{\tau^2}\right|
 &\longrightarrow0.
\end{aligned}
\label{eq:pinning}
\end{equation}
For each fixed $c\in(0,1)$, there is an interval
$I_j=[\ell_j,u_j]\subset(a_j,b_j)$, with $\ell_j<0<u_j$, such that
\[
\|(Y_j)_t(\ell_j)\|_2=\|(Y_j)_t(u_j)\|_2=c,
\qquad c\le\|(Y_j)_t\|_{L_t^2}\le1\quad\text{on }I_j,
\]
and 
\[\min\{\ell_j-a_j,b_j-u_j\}\to\infty.\]
\end{lem}

\begin{proof} We split the proof into three steps.
\vskip 0.1cm

\noindent\textbf{Step 1. First we prove the following claim:} 

Let $h_j\to\infty$, if an annular crossing from $h_j$ to $2h_j$ has flow energy $\delta_j\to0$, then it contains a slice with $S(z_j)\to-2\pi^2/\tau$.

\vskip 0.1cm
To do it, rescale this crossing at the scale $h_j$:
\[
\widehat Y_j(r,t):=
h_j^{-1}z_j(s_j+h_j^2r,t),
\]
where $s_j$ is any reference point in the crossing. Then 
$1\le \|(\widehat Y_j)_t\|_2\le2$ throughout the rescaled crossing, and the flow energy remains
$\delta_j$. If $L_j$ denotes its rescaled length, integrating
\eqref{eq:rescaled-derivative-bound} gives
\[
3=\bigl|2^2-1^2\bigr|\le
C\bigl(\sqrt{L_j\delta_j}+\delta_j\bigr).
\]
Thus $L_j\to\infty$. Since the middle third has length $L_j/3$ and energy at most $\delta_j$, there exists a point $r_j^*$ in the middle third such that
\[
\|(\widehat Y_j)_r(r_j^*)\|_2^2
\le \frac{3\delta_j}{L_j}
\longrightarrow0.
\]

At this slice the $H_t^2$ bound and
$\|\partial_t\sigma_j\|_1\le C\| (\widehat Y_j)_r\|_2\to0$
give a subsequence with $\sigma_j\to\sigma_*$ uniformly and $\widehat Y_j\to U$ in $H_t^1$. The rescaled equation gives $(\widehat Y_j)_{tt}-\sigma_j\widehat Y_j\to0$ in $L_t^2$; thus the convergence is strong in $H_t^2$ and
\[
U(t)=v\sin(n\pi t/\tau),\qquad
\sigma_*=-n^2\pi^2/\tau^2,\qquad v\ne0.
\]
The action identity gives a limiting action
$-n^2\pi^2/(2\tau)$ in $[-2\pi^2/\tau,0]$, so $1\le n\le2$. Moreover, the rescaled constraint implies $v^T\eta v=0$, and
\[
\partial_t\bigl(h_j\widehat Y_j\wedge(\widehat Y_j)_t\bigr)
=h_j^{-1}e^{i\theta_j}\widehat Y_j\wedge\eta\overline{(\widehat Y_j)_r}.
\]
Its $L_t^1$ norm tends to zero. Strong $H_t^2$ convergence supplies
the endpoint derivatives, so comparison at $0,\tau$ gives
\[
a\wedge v=(-1)^na\wedge v.
\]
Since $a^T\eta a=-1$ and $v$ is nonzero, $a\wedge v\ne0$. Thus $n$ is even, and $n=2$. Therefore every such small-energy crossing contains a slice on which
\[
S(z_j)\longrightarrow-\frac{2\pi^2}{\tau}.
\]

\noindent\textbf{Step 2. We now return to the maximal scale}
\[
H_j:=\sup_s\|(z_j)_t(s)\|_2.
\]
On each side of the maximal slice choose, among the annular crossings
with
\[
H_j^{1/3}\le h_j^\pm\le H_j^{2/3},
\]
one whose energy tends to zero. This is possible because the number of disjoint dyadic crossings in this range tends to infinity, whereas
the total flow energy is uniformly bounded.
Let $s_j^-<0<s_j^+$ be the slices
provided by the preceding argument, so that
\[
S(z_j(s_j^\pm))\longrightarrow-\frac{2\pi^2}{\tau}.
\]

Now set
\[
a_j:=H_j^{-2}s_j^-, \qquad
b_j:=H_j^{-2}s_j^+.
\]
Note that $h_j^\pm/H_j\rightarrow0$, the $H_j$-rescaled connector $Y_j$ satisfies
\[
\|(Y_j)_t(a_j)\|_2+
\|(Y_j)_t(b_j)\|_2
\longrightarrow0.
\]
The flow identity between these two slices then gives $E_j\to0$, and keeps the action uniformly within $o(1)$ of $-\frac{2\pi^2}{\tau}$ on $[a_j,b_j]$. This is exactly \eqref{eq:pinning}.

\vskip 0.1cm
\noindent\textbf{Step 3. Extract the core interval.} Finally fix $c\in(0,1)$. Since
\[
 \|(Y_j)_t(0)\|_2=1,
 \qquad
 \|(Y_j)_t(a_j)\|_2+\|(Y_j)_t(b_j)\|_2\to0,
\]
let $\ell_j<0<u_j$ be the first points, moving away from $0$, at
which
\[
\|(Y_j)_t(\ell_j)\|_2=
\|(Y_j)_t(u_j)\|_2=c.
\]
Then
\[
c\le\|(Y_j)_t(r)\|_2\le1,
\qquad r\in[\ell_j,u_j].
\]

If an interval of length $L$ joins one of these radius-$c$ slices to
$a_j$ or $b_j$, then \eqref{eq:rescaled-derivative-bound} gives
\[
c^2-o(1)\le
C\bigl(\sqrt{LE_j}+E_j\bigr).
\]
Since $E_j\to0$, necessarily $L\to\infty$. Therefore
\[
\ell_j-a_j\to\infty,\qquad
b_j-u_j\to\infty.\qedhere
\]
\end{proof}

\subsubsection{The resonant mode and the rotation current}
\label{sec:resonant-reduction}

By Lemma \ref{lem:upper-action-concentration}, the only mode that can survive in a blow-up limit is $\sin(2\pi t/\tau)$. Since $Y$ has nonzero endpoint values, we first subtract a boundary lift chosen to be orthogonal to the resonant mode.

Let $f(t)=t/\tau+\pi^{-1}\sin(2\pi t/\tau)$, it satisfies
\[
f(0)=0,\qquad f(\tau)=1,\qquad
\int_0^\tau f(t)\sin(2\pi t/\tau)\,\dd t=0.
\]
Moreover, define
\begin{equation}
\label{eq:boundary-lift}
\Xi(t)=\big(\Xi_0(t),\Xi_1(t),\Xi_2(t)\big)=\lambda_{\rm sc}\bigl(a\cos(2\pi t/\tau)+(b_d-a)f(t)\bigr).
\end{equation}
For $L=\partial_t^2+\frac{4\pi^2}{\tau^2}$, direct calculation gives
\begin{equation}
 \|L\Xi\|_{L_t^2}\le C\eps q,\qquad
 |\Xi_{0,t}(0)|\le C\eps qd.
 \label{eq:boundary-lift-bounds}
\end{equation}
Write
\begin{equation}
 Y=\Xi+P(r)\sin(2\pi t/\tau)+R(r,t),\qquad
 P(r)=\frac{2}{\tau}\int_0^\tau Y(r,t)\sin(2\pi t/\tau)\,\dd t.
 \label{eq:sine-decomposition}
\end{equation}
Then $R$ has zero endpoints and is orthogonal to $\sin(2\pi t/\tau)$, so $L$ is invertible on its Dirichlet subspace. Since $\Xi$ is independent of $r$ and the decomposition is $L^2$-orthogonal, we have
\[
 Y_r=P_r\sin(2\pi t/\tau)+R_r,
 \qquad
 \frac{\tau}{2}|P_r|^2+\|R_r\|_{L_t^2}^2=\|Y_r\|_{L_t^2}^2=a_Y^2.
\]
Thus
\begin{equation}
 |P_r|\le Ca_Y,\qquad \|R_r\|_{L_t^2}\le a_Y.
 \label{eq:projection-speed}
\end{equation}
We write $P_0$ for the first component of $P$ and set
$\widetilde{\sigma}=\tau^{-1}\int_0^\tau\sigma\,\dd t$.

\vskip 0.1cm

The next lemma says that, once the action is closed to the once--round level, all nonresonant modes are controlled by the flow speed and boundary correction.
\begin{lem}
\label{lem:nonresonant-estimates}
Suppose on an interval $J^*$ that
\[
\|Y_t\|_{L_t^2}\le M,\qquad
E=\int_{J^*}a_Y^2\,\dd r,\qquad
\sup_{J^*}|\widetilde{\sigma}+\frac{4\pi^2}{\tau^2}|\le\delta,
\]
where $\delta$ is sufficiently small. At every $r$ admitting a
one-sided unit interval in $J^*$,
\begin{equation}
 \|R(r)\|_{H_t^1}+|P_0(r)|
       \le C_M(\delta+\sqrt E+\eps q).
 \label{eq:rough-concentration}
\end{equation}
If $|P(r)|\ge p_*>0$, the sharper pointwise estimate is
\begin{equation}
 |\widetilde{\sigma}(r)+\frac{4\pi^2}{\tau^2}|+\|R(r)\|_{H_t^2}+|P_0(r)|\le C_{M,p_*}(a_Y(r)+\eps q).
 \label{eq:sharp-slaving}
\end{equation}
Here and below, $C_{M,p_*}$ denotes a constant depending only on $M$, $p_*$, and the fixed background parameters.
\end{lem}

\begin{proof} 
\textbf{(1). We first control the nonresonant part.} 

\vskip 0.1cm
Set $m=\widetilde{\sigma}+\frac{4\pi^2}{\tau^2}$. Since
$\|\sigma-\widetilde{\sigma}\|_\infty\le C_Ma_Y$, the rescaled equation becomes
\begin{equation}
 LY=mY+F,\qquad F:=(\sigma-\widetilde\sigma)Y+e^{i\theta}(Y^T\overline{Y_r})Y+\eps e^{i\theta}\eta\overline{Y_r}.
 \label{eq:shifted-equation}
\end{equation}
The bounds for $\sigma-\widetilde\sigma$ and $Y$ give
\[
 \|F\|_{L_t^2}\le C_Ma_Y.
\]
Let $Q$ denote the $L_t^2$-orthogonal projection onto $\sin(2\pi t/\tau)^\perp$, by decomposition \eqref{eq:sine-decomposition}, we have
\[
LR=mR+m\Xi+QF-QL\Xi.
\]
Thus, by \eqref{eq:boundary-lift-bounds},
\[
\|LR-mR\|_{L_t^2}\le C_M(\delta+a_Y+\eps q).
\]
Since $L$ is invertible on the nonresonant subspace, for $\delta$
sufficiently small, we can absorb the $mR$ term and use the standard elliptic estimate to obtain
\begin{equation*}
\|R(r)\|_{H_t^2}\le C_M(\delta+a_Y(r)+\eps q).
\end{equation*}
In a one-sided unit interval $J$ at $r$, choose $r_0$ with
$a_Y(r_0)\le\sqrt E$. At this slice the right-hand side is bounded
by $C_M(\delta+\sqrt E+\eps q)$. Also,
\[
\left|\partial_r\|R_t\|_{L_t^2}^2\right|
\le2\|R_r\|_{L_t^2}\|R_{tt}\|_{L_t^2}
\le C_Ma_Y(\delta+a_Y+\eps q).
\]
Integration over $J$ proves the $H^1$ estimate at $r$.
Tangency at $t=0$ gives $\frac{2\pi}{\tau}P_0=-\Xi_{0,t}(0)-R_{0,t}(0)$. The $H^2$ estimate at $r_0$ and $|(P_0)_r|\le Ca_Y$ give the same bound for $P_0$.

\vskip 0.1cm
\noindent\textbf{(2) Now we consider the resonant part.} 

\vskip 0.1cm
Taking the sine component of
\eqref{eq:shifted-equation} gives
\[m\,\frac{\tau}{2}P=\int_0^\tau\sin(2\pi t/\tau)L\Xi\,\dd t
-\int_0^\tau\sin(2\pi t/\tau)F\,\dd t.
\]
Then the condition $|P|\ge p_*>0$ gives $|m|\le C_{M,p_*}(a_Y+\eps q)$.
Substitution in the complementary equation and the endpoint
identity for $P_0$ proves \eqref{eq:sharp-slaving}.
\end{proof}

We now introduce null coordinates for the resonant coefficient $P$. In these coordinantes, the quadratic constraint will determine all but one complex amplitude. Set
\[
\nu_{\varsigma}
=\frac1{\sqrt2}(0,1,\varsigma i),
 \qquad
e_0=(1,0,0),
 \qquad
\nu_{-\varsigma}=\overline{\nu_{\varsigma}},
 \qquad
\varsigma\in\{\pm1\}.
\]
These vectors are Hermitian orthonormal and satisfy
\[
\nu_{\varsigma}^T\eta\nu_{\varsigma}=0,\qquad\nu_{\varsigma}^T\eta\nu_{-\varsigma}=1,\qquad
e_0^T\eta e_0=-1.
\]
Write
\begin{equation}
\label{eq:null-coordinates}
P=(P_0,P_1,P_2)=A\nu_\varsigma+Be_0+D\nu_{-\varsigma}.
\end{equation}
Then
\[
P^T\eta P=2AD-B^2,
\qquad B=P_0.
\]

\begin{lem}
\label{lem:constraint-reduction}
Under the hypotheses of Lemma \ref{lem:nonresonant-estimates}, suppose that $0<a_0\le|A|\le a_1$ and $\|R\|_{H_t^1}+|B|\le\eta_0$, with
$\eta_0$ sufficiently small. Then
\begin{equation*}
 D=-\frac{\eps}{2A}+W,\qquad\text{with }~
 |W|\le C(a_Y+\eps q)~\text{ and }~ |W_r|\le Ca_Y.
\end{equation*}
If $\sup(\|R\|_{H_t^1}+|B|)\to0$ and $d\to0$, then
$\sup|W|\to0$. Moreover, on every subinterval $J$,
\begin{equation*}
\|R\|_{L_r^2(J;H_t^2)}+\|W\|_{L_r^2(J)}\le C\bigl(\sqrt{E_J}+\eps q\sqrt{|J|}\bigr),
 \qquad E_J=\int_Ja_Y^2\,\dd r.
\end{equation*}
\end{lem}

\begin{proof}
Integrating $Y^T\eta Y=-\eps$ and using the orthogonality of decomposition gives
\[
\frac{\tau}{2}(2AD-B^2)+\int_0^\tau(\Xi+R)^T\eta(\Xi+R)\,\dd t
=-\eps\tau.
\]
Thus $D$ could be solved as
\[
D=-\frac{2\varepsilon}{W},\qquad\text{where}\quad W=\frac1{2A}\left(B^2-\frac{2(J_0+F_R)}{\tau}\right),
\]
with
\[
J_0=\int_0^\tau\Xi^T\eta\Xi\,\dd t+\frac{\eps\tau}{2},\qquad
F_R=2\int_0^\tau\Xi^T\eta R\,\dd t+\int_0^\tau R^T\eta R\,\dd t.
\]
By \eqref{eq:boundary-lift}, for sufficiently small $d>0$,
\[|J_0|\le C\eps d^2\le C\eps^2q^2\le C\eps q.\]
Since $|P|\ge|A|\ge a_0$, \eqref{eq:sharp-slaving} applies,
\[|B|+\|R\|_{H_t^2}\le C(a_Y+\eps q).\]
Together with the exact formula for $W$ and
\eqref{eq:projection-speed}, we get 
\[|W|\le C(a_Y+\varepsilon q),\qquad|W_r|\le Ca_Y.\]

If $\sup(\|R\|_{H_t^1}+|B|)\to0$ and $d\to0$, then $J_0\to0$, $F_R\to 0$ and $B\to0$. It follows from the exact formula again that $\sup|W|\to 0$. Finally, taking the $L^2(J)$ norm gives
\[
 \|R\|_{L_r^2(J;H_t^2)}
 +\|W\|_{L_r^2(J)}
 \le
 C\left(
 \sqrt{E_J}+\eps q\sqrt{|J|}
 \right).\qedhere
\]
\end{proof}

Next we use the rotation current to derive an effective equation for the resonant amplitude. A small correction of $A$ removes the total-derivative terms and gives an error estimate on arbitrary subinterval.

\begin{prop}
\label{prop:amplitude-equation}
Under the hypotheses of Lemma \ref{lem:constraint-reduction} on an interval $I$, suppose
$\eps q\le1$ and set
$\eta_1=\sup(\|R\|_{H_t^1}+|B|+|W|)$.
For $\eta_1$ sufficiently small, there is a corrected amplitude
$\widetilde A$ such that
\[
\|\widetilde A-A\|_{L^\infty(I)}\le C\eta_1\qquad\text{and}\qquad\widetilde A_r=q\,b_{\eps,d}(\widetilde A)+g,
\]
where
\[
b_{\eps,d}(A)=-e^{i\theta(d)}\frac{4\pi}{\sqrt2\,\tau^2}\frac{1+\eps/(2\bar A^2)}{1+\eps^2/(4|A|^4)}.
\]
Moreover, for every subinterval $J\subset I$,
\begin{equation}
 \int_J|g|\,\dd r
 \le C\left(E_J+q\sqrt{|J|E_J}+(\eps q+\eta_1)q|J|\right).
 \label{eq:current-error}
\end{equation}
\end{prop}

\begin{proof}
Set
\[
D_0=-\frac{\varepsilon}{2A},\qquad
D=D_0+W.
\]
Substituting \eqref{eq:sine-decomposition} and \eqref{eq:null-coordinates} 
into the rotation current, and using
$K_{\rm rot}\nu_\sigma=-\sigma i\nu_\sigma$, gives
\[
\frac{\sigma i\tau}{2}
\bigl(A_r\bar A-D_r\bar D\bigr)+(D_\Xi)_r+J_R
=-\frac{\sigma i e^{i\theta}q\,2\pi}{\sqrt2\,\tau}
(\bar A-\bar D)
+e^{i\theta}q\,\overline{R_{2,t}(\tau)},
\]
where
\[
D_\Xi=
\int_0^\tau R^T\overline{K_{\rm rot}\Xi}\,\dd t,
\qquad
J_R=
\int_0^\tau R_r^T\overline{K_{\rm rot}R}\,\dd t.
\]

We first ignore the remainder terms $R$ and $W$. Then $D=D_0$, and
\[
A_r\bar A-(D_0)_r\bar D_0
=\bar A F_\varepsilon(A)A_r,
\qquad
F_\varepsilon(A):=
1+\frac{\varepsilon^2}{4|A|^4}.
\]
Thus the current identity reduces to
\[
A_r=q\,b_{\varepsilon,d}(A),
\quad
\text{with }~ b_{\varepsilon,d}~ \text{ as above}.\]

For the full equation, write
\[
D_r\bar D=
(D_0)_r\bar D_0
+(W\bar D_0)_r+T,
\qquad
T:=(D_0)_r\bar W
-W(\bar D_0)_r+W_r\bar W,
\]
and set
\[
\mathcal D:=
\frac{\sigma i\tau}{2}W\bar D_0-D_\Xi,
\qquad
\Gamma_\varepsilon(A)
:=\frac{2}{\sigma i\tau\,\bar A F_\varepsilon(A)}.
\]
Moreover, we define the corrected amplitude
\[
\widetilde A:=A-\Gamma_\varepsilon(A)\mathcal D.
\]
On the fixed annulus,
\[
|\mathcal D|
\le C\bigl(\varepsilon|W|
+\varepsilon q\|R\|_{L_t^2}\bigr),
\]
so, for $\eta_1$ sufficiently small,
\[
\|\widetilde A-A\|_{L^\infty(I)}
\le C\eta_1.
\]

After this change of variable, the current identity takes the form
\[
\widetilde A_r=
q\,b_{\varepsilon,d}(A)+g_0,
\]
where $g_0$ consists of the terms involving
\[
J_R,\qquad T,\qquad
(\Gamma_\varepsilon)_r\mathcal D,\qquad
qR_{2,t}(\tau),\qquad qW.
\]
By Lemmas~\ref{lem:nonresonant-estimates} and
\ref{lem:constraint-reduction}, together with Cauchy--Schwarz and
the $H_t^2$ boundary estimate, for every subinterval $J\subset I$,
\[
\int_J\bigl(|J_R|+|T|+
|(\Gamma_\varepsilon)_r\mathcal D|\bigr)\,\dd r\le C\sqrt{E_J}
\bigl(\sqrt{E_J}+\varepsilon q\sqrt{|J|}\bigr),
\]
and
\[
q\int_J
\bigl(
|R_{2,t}(\tau)|+|W|
\bigr)\,\dd r\le
Cq\sqrt{|J|E_J}+C\varepsilon q^2|J|.
\]
Hence
\[
\int_J|g_0|\,\dd r\le
C\left(E_J
+q\sqrt{|J|E_J}
+\varepsilon q^2|J|
\right).
\]

Finally, $b_{\varepsilon,d}$ is uniformly Lipschitz on a slightly
larger fixed annulus. Since $|\widetilde A-A|\le C\eta_1$, we replace $b_{\varepsilon,d}(A)$ by
$b_{\varepsilon,d}(\widetilde A)$ at the cost $C\eta_1q|J|$. Thus
\[
\widetilde A_r=
q\,b_{\varepsilon,d}(\widetilde A)+g,
\qquad
\int_J|g|\,\dd r\le
C\left(E_J+q\sqrt{|J|E_J}
+(\varepsilon q+\eta_1)q|J|\right).\qedhere
\]
\end{proof}

\subsubsection{Uniform bounds for connectors}
\label{sec:uniform-bound}

\begin{prop}
\label{prop:uniform-bound}
There are $d_*>0$ and $M_*<\infty$ such that every connector with
$0<d<d_*$ satisfies
\[
 \sup_s\norm{z_t(s)}_{L^2_t}\le M_*.
\]
\end{prop}

\begin{proof}
Recall that
\[\varepsilon_j=H_j^{-2},\qquad q_j=H_j\sinh d_j.\]
(1). Suppose $d_j\to0$ and $H_j:=\sup_s\|(z_j)_t(s)\|_2\to\infty$.
For a fixed $c\in(0,1)$, using the rescaling $Y_j$ and intervals $I_j=[\ell_j,u_j]$ of Lemma \ref{lem:upper-action-concentration}, we obtain that $c\le\|(Y_j)_t\|_{L_t^2}\le1$ on $I_j$ and the energy on the larger interval $[a_j,b_j]$ tends to zero. The long collars allow us to apply \eqref{eq:rough-concentration} throughout $I_j$:
\[
 \sup_{I_j}(\|R_j\|_{H_t^1}+|(P_j)_0|)\longrightarrow0.
\]
Since $\lambda_{\mathrm{sc},j}\to0$, the lift $\Xi_j$ tends to zero in $H^1$. Note that
\[\|(Y_j)_t\|_{L_t^2}^2=\frac{2\pi^2}{\tau}|P_j|^2+o(1),\]
then $|P_j|\ge p_*>0$, while the constraint gives
$P_j^T\eta P_j=o(1)$. Then continuity fixes one of the two null directions on $I_j$, with
\[
a_0\le|A_j|\le a_1,\qquad
\sup_{I_j}(|B_j|+|D_j|)\longrightarrow0.
\]
Since $\eps_jq_j=H_j^{-1}\sinh d_j\to0$, Lemma \ref{lem:constraint-reduction} also gives $W_j\to0$ uniformly. Thus the hypotheses of Proposition \ref{prop:amplitude-equation} hold for large $j$. Write $E_j=\int_{I_j}a_{Y_j}^2\,\dd r$, then $E_j\to0$.

\vskip 0.1cm
\noindent(2). We show that $q_j\to0$. Applying
\eqref{eq:rescaled-derivative-bound} on $[\ell_j,0]$ and $[0,u_j]$ gives
\[
1-c^2\le C\bigl(\sqrt{(-\ell_j)E_j}+E_j\bigr),\qquad1-c^2\le C\bigl(\sqrt{u_jE_j}+E_j\bigr).
\]
Thus $-\ell_j,u_j\to\infty$. The same estimate gives
$\|(Y_j)_t\|_{L_t^2}=1+o(1)$ on $[-1,1]$. Choose $r_j$ there
with $a_{Y_j}(r_j)\le C\sqrt{E_j}$. The sharp sine estimate gives
\[
|(Y_j)_{2,t}(r_j,\tau)|
=\left|\frac{\varsigma i\,2\pi}{\sqrt2\,\tau}A_j(r_j)+o(1)\right|
 \ge c_0>0.
\]
Therefore the current identity \eqref{eq:rescaled-current} implies
\[q_j\le C\sqrt{E_j}\to0.\]

\noindent(3). Next we apply the amplitude equation on $I_j$. Proposition \ref{prop:amplitude-equation} gives
\[
 \|\widetilde A_j-A_j\|_\infty\to0.
\]
Moreover, uniformly on the fixed annulus,
\[
 b_{\eps_j,d_j}(A)=c_*+o(1),
 \qquad
 c_*=\frac{4\pi}{\sqrt2\,\tau^2}>0.
\]
For every $J\subset I_j$, $E_J\le E_j$, Proposition \ref{prop:amplitude-equation} gives
\[
 \int_J
 \bigl|(\widetilde A_j)_r-c_*q_j\bigr|\dd r
 \le
 C\left(
 E_j
 +q_j\sqrt{|J|E_j}
 +o(1)\,q_j|J|
 \right).
\]
The mixed term satisfies
\[
 q_j\sqrt{|J|E_j}
 =\sqrt{q_jE_j}\,\sqrt{q_j|J|}
 \le\frac12\sqrt{q_jE_j}\,
 \bigl(1+q_j|J|\bigr).
\]
Since $E_j\to0$ and $q_j\to0$, all the error terms can therefore be
absorbed into
\[
 \omega_j\bigl(1+q_j|J|\bigr),
 \qquad
 \omega_j\to0.
\]
Integrating over $J=[x,y]\subset I_j$ yields
\begin{equation}
 \left|\widetilde A_j(y)-\widetilde A_j(x)-c_*q_j(y-x)\right|
 \le\omega_j\bigl(1+q_j(y-x)\bigr),\qquad\omega_j\to0.
 \label{eq:all-intervals-drift}
\end{equation}
Boundedness of $\widetilde A_j$ and
\eqref{eq:all-intervals-drift} give $q_j|I_j|\le C$.

\vskip 0.1cm
\noindent(4). The sine decomposition also gives, uniformly on $I_j$,
\[
 \|(Y_j)_t(r)\|_2^2=\frac{2\pi^2}{\tau}|A_j(r)|^2+o(1).
\]
Its value is $1$ at $0$ and $c^2$ at both endpoints. Thus neither
$q_j(-\ell_j)$ nor $q_j(u_j)$ can tend to zero, by
\eqref{eq:all-intervals-drift}. After taking a subsequence,
\[
 q_j(-\ell_j)\to L_->0,\qquad q_ju_j\to L_+>0,
 \qquad A_j(0)\to A_*,\qquad |A_*|^2=\frac{\tau}{2\pi^2}.
\]
The endpoint values now give
\[
 |A_*-c_*L_-|^2=|A_*+c_*L_+|^2=c^2\frac{\tau}{2\pi^2}<|A_*|^2.
\]
The first equality forces $\re A_*>c_*L_-/2>0$, whereas the
second forces $\re A_*<-c_*L_+/2<0$. This contradiction proves
the uniform bound.
\end{proof}

\subsubsection{Compactness and limiting connectors}
\label{sec:compactness}

\begin{lem}
\label{lem:bounded-regularity}
Let $d_j\to0$ and let $z_j$ be complete connectors representing classes in $\mathcal M_{d_j}$. Every sequence of $s$-translates has a subsequence converging in $C^\infty_{\mathrm{loc}}(\mathbb R\times[0,\tau])$ to a solution of the $d=0$ equation.
\end{lem}

\begin{proof}
Let
\[
 \widetilde z_j(s,t):=z_j(s+s_j^0,t)
\]
be any sequence of $s$-translates. Proposition~\ref{prop:uniform-bound}
and the energy identity give
\[
 \sup_s\|(\widetilde z_j)_t(s)\|_2\le M_*,
 \qquad
 \int_{\R}\|(\widetilde z_j)_s\|_2^2\,\dd s\le C.
\]
At $\lambda_{\rm sc}=1$, \eqref{eq:rescaled-slice-bounds} gives
\[
 \|(\widetilde z_j)_{tt}(s)\|_2
 \le C\bigl(1+\|(\widetilde z_j)_s(s)\|_2\bigr).
\]
Thus every fixed strip has a uniform $W_s^{1,2}L_t^2\cap L_s^2H_t^2$ bound.

Since $\widetilde z_j(s,0)=a$,
\[
\|\widetilde z_j(s,\cdot)\|_{L_t^\infty}\le |a|+\sqrt{\tau}\,\|(\widetilde z_j)_t(s)\|_2
 \le |a|+\sqrt{\tau}\,M_*,
\]
then all images lie in a fixed compact subset of the quadric. Then Lemma \ref{lem:mixed-regularity} and a diagonal argument give, after passing to a subsequence,
\[
\widetilde z_j\longrightarrow z
\qquad\text{in }C^\infty_{\loc}(\R\times[0,\tau]).
\]
Since $d_j\to0$, the equation and endpoint conditions pass to the
$d=0$ limit. Lower semicontinuity of the flow energy gives
\[
\int_{\R}\|z_s\|_2^2\,\dd s<\infty.\qedhere
\]
\end{proof}

\begin{cor}
\label{cor:critical-band}
Let $z$ be a complete limit given by
Lemma \ref{lem:bounded-regularity}. If
\[
 0\le F_\pi(z(s))\le \frac{2\pi^2}{\tau},\qquad
 0<F_\pi(z(s_0))<\frac{2\pi^2}{\tau}\quad\text{for some }s_0,
\]
then
\[
z(s)\longrightarrow X_{w_*} \quad (s\to-\infty),
\qquad
z(s)\longrightarrow a \quad (s\to+\infty)
\]
smoothly, for some real $w_*$.
\end{cor}

\begin{proof}
Monotonicity gives the limits
\[
F_\pm=\lim_{s\to\pm\infty}F_\pi(z(s)).
\]
Finite energy and local compactness show that every end sequence has
a stationary subsequential limit. The only critical values in $[0,2\pi^2/\tau]$ are $0$ and $2\pi^2/\tau$, so the strict intermediate value implies
\[
F_-=\frac{2\pi^2}{\tau},
\qquad F_+=0.
\]
The standard Morse--Bott gradient inequality near the two critical sets gives unique end limits and smooth convergence. Hence the
negative limit is a once-round Bott loop $X_{w_*}$, while the positive limit is $a$. By Theorem~\ref{thm:exact-two-d0}, the limiting connector lies in $\Sigma$, so $w_*\in\R$.
\end{proof}

Corollary \ref{cor:critical-band} gives us an upper end $X_{w_*}$. It remains to prove that
\[w_*=0\]
and that the whole upper tail stays near $X_0$. The first point comes from the first-order endpoint condition, and the second from the slow centre dynamics.

\subsubsection{Exclusion of slow excursions}
\label{sec:no-slow-excursions}

After passing to a diagonal subsequence, choosing $T_j\to\infty$ such that
\begin{equation}\label{eq:expanding-strip-convergence}
 \|z_j-z_\infty\|_{C^2([-T_j,T_j]\times[0,\tau])}\longrightarrow0.
\end{equation}
Set $S_j=-T_j$. Since the ends of $z_{\infty}$ are $X_{w_*}$ and $a$, the action identity gives
\begin{equation}
 \begin{gathered}
 z_j(S_j)\longrightarrow X_{w_*},\qquad z_j(T_j)\longrightarrow a,\\
 E_j^-:=\int_{-\infty}^{S_j}\|(z_j)_s\|_2^2\,\dd s\longrightarrow0,
 \qquad
 E_j^+:=\int_{T_j}^\infty\|(z_j)_s\|_2^2\,\dd s\longrightarrow0.
 \end{gathered}
 \label{eq:upper-tail-energy}
\end{equation}
On $s\le S_j$, we use the sine decomposition with
\[
\lambda_{\rm sc}=\eps=1,\qquad
 r=s,\qquad q_j=\sinh d_j,\qquad \varsigma=1.
\]
On $s\le S_j$, the action stays within $o(1)$ of $-2\pi^2/\tau$ and $E_j^-\to0$. Then the nonresonant and constraint estimates (Lemma \ref{lem:nonresonant-estimates}) give
\begin{equation}
\label{eq:upper-tail-reduction}
\sup_{s\le S_j}
\bigl(\|R_j\|_{H_t^1}+|B_j|+|W_j|\bigr)\longrightarrow0,\qquad 0<a_0\le |A_j|\le a_1.
\end{equation}
Thus the amplitude equation is valid on the whole upper half-line. We also write
\[
A_j=\frac{i}{\sqrt2}e^{-iw_j},\qquad w_j\in\C/(2\pi\Z).
\]
Set
\[
 \mathscr H(w):=\re\cos w, \qquad \mathscr{K}_B(w):=-\im\cos w.
\]
The next two steps use $\mathscr H$ to identify the limit Bott point and $\mathscr{K}_B$ to rule out motion along the Bott family.

\begin{lem}\label{lem:first-order-action}
On the upper half-line,
\begin{equation}\label{eq:first-order-pin}
 |\mathscr H(w_j(s))-1|\le C\left(d_j+\frac{\|(z_j)_s(s,\cdot)\|_{L_t^2}^2}{d_j}\right).
\end{equation}
\end{lem}

\begin{proof}
Suppress the index $j$ and take $\sigma=1$. Write $z=X_w+\zeta$. The sharp nonresonant estimate and the
constraint reduction give
\[
 \|\zeta(s)\|_{H_t^1}\le C\bigl(d+\|z_s(s,\cdot)\|_{L_t^2}\bigr),\qquad \zeta(s,0)=0,\quad \zeta(s,\tau)=b_d-a.
\]
Since $z$ and $X_w$ lie on the same quadric, $X_w^T\eta\zeta=-\frac12\zeta^T\eta\zeta$.Using $(X_w)_{tt}=-(4\pi^2/\tau^2)X_w$ and integrating by parts gives
\begin{align*}
S(z)=&-\frac{2\pi^2}{\tau}+i\frac{2\pi}{\tau}\sinh d\cos w +\frac12\int_0^\tau\left(\zeta_t^T\eta\zeta_t-\frac{4\pi^2}{\tau^2}\zeta^T\eta\zeta\right)\,\dd t\\
=&-\frac{2\pi^2}{\tau}+i\frac{2\pi}{\tau}\sinh d\cos w +O(d^2+\|z_s(s,\cdot)\|_{L_t^2}^2).
\end{align*}
On the other hand,
\[
S(\gamma_{1,d})=-\frac{2\pi^2}{\tau}+i\frac{2\pi d}{\tau}+\frac{d^2}{2\tau}.
\]
Conservation of $G_{\theta(d)}$ says that the imaginary part of
$e^{-i\theta(d)}(S(z)-S(\gamma_{1,d}))$ vanishes. Since
$e^{-i\theta(d)}=-1+O(d)$ and $\cos w$ is uniformly bounded on the
fixed amplitude annulus,
\[
 |\sinh d\,\re\cos w-d|\le C\bigl(d^2+\|z_s(s,\cdot)\|_{L_t^2}^2\bigr).
\]
Dividing by $\sinh d$ proves \eqref{eq:first-order-pin}.
\end{proof}

The endpoint estimate forces $\mathscr{H}(w_j)$ close to $1$. To control the remaining motion along the Bott family, we use the limiting centre field $\mathcal G_0$. Set
\[
\Psi(w):=\dd \mathscr{K}_B(w)\bigl[\mathcal G_0(w)\bigr].
\]
Then $\Psi\ge0$, and on the level $\mathscr H=1$,
\[
\Psi(w)=0\quad\Longleftrightarrow\quad w=0.
\]

\begin{lem}
\label{lem:approximate-slow-dynamics}
There is $\eta_j\to0$ such that, for every
$J=[a,b]\subset(-\infty,S_j]$,
\begin{align}
 |A_j(b)-A_j(a)|
&\le Cq_j|J|+\eta_j(1+q_j|J|),
\label{eq:slow-a}\\
\left|\mathscr{K}_B(w_j(b))-\mathscr{K}_B(w_j(a))
-q_j\int_J\Psi(w_j(s))\dd s\right|
&\le\eta_j(1+q_j|J|).\label{eq:slow-monotonicity}
\end{align}
Moreover, let $s_j\le S_j$ and set
\[
w_j^{(s_j)}(\rho):=
w_j\!\left(s_j+\frac{\rho}{q_j}\right).
\]
On every bounded $\rho$-interval on which these functions are defined, every locally uniform limit solves
\[
w_\rho=\mathcal G_0(w).
\]
\end{lem}

\begin{proof}
Proposition \ref{prop:amplitude-equation} gives
$\widetilde A_j=A_j+o(1)$ uniformly and
\[
(\widetilde A_j)_s=q_jb_{1,d_j}(\widetilde A_j)+g_j,
\qquad b_{1,d_j}=b_{1,0}+O(d_j)
\]
on the fixed annulus. Moreover, for every $J\subset(-\infty,S_j]$,
\[
\int_J|g_j|\dd s\le
C\left(E_j^-+q_j\sqrt{|J|E_j^-}
+(q_j+o(1))q_j|J|\right).
\]
Since
\[
q_j\sqrt{|J|E_j^-}\le
\frac12\sqrt{q_jE_j^-}\,(1+q_j|J|),
\]
all errors have the form $\eta_j(1+q_j|J|)$.
This proves \eqref{eq:slow-a}. 

Recall that
\[A=\frac{i}{\sqrt{2}}e^{-iw},\]
then the field $b_{1,0}$ is exactly $A$-coordinate form of $\mathcal{G}_0$. Applying the chain rule to $\mathscr{K}_B$ then gives \eqref{eq:slow-monotonicity}.

On a bounded slow-time interval, introduce the slow-time variable
\[\rho=q_j(s-s_j)<q_j|J|,\qquad s_j\leq S_j.\]
The integrated error tends to $0$, while $b_{1,d_j}\to b_{1,0}$ uniformly on the fixed annulus. Thus every locally uniform limit $A(\rho)$ satisfies the limiting centre equation
\[A(\rho)=b_{1,0}A,\]
which is equaivalent to, in terms of coordinate $w$,
\[w_{\rho}=\mathcal{G}_0(w).\qedhere\]
\end{proof}

We first identify the Bott point of the local limit.

\begin{lem}
\label{lem:select-x0}
The point $X_{w_*}$ in \eqref{eq:upper-tail-energy} is $X_0$.
\end{lem}

\begin{proof}
Set 
\[L_j:=(\sqrt{E_j^-}+d_j)/q_j.\] 
Since the flow energy on $(-\infty,S_j]$ is $E_j^-$, there is $\widehat s_j\in[S_j-L_j,S_j]$ such that
\[
\|(z_j)_s(\widehat s_j,\cdot)\|_{L_t^2}^2\le\frac{E_j^-}{L_j}.
\]
Since $q_j/d_j\to1$,
\[
\frac{\|(z_j)_s(\widehat s_j,\cdot)\|_{L_t^2}^2}{d_j}
\le\frac{q_jE_j^-}{d_j(\sqrt{E_j^-}+d_j)}\longrightarrow0.
\]
On the other hand, \eqref{eq:slow-a} and
$q_jL_j=\sqrt{E_j^-}+d_j\to0$ give
\[
A_j(\widehat s_j)-A_j(S_j)\longrightarrow0.
\]
Thus $w_j(\widehat s_j)\to w_*$. Passing to the limit in
\eqref{eq:first-order-pin} gives
\[
\mathscr H(w_*)=\re\cos w_*=1.
\]
By Corollary~\ref{cor:critical-band}, $w_*$ is real. Therefore $\cos w_*=1$, and then $w_*=0$ in $\R/(2\pi\Z)$. 
\end{proof}

It remains to rule out a slow excursion away from $X_0$ on the upper tail.

\begin{prop}
\label{prop:no-slow-return}
On the entire upper half-line
\[
\sup_{s\le S_j}\dist(w_j(s),0)\longrightarrow0.
\]
\end{prop}

\begin{proof}
We argue by contradiction. Suppose that there are $s_j\le S_j$ and $\varepsilon_0>0$ such that 
\[\dist(w_j(s_j),0)\ge\varepsilon_0.\] As before, set
\[\rho=q_j(s-s_j),\qquad L_j=q_j(S_j-s_j).\] 
Since $w_j(S_j)\to0$, the slow-variation estimate \eqref{eq:slow-a} implies $\liminf L_j>0$.

For every fixed $\delta>0$, \eqref{eq:first-order-pin} gives
\[
 q_j\bigl|\{s\le S_j:\ |\mathscr H(w_j(s))-1|\ge\delta\}\bigr|\le C_\delta E_j^-\longrightarrow0.
\]
Thus every locally uniform slow-time limit lies on the level $\mathscr H=1$.

By Lemma \ref{lem:approximate-slow-dynamics}, after passing to a subsequence the shifted amplitudes converge locally uniformly to a solution of the limiting centre equation $w_{\rho}=\mathcal{G}_0(w)$. 
\begin{itemize}
    \item If $L_j\to L<\infty$, the limiting solution satisfies $w(L)=0$. Then uniqueness for the limiting ODE gives $w\equiv0$, contracting with the distance assumption.
    \item If $L_j\to\infty$, the limit is a bounded complete solution on $\R$ with $\mathscr{H}=1$. Along this solution, 
    \[(\mathscr{K}_B)_{\rho}=\Psi\ge0.\]
    Thus $\mathscr{K}_B$ is bounded and nondecreasing. Its two end limits are both $\mathscr{K}_B(0)=0$, since on $\mathscr H=1$ the condition $\Psi=0$ forces $w=0$. Hence $\mathscr{K}_B$ is constant, so $\Psi\equiv0$ and  $w\equiv0$, again a contradiction.
\end{itemize}
Thus, we obtain
\[
\sup_{s\le S_j}\dist(w_j(s),0)
\longrightarrow0.\qedhere
\]
\end{proof}

Combining \eqref{eq:upper-tail-reduction} with
Proposition \ref{prop:no-slow-return}, the whole upper tail
converges to the $X_0$ tail in the resonant coordinates.

\subsubsection{The exact count}

We now pass from the preceding convergence results to the matching coordinates of Subsection \ref{sec:at-least-two}.

\begin{lem}
\label{lem:uniform-tail-entry}
After the preceding translations and subsequence selection,
\[
 \lim_{S\to\infty}\limsup_{j\to\infty}
 \sup_{s\le-S}\|z_j(s)-\gamma_{1,d_j}\|_{H_t^1}=0,
 \qquad
 \lim_{S\to\infty}\limsup_{j\to\infty}
 \sup_{s\ge S}\|z_j(s)-\gamma_{0,d_j}\|_{H_t^1}=0.
\]
Thus, for a fixed sufficiently large $S$, both complete half-solutions
belong to the fixed small neighborhoods of
Proposition \ref{prop:complete-tails}, for all sufficiently large $j$.
\end{lem}

\begin{proof} For convenience, below we use the notation $\pm T_j$ (recall that $S_j=-T_j$) for the expanding strip. We use the expanding-strip convergence
\[
z_j\longrightarrow z_\infty
\qquad\text{on }[-T_j,T_j],
\]
where
\[
z_\infty(s)\longrightarrow X_0
\quad(s\to-\infty),
\qquad
z_\infty(s)\longrightarrow a
\quad(s\to+\infty).
\]

For the upper tail, Proposition~\ref{prop:no-slow-return},
together with the resonant decomposition, gives
\[
\sup_{s\le -T_j}
\|z_j(s)-X_0\|_{H_t^1}\longrightarrow0.
\]
On $[-T_j,-S]$, the expanding-strip convergence gives
\[ 
\sup_{-T_j\le s\le-S}
\|z_j(s)-X_0\|_{H_t^1}
\longrightarrow0
\]
as first $j\to\infty$ and then $S\to\infty$.
Since $\gamma_{1,d_j}\to X_0$, this proves
\[
\lim_{S\to\infty}\limsup_{j\to\infty}
\sup_{s\le-S}
\|z_j(s)-\gamma_{1,d_j}\|_{H_t^1}=0.
\]

For the lower tail, the same expanding-strip convergence gives
\[
\lim_{S\to\infty}\limsup_{j\to\infty}
\sup_{S\le s\le T_j}\|z_j(s)-a\|_{H_t^1}=0.
\]
Moreover, we claim that
\[\sup_{s\ge T_j}\|z_j(s)-a\|_{H_t^1}\longrightarrow0.\]
Otherwise, we can choose $s_j\ge T_j$ with $\|z_j(s_j)-a\|_{H_t^1}\ge\varepsilon_0$. For every fixed $R>0$, 
\[\int_{s_j-R}^{s_j+R}\|(z_j)_s\|_2^2\,\dd s\longrightarrow0.\]
Then Lemma \ref{lem:bounded-regularity} gives a stationary $d=0$ limit of $z_j(s_j+\cdot)$, which has action $0$, so it is the constant path $a$, contradicting with the choice of $s_j$. Thus the claim holds. And since $\gamma_{0,d_j}\to a$, we have
\[
\lim_{S\to\infty}\limsup_{j\to\infty}
\sup_{s\ge S}
\|z_j(s)-\gamma_{0,d_j}\|_{H_t^1}=0.
\]

For $S$ sufficiently large, both half-tails therefore lie in the
neighborhoods of Proposition \ref{prop:complete-tails}.
\end{proof}

\begin{cor}
\label{cor:matching-convergence}
Let $d_j\to0$ and let $z_j\in\mathcal M_{d_j}$. After translation and passage to a subsequence, either $z_j$ or $\mathcal Rz_j$ converges to $v_{0,+}$ in the matching coordinates.
\end{cor}

\begin{proof}
By Lemma \ref{lem:bounded-regularity},
Corollary~\ref{cor:critical-band}, and Lemma \ref{lem:select-x0}, after translation and passage to a
subsequence, the local limit is a nonconstant $d=0$ connector from $X_0$ to $a$. By Theorem \ref{thm:exact-two-d0}, it is $v_{0,+}$ or $v_{0,-}=\mathcal{R}v_{0,+}$. 

Lemma~\ref{lem:uniform-tail-entry} and the complete tail parametrizations of Proposition~\ref{prop:complete-tails} then upgrade the local convergence to convergence in all matching coordinates.
\end{proof}

\begin{proof}[Proof of Theorem \ref{thm:exact-two-small-d}]
Theorem \ref{thm:at-least-two} gives two distinct classes
$[v_{d,+}]$ and $[v_{d,-}]$ for all sufficiently small $d>0$.

Suppose that additional classes existed for arbitrarily small $d$.
Then there are $d_j\to0^+$ and connectors $z_j$ such that
\[
 [z_j]\notin\{[v_{d_j,+}],[v_{d_j,-}]\}.
\]
By Corollary~\ref{cor:matching-convergence}, either $z_j$ or $\mathcal Rz_j$ converges
to $v_{0,+}$ in the matching coordinates. The local uniqueness in
the implicit-function-theorem construction then gives, respectively,
\[
 [z_j]=[v_{d_j,+}]
 \qquad\text{or}\qquad
 [\mathcal Rz_j]=[v_{d_j,+}].
\]
Since $v_{d_j,-}=\mathcal Rv_{d_j,+}$, then $[z_j]$ is one of the two known classes, contradiction. Thus
\[
 \mathcal M_d=\{[v_{d,+}],[v_{d,-}]\}
\]
for all sufficiently small $d>0$.

The uniform bound is Proposition \ref{prop:uniform-bound}, while the branchs $v_{d,\pm}$ depends $C^1$ on $d$ by Theorem \ref{thm:at-least-two}. Together with
Theorem \ref{thm:exact-two-d0}, this proves
\[
\#\mathcal M_d=2,\qquad 0\le d<d_0,
\]
after decreasing $d_0$ if necessary.
\end{proof}

\subsubsection{Orientations and the signed count}
\label{sec:signed-count}

We orient the isolated connectors by the determinant line of the linearized Fredholm operator. For a real Fredholm operator $D$, set
\[
\det D:=\Lambda^{\max}\ker D \otimes\bigl(\Lambda^{\max}\operatorname{coker}D\bigr)^*.
\]

Lemma \ref{lem:matching-equivalence} transfers the stabilized full-line
orientation to the finite-strip matching problem. We use the
exact-sequence conventions of \cite{ZingerDet}, with the factor
order fixed below.

\begin{defn}
\label{def:positive-reduced-orientation}
At the base connector $v_{0,+}$, orient the translation line and the upper soft line, respectively, by
\[
T=v_s,\qquad
\sigma_*^u=e^{-i\pi/4}\sin(2\pi t/\tau)N,
\]
and orient the $\lambda$-line by $\lambda>0$ in $\chi_0JT$.
Order the obstruction space by $(JZ,JT)$, and, after stabilization, order the remaining domain and range coordinates by
\begin{equation}\label{eq:order-basis}
(\beta,\lambda,\alpha_{\rm tr})\quad\text{and}\quad
(JZ,JT,\Lambda),    
\end{equation}
where $\alpha_{\rm tr}T$ is the translation direction.
The positive reduced orientation is the determinant orientation defined by these ordered choices.
\end{defn}

At $d=0$, operators with the fixed end limits and trivializations form the contractible affine family
\[
 D_{\mathrm{ref}}+\mathcal K,
\]
where $\mathcal K$ is the vector space of relatively compact
lower-order perturbations with zero end limits. Thus the determinant orientation extends uniquely over this family and, through the stabilization and matching identification above, to the finite-strip problem.

Set
\[
c_B:=2^{-1/2}\|\phi_B\|_2^2>0,
\qquad
c_T:=\int_{\mathbb R}\chi_0\|T\|_2^2\,\dd s>0.
\]
The upper soft direction and the forcing direction satisfy
\[
\langle D\zeta_\beta,JZ\rangle=c_B,
\qquad
\langle \chi_0JT,JT\rangle=c_T,
\]
while
\[
\langle D\zeta_\beta,JT\rangle=
\langle \chi_0JT,JZ\rangle=0,
\qquad\Lambda(T)=1.
\]
Thus, in the ordered coordinates \eqref{eq:order-basis}, the reduced matching operator is
\begin{equation}\label{eq:positive-orientation-matrix}
 \begin{pmatrix}
  c_B&0&0\\
  0&c_T&0\\
  0&0&1
 \end{pmatrix}.
\end{equation}
Its determinant is positive. So
\[
 \operatorname{sign}(v_{0,+})=+1
\]
in the positive reduced orientation.

Finally, we compare this orientation with that of the reflected connector and obtain the signed count.

\begin{thm}
\label{thm:signed-exact-two}
With the positive reduced orientation of
Definition~\ref{def:positive-reduced-orientation}, the reflection
$\mathcal R$ preserves the reduced orientation. Thus, for all sufficiently small $d>0$,
\[
\operatorname{sign}[v_{d,+}]=
\operatorname{sign}[v_{d,-}]
=+1,\qquad
\#_{\mathrm{sign}}\mathcal M_d=2.
\]
At $d=0$, the two limiting branches have the same sign after adjoining the soft-line orientation.
\end{thm}

\begin{proof}
At $d=0$, let $r$ denote the fixed-frame action of $\mathcal R$.
For the reference operator, $r$ acts by $-1$ on the one-dimensional
kernel and by $-1$ on each of the two cokernel directions. Hence
\[
\chi_{\det D}(r)=
\det(r|_{\ker D})\,
\det(r|_{\coker D})^{-1}
=(-1)\bigl((-1)^2\bigr)^{-1}
=-1.
\]
By continuity over the fixed-end Fredholm family, the same character
holds for the linearization along $v_{0,+}$.

On the soft coordinate, reflection is $w\mapsto-w$, so its character
is $-1$. The forcing and translation factors are preserved:
\[
T_{\mathcal Rv}=\mathcal RT_v,\qquad
J\mathcal R=\mathcal RJ,\qquad
\Lambda_{\mathcal Rv}(\mathcal R\xi)=\Lambda_v(\xi).
\]
Thus the stabilized determinant character is
\[
(-1)_{\det D}(-1)_{\mathrm{soft}}
(+1)_{\mathrm{forcing}}
(+1)_{\mathrm{translation}}
=+1.
\]

Lemma~\ref{lem:matching-equivalence} transfers this orientation to the
finite-strip matching problem. Since
$\operatorname{sign}(v_{0,+})=+1$ by \eqref{eq:positive-orientation-matrix},
reflection gives
\[
\operatorname{sign}(v_{0,-})=+1.
\]
The signs persist along the two small-$d$ branches, and the exact-two
theorem then gives
\[
\#_{\mathrm{sign}}\mathcal M_d=2.\qedhere
\]
\end{proof}

\begin{rmk}
Proposition~\ref{prop:H2-pointed-alien-stokes-constant} and
Theorem~\ref{thm:signed-exact-two} give
\[
S_{01}^{+}=\#_{\mathrm{sign}}\mathcal M_d=2,
\qquad 0<d\ll1.
\]
The Picard--Lefschetz picture predicts that the geometric
orientation convention used here is compatible with the relative
thimble orientations specified by the leading-term normalization
in Remark~\ref{rmk:alien-orientation}.
\end{rmk}

\begin{appendices}
\section{Mixed regularity and complete half-tail coordinates}
\label{appendix_A}

In this appendix, we give the explicit proof of Lemmas \ref{lem:mixed-regularity}, \ref{lem:lyapunov-perron} and Proposition \ref{prop:center-fibres}. 

\begin{proof}[Proof of Lemma \ref{lem:mixed-regularity}]
In a local tangent frame, the real equation has the principal part
\[\mathcal{L}:=\partial_s-C_{\theta,z}\partial_t^2,\qquad C_{\theta,z}\xi:=P_z(e^{i\theta}\eta\bar{\xi})\]
The lower-order terms are at most quadratic lower-order terms in $z_t$. Recall that $h$ is the restriction of the standard
Hermitian metric on $\mathbb C^3$. Since $P_z$ is the $h$-orthogonal projection, for tangent vectors
$\xi,\zeta$, we have
\[
\operatorname{Re}h_z(\mathsf C_{\theta,z}\xi,\zeta)
=\operatorname{Re}(e^{-i\theta}\xi^T\eta\zeta).
\]
Then the $C_{\theta,z}$ is real self-adjoint and invertible. On a fixed compact target set, its eigenvalues are uniformly bounded away from $0$. Then the frozen real matrix $C_0$ satisfies
\[|(i\omega+\xi^2C_0)^{-1}|\le C(|\omega|+\xi^2)^{-1},\qquad (\omega,\xi)\ne(0,0).\]

Each eigenvalue $b\ne0$ of the frozen matrix $C_0$
gives a scalar operator
\[
 \partial_s-b\partial_t^2.
\]
For $b>0$, this is the standard heat operator; for $b<0$ it becomes one after reversing $s$.
In both cases the Dirichlet boundary condition satisfies
the complementing condition.

The assumed $\mathcal E$ bound gives a uniform H\"older modulus for the principal coefficients.
The standard Dirichlet $L^p$ estimates, followed by localization and freezing of coefficients \cite{AV,DHP}, therefore give, on smaller $s$-strips,
\begin{equation}
\label{eq:z_sz_tt}
\|z_s\|_{L^p(I'\times(0,\tau))}+\|z_{tt}\|_{L^p(I'\times(0,\tau))}\le C\left(\|\mathcal Lz\|_{L^p(I\times(0,\tau))}+\|z\|_{L^p(I\times(0,\tau))}\right).
\end{equation}
The Sobolev embedding also gives $z_t\in L^p$ for $p<6$. Then the quadratic terms in $z_t$ is in $L^p$ for $p<3$. Iterating \eqref{eq:z_sz_tt} and differentiating the equation proves the $C^k$ estimates.

On an eigenmode $\lambda\ne0$, the inverse of $\partial_s+\lambda$ is the one-sided exponential convolution. Multiplication
by $e^{\pm a}$ replaces $\lambda$ by $\lambda\mp a$. Plancherel and
\[\sup_{\omega\in\R}\frac{\omega^2+\lambda^2}{|i\omega+\lambda\mp a|^2}\le C_a\]
give the Green estimate after summing the modes. The homogeneous terms satisfy
\[\|e^{-sA}\Pi^sh\|_{\mathcal E_a^+}+\|e^{-sA}\Pi^uh\|_{\mathcal E_a^-}\le C_a\|h\|_V.\qedhere\]
\end{proof}

\begin{proof}[Proof of Lemma \ref{lem:lyapunov-perron}]
Set
\[
J_s=\begin{cases}
[s-1,s+1],&s\in\R,\\
[s,s+1],&s\in\R_+,\\
[s-1,s],&s\in\R_-,
\end{cases}
\]
according to the domain of the trajectory. Choose a cut-off
$\chi\in C^\infty([0,\infty))$ with $\chi=1$ on $[0,1]$
and $\chi=0$ on $[4,\infty)$, and define
\begin{equation}\label{eq:localized-cutoff-definition}
 Q_u(s)=\|u\|_{\mathcal E(J_s)}^2,\qquad
 \mathcal N_{d,R}(u)(s)
 =\chi(Q_u(s)/R^2)\mathcal N_d(u(s)).
\end{equation}
For the full-line construction, set
\[
\mathbb E_\nu=\{u:e^{-\nu\langle s\rangle}u\in\mathcal E(\mathbb R)\},
\quad\mathbb F_\nu=L^2(\mathbb R,e^{-2\nu\langle s\rangle}ds;H),
\quad\langle s\rangle=(1+s^2)^{1/2}.
\]
When the cutoff is nonzero, we first have $Q_u(s)<4R^2$. Then the embedding $\mathcal{E}(J_s)\hookrightarrow C(J_s;V)$ and Cauchy-Schwarz give
\[\sup_{r\in J_s}\|u\|_V+\|u\|_{L_s^2(J_s;D_2)\le CR},\qquad |DQ_u(s)[v]|\le CR\|v\|_{\mathcal{E}(J_s)}.\]
In particular, $u(s)$ lies in the domain of $\mathcal{N}$. Differentiating $\mathcal{N}_R$ gives
\[D\mathcal{N}_R(u)v=\chi(Q_u/R^2)D\mathcal{N}(u)v+R^{-2}\chi'(Q_u/R^2)DQ_u[v]\mathcal{N}(u).\]
Since the exponential weights are comparable on each unit interval, summing the local estimates and using bounded overlap gives\footnote{For the family $\mathcal{N}_d$ of \eqref{eq:local}, the factor is $C_a(|d|+R)$. In the full-line growth spaces of Proposition \ref{prop:center-fibres}, the same block argument gives, for $1\le m\le 3$ and $\eta>\sum_{j=1}^m\eta_j$,
\[\|D^m\mathcal{N}_{d,R}(u)[v_1,\ldots,v_m]\|_{\mathbb{F}_{\eta}}\le C\prod_{j=1}^m\|v_j\|_{\mathbb{E}_{\eta_j}}.\]
This is the localized Sobolev cutoff used below.}
\begin{equation}
\label{eq:NRu-v}
\|\mathcal{N}_R(u)-\mathcal{N}_R(v)\|_{\mathcal{F}_a^{\pm}}\le C_aR\|u-v\|_{\mathcal{E}_a^{\pm}}.
\end{equation}

By Lemma \ref{lem:mixed-regularity} and \eqref{eq:NRu-v}, the Lyapunov-Perron map
\[u\longmapsto e^{-sA}h+\mathcal{G}^+\mathcal{N}_R(u)\]
is a contraction in both $\mathcal E_\mu^+$ and
$\mathcal E_{-\omega}^+$, for $0<\omega<\mu_*$. Its fixed point $u_h$ defines
the stable graph $\Psi(h)=\Pi^u u_h(0)$.

To prove the completeness, let $u$ be any forward half-solution with $\sup_{s\ge0}\|u(s)\|_V\le\eps$. The gradient identity and the graph-norm estimate give
\[\int_{\R_+}\|u_s(s)\|_{H^2}\,\dd s\le C\eps^2,\qquad \sup_{J\subset\R_+,|J|=1}\|u\|_{\mathcal{E}(J)}\le C\eps.\]
The unit-strip bound implies $u\in\mathcal E_{-\omega}^+$,
and the cutoff equals one for $\varepsilon\ll R$. Thus $u$ satisfies the same Lyapunov-Perron equation as $u_h$, with $h=\Pi^s u(0)$. Uniqueness gives $u=u_h$.

Repeating the argument on a sufficiently late tail gives
exponential decay for every $\mu<\mu_*$. Differentiating the fixed-point equation gives the stated $C^1$ dependence, uniformly under a fixed spectral gap.
\end{proof}

Next we prove Proposition \ref{prop:center-fibres} on the centre graph and complete strong fibers.
\begin{proof}[Proof of Proposition \ref{prop:center-fibres}]
Choose $\nu,\kappa>0$ with $4\nu<\kappa$ and $\kappa+4\nu<\mu_*$. The hard Green operator and centre integral are given by
\begin{align*}
(\mathcal G^hf)(s)
&=\int_{-\infty}^s e^{-(s-r)A_0}\Pi^sf(r)\,\dd r
-\int_s^\infty e^{-(s-r)A_0}\Pi^uf(r)\,\dd r,\\
 (\mathcal I^cf)(s)&=\int_0^s\Pi^cf(r)\,\dd r.    
\end{align*}

On the full-line growth spaces $\mathbb E_\nu$ and $\mathbb F_\nu$ of last proof, the fixed hard gap and the weighted integral estimate give
\[
\|\mathcal G^hf\|_{\mathbb E_\nu}
+\|\mathcal I^cf\|_{\mathbb E_\nu}
\le C_\nu\|f\|_{\mathbb F_\nu}.
\]
Using the above cutoff $\chi$ with $\mathcal N_d$ in \eqref{eq:local},
\begin{equation}
\label{eq:UcNdR}
U=c+(\mathcal{G}^h+\mathcal{I}^c)\mathcal{N}_{d,R}(U)
\end{equation}
is a contraction for $|d|+R$ small.  Its fixed point defines
$\Gamma_d(c)=U_{d,c}(0)$ and a locally invariant centre graph.
The estimates at weights $\nu,2\nu,3\nu,4\nu$ give $C^3(V)$ dependence. The nearby stationary Bott loops satisfy the same fixed-point equation at $d=0$, and centre projection identifies them with
$\Gamma_0$.   

Furthermore, we study the $D_2$ regularity. Integration by parts against the centre sine modes shows that the reduced field $g_d(c):=\Pi^c\mathcal{N}_d(\Gamma_d(c))$ is $C^2$. Then the invariant equation and the invertibility of its principal coefficient give
\[
\Gamma_d\in C^2(W_t^{2,1})
\subset C^2(W_t^{1,\infty}).
\]
Substituting this improved regularity back into the equation yields $\Gamma_d\in C^2(D_2)$.

For a centre orbit $U=U_{d,c}$, set 
\[f_U(v)=\mathcal{N}_{d,R}(U+v)-\mathcal{N}_{d,R}(U).\] 
The fibre equation is
\begin{equation}
\label{eq:fiber_v}
\begin{aligned}
v(s)=\,&e^{-sA_0}h
+\int_{-\infty}^s e^{-(s-r)A_0}\Pi^sf_U(v)(r)\,\dd r\\
&-\int_s^0e^{-(s-r)A_0}\Pi^uf_U(v)(r)\,\dd r+\int_{-\infty}^s\Pi^cf_U(v)(r)\,\dd r,\qquad s\le0. 
\end{aligned}
\end{equation}
Its Lipschitz constant in $\mathcal{E}_{\kappa}^-$ is $C(|d|+R)$, giving $v_{d}(c,h)$. Differentiating at the stronger decay weights below $\kappa+2\nu$ gives $C^1$ dependence at weight $\kappa$.

Let $P=(\Pi^c,\Pi^u)$. Define the endpoint map $F_d(c,h)=\Gamma_d(c)+v_{d}(c,h)(0)$, differentiating the fixed-point equations gives
\[
D_{(c,h)}(PF_d)=\operatorname{Id}+O(|d|+R).
\]
The inverse function theorem gives a uniformly bounded local inverse for $PF_d$, and hence for $F_d$ on its image.

By the unit-strip estimate of last proof, every sufficiently
small negative half-solution belongs to $\mathcal E_{-\omega}^-$, with cutoff equal to one. Its projected endpoint values determine $(c,h)$ uniquely.
The same Green estimates and contraction argument give
uniqueness in $\mathcal E_{-\omega}^-$ for these data.
Thus the half-solution equals $U_{d,c}+v_d(c,h)$, proving completeness.
\end{proof}
\end{appendices}

\section*{Statement on AI Assistance}

The research presented in this paper began in May 2024, and a first complete version of the manuscript was finished in June 2026. In August 2026, while revisiting the proof of the exact-two count, ChatGPT 5.6 pointed out a gap in the argument. With the assistance of ChatGPT 5.6 and GPT-6, we re-examined the relevant analytic estimates, obtained the necessary uniform bounds, and completed a revised proof of the exact-two count. The authors are responsible for the verification and final form of all mathematical arguments in the paper.


\begin{thebibliography}{99}

\bibitem{AV} M.~S. Agranovich and M.~I. Vishik,
\textit{Elliptic problems with a parameter and parabolic problems of general type},
Russian Math. Surveys \textbf{19} (1964), no.~3, 53--157.


\bibitem{AG} D. N. Akhiezer and S. G. Gindikin.
\textit{On Stein extensions of real symmetric spaces}.
Math. Ann. 286 (1990), No. 1-3, 1-12.

\bibitem{Az} H. Azad.
\textit{Levi-curvature of manifolds with a Stein rational fibration}.
Manuscr. Math. 50 (1985), 269-311.

\bibitem{Balser}
W. Balser.
\textit{From Divergent Power Series to Analytic Functions:
Theory and Application of Multisummable Power Series}.
Lecture Notes in Mathematics, Vol. 1582, Springer-Verlag, Berlin--Heidelberg, 1994.

\bibitem{BiswasDumitrescuHRM} I. Biswas and S. Dumitrescu.
\textit{Holomorphic Riemannian metric and the fundamental group},
Bulletin de la Société Mathématique de France \textbf{147} (2019), no.3, 455--468.

\bibitem{BleisteinHandelsman} N. Bleistein and R. A. Handelsman,
\textit{Asymptotic Expansions of Integrals}.
Dover Publications, 1986.

\bibitem{Costin}
O. Costin.
\textit{Asymptotics and Borel Summability}.
Monographs and Surveys in Pure and Applied Mathematics, Vol. 141,
Chapman \& Hall/CRC, Boca Raton, FL, 2008.

\bibitem{CPT} O. Costin, H. Park and Y. Takei.
\textit{Borel summability of the heat equation with variable coefficients}.
J. Differential Equations 252 (2012), no. 4, 3076--3092.


\bibitem{DHP} R. Denk, M. Hieber, and J. Pr\"uss,
\textit{$\mathcal R$-boundedness, Fourier multipliers and problems of
elliptic and parabolic type},
Mem. Amer. Math. Soc. \textbf{166} (2003), no.~788, viii+114.


\bibitem{Du} G. V. Dunne.
\textit{Borel Summation and Analytic Continuation of the Heat Kernel on Hyperbolic Space}.
In \textit{Peter Suranyi Festschrift: A Life in Quantum Field Theory}, pp. 167--189, World Scientific, 2022.

\bibitem{Ecalle1981I}
J.~Écalle,
\textit{Les fonctions résurgentes. Tome I: Les algèbres de fonctions résurgentes},
Publications Mathématiques d'Orsay, 81-05,
Université de Paris-Sud, Département de Mathématique, Orsay, 1981.

\bibitem{Ecalle1981II}
J.~Écalle,
\textit{Les fonctions résurgentes. Tome II: Les fonctions résurgentes appliquées à l'itération},
Publications Mathématiques d'Orsay, 81-06,
Université de Paris-Sud, Département de Mathématique, Orsay, 1981.

\bibitem{Ecalle1985III}
J.~Écalle,
\textit{Les fonctions résurgentes. Tome III: L'équation du pont et la classification analytique des objets locaux},
Publications Mathématiques d'Orsay, 85-05,
Université de Paris-Sud, Département de Mathématique, Orsay, 1985.

\bibitem{FSCorr}
G.~Faye and A.~Scheel, \textit{Corrigendum to ``Center Manifolds without
a Phase Space''}, arXiv:2007.14260 (2020).

\bibitem{Gr} H. Grauert.
\textit{On Levi's problem and the imbedding of real-analytic manifolds}.
Ann. Math. 68 (1958), 460-472.

\bibitem{GS} V. Guillemin and M. Stenzel.
\textit{Grauert tubes and the homogeneous Monge-Amp\'ere equation}.
J. Differential Geom. 34 (1991), no.2, 561-570.

\bibitem{Har} T. Harg\'e.
\textit{Borel summation of the small time expansion of the heat kernel. The scalar potential case.}.
arXiv:1301.7742, 2013.

\bibitem{Har2} T. Harg\'e.
\textit{Borel summation of the small time expansion of the heat kernel with a vector potential}.
arXiv:1302.0604, 2013.

\bibitem{HilgertNeeb2012}
J. Hilgert and K.-H. Neeb.
\textit{Structure and Geometry of Lie Groups}.
Springer Monographs in Mathematics. Springer, New York, 2012.

\bibitem{HamiltonBoundary}
R. S. Hamilton.
\textit{Harmonic Maps of Manifolds with Boundary}.
Lecture Notes in Mathematics, Vol. 471, Springer-Verlag, 1975.

\bibitem{Hochschild1965}
G. Hochschild.
\textit{The Structure of Lie Groups}.
Holden-Day, San Francisco, 1965.

\bibitem{Hochschild1966}
G. Hochschild.
\textit{Complexification of real analytic groups}.
Trans. Amer. Math. Soc. 125 (1966), no. 3, 406--413.

\bibitem{HS} A. Huckleberry and D. Snow.
\textit{A classification of strictly pseudoconcave homogeneous manifold}.
Ann. Scuola Norm. Sup. Pisa Cl. Sci. (4) 8 (1981), no. 2, 231--255.

\bibitem{Kont} M. Kontsevich.
\textit{Private discussion}, YMSC \& BIMSA, 2023.

\bibitem{Kont2020} M. Kontsevich.
\textit{Talk Slides: Exponential integrals, Lefschetz thimbles and linear resurgence, 30 June, 2020}.

\bibitem{KS-2022}
M.~Kontsevich, and Y.~Soibelman. \textit{Analyticity and resurgence in wall-crossing formulas.} Lett. Math. Phys. 112 (2022), no. 2, Paper No. 32.

\bibitem{KS-2024}
M.~Kontsevich, and Y.~Soibelman. \textit{Holomorphic Floer theory I: exponential integrals in finite and infinite dimensions.} arXiv preprint arXiv:2402.07343 [math.SG], 2024.

\bibitem{Ku} R. S. Kulkarni.
\textit{On complexifications of differentiable manifolds}.
Invent. Math. 44 (1978), 49--64.

\bibitem{LS} L. Lempert and R. Sz\"oke.
\textit{Global solutions of the homogeneous complex Monge-Amp\'ere equation and complex structures on the tangent bundle of Riemannian manifolds}.
Math. Ann. 290 (1991), no.4, 689--712.

\bibitem{LiLiTang2026}
S.~Li, Y.~Li, and X.~Tang,
\textit{Picard--Lefschetz Theory and Alien Calculus: A Case Study}, arXiv:2605.08867v2 [math-ph], 2026.

\bibitem{LM1960} J.-L. Lions and B. Malgrange.
\textit{Sur l'unicit{\'e} r{\'e}trograde dans les probl{\`e}mes mixtes paraboliques}.
Math. Scand. 8 (1960), 277--286.

\bibitem{LMS} D. A. Lutz, M. Miyake and R. Sch\"afke.
\textit{On the Borel summability of divergent solutions of the heat equation}.
Nagoya Math. J. 154 (1999), 1--29.

\bibitem{Mielke}
A.~Mielke, \textit{Reduction of quasilinear elliptic equations in
cylindrical domains with applications}, Math. Methods Appl. Sci.
\textbf{10} (1988), 51--66, doi:10.1002/mma.1670100105.

\bibitem{MielkeNH}
A.~Mielke, \textit{Normal hyperbolicity of center manifolds and
Saint-Venant's principle}, Arch. Rational Mech. Anal. \textbf{110}
(1990), 353--372, doi:10.1007/BF00393272.

\bibitem{MitschiSauzin2016}
C.~Mitschi and D.~Sauzin,
\textit{Divergent Series, Summability and Resurgence I:
Monodromy and Resurgence},
Lecture Notes in Mathematics, Vol.~2153, Springer, Cham, 2016.

\bibitem{MN} A. Morimoto and T. Nagano.
\textit{On pseudo-conformal transformations of hypersurfaces}.
J. Math. Soc. Japan 15 (1963), no.3, 289-300.

\bibitem{Neeb2000}
K.-H. Neeb.
\textit{Holomorphy and Convexity in Lie Theory}.
de Gruyter Expositions in Mathematics, Vol. 28. 
Walter de Gruyter, Berlin, 2000.

\bibitem{PW} G. Patrizio and P-M. Wong.
\textit{Stein manifolds with compact symmetric center}.
Math. Ann. 289 (1991), no. 3, 355--382.

\bibitem{PrussSimonett}
J.~Pr\"uss and G.~Simonett,
\textit{Moving Interfaces and Quasilinear Parabolic Evolution Equations},
Monographs in Mathematics, Vol.~105, Birkh\"auser, Cham, 2016.

\bibitem{PSS1997}
D. Peterhof, B. Sandstede and A. Scheel.
\textit{Exponential dichotomies for solitary-wave solutions of semilinear elliptic equations on infinite cylinders}.
J. Differential Equations 140 (1997), no. 2, 266--308.

\bibitem{PhamSaddlepoint}
F. Pham,
\textit{Vanishing homologies and the $n$ variable saddlepoint method}. 
In Singularities, Part 2 (Arcata, Calif., 1981), Proc. Sympos. Pure Math., Vol. 40, Amer. Math. Soc., 1983, pp. 319--333.

\bibitem{SauzinSplitting}
D. Sauzin.
\textit{Resurgent functions and splitting problems}.
RIMS K\^oky\^uroku \textbf{1493} (2006), 48--117.

\bibitem{Sz} R. Sz\"oke.
\textit{Complex structures on tangent bundles of Riemannian manifolds}.
Math. Ann. 291 (1991), no. 3, 409--428.

\bibitem{BW} H. Whitney and F. Bruhat.
\textit{Quelques propri\'et\'es fondamentales des ensembles analytiques-r\'eels}.
Comment. Math. Helv. 33 (1959) 132-160.

\bibitem{Witten2010}
E. Witten.
\textit{A new look at the path integral of quantum mechanics}.
In \textit{Surveys in Differential Geometry. Volume XV. Perspectives in Mathematics and Physics}, Int. Press, Somerville, MA, 2011, 345--419.

\bibitem{Wu-thesis}
Y.~Wu, \textit{Resurgence theory on the heat
kernel of Riemann surfaces}. Undegraduate thesis (2024), Qiuzhen college, Tsinghua University.

\bibitem{ZingerDet}
A. Zinger.
\textit{The determinant line bundle for Fredholm operators:
construction, properties, and classification}.
Math. Scand. \textbf{118} (2016), no. 2, 203--268.
\href{https://doi.org/10.7146/math.scand.a-23687}{doi:10.7146/math.scand.a-23687}.

\end{thebibliography}
\end{document}